%% file: main.tex
\RequirePackage{fix-cm}

\documentclass[smallextended]{svjour3}       % onecolumn (second format)

\usepackage[numbers]{natbib}

\usepackage{amssymb}
\usepackage{amsmath}

\usepackage{pgfplots} % Used for coloring text

\usepackage{subcaption}
\usepackage{bm}

\usepackage{algorithm}
\usepackage{algorithmicx}
\usepackage[noend]{algpseudocode}

\newcommand{\basingstoke}{\texttt{DO-A}}
\newcommand{\colchester}{\texttt{DO-B}}
\newcommand{\northampton}{\texttt{DO-C}}

\allowdisplaybreaks

\journalname{Journal of Scheduling}

\newcommand{\change}[1]{{#1}}

\begin{document}

% \title{Approximating Bounded Job Start Scheduling with Application in Royal Mail Deliveries under Uncertainty\footnote{A preliminary version has been accepted in COCOA 2019.}
% }
\title{
%Bounded Job Start Scheduling in Royal Mail Delivery Offices:
%Approximation Algorithms and a Lexicographic Robust Optimization Approach
Approximate and Robust Bounded Job Start Scheduling for Royal Mail Delivery Offices\footnote{A preliminary version has been accepted in COCOA 2019.}
}
\titlerunning{Approximating Bounded Job Start Scheduling}

\author{
Dimitrios Letsios \and
Jeremy T.\ Bradley \and
Suraj G \and
Ruth Misener \and
Natasha Page 
}


\institute{
Dimitrios Letsios \at 
King's College London \\
\email{dimitrios.letsios@kcl.ac.uk}
\and
Jeremy T. Bradley \at
GBI/Data Science Group, Royal Mail \\
\email{jeremy.bradley@royalmail.com}
\and
Suraj G, Natasha Page, and Ruth Misener \at
Imperial College London \\
\email{\{suraj.g15,r.misener,natasha.page17\}@imperial.ac.uk}
}

\date{Received: date / Accepted: date}

\maketitle

\begin{abstract}
% Motivated by mail delivery scheduling problems arising in Royal Mail, we 
% generalize the fundamental makespan scheduling $P||C_{\max}$ problem to develop the \emph{bounded job start scheduling problem}. 
%that generalizes the fundamental makespan scheduling and bin packing problems.
%This problem assumes a set $\mathcal{J}$ of jobs, each specified by an integer processing time $p_j$, that have to be executed non-preemptively by a set of $m$ parallel identical machines and an integer $g$ specifying an upper bound on the number of jobs that may simultaneously begin per unit of time.
%The goal is to compute a minimum makespan schedule satisfying the bound on simultaneous job starts.
%{\color{red} 
Motivated by mail delivery scheduling problems arising in Royal Mail, we study a generalization of the fundamental makespan scheduling $P||C_{\max}$ problem which we call the \emph{bounded job start scheduling problem}.
Given a set of jobs, each specified by an integer processing time $p_j$, that have to be executed non-preemptively by a set of $m$ parallel identical machines, the objective is to compute a minimum makespan schedule subject to an upper bound $g\leq m$ on the number of jobs that may simultaneously begin per unit of time. 
\change{With perfect input knowledge,} we show that Longest Processing Time First (LPT) algorithm is tightly 2-approximate.
After proving that the problem is strongly $\mathcal{NP}$-hard even when $g=1$, we elaborate on improving the 2-approximation ratio for this case.
We distinguish the classes of long and short instances satisfying $p_j\geq m$ and $p_j<m$, respectively, for each job $j$.
We show that LPT is 5/3-approximate for the former and optimal for the latter.
Then, we explore the idea of scheduling  long jobs in parallel with short jobs to obtain tightly satisfied packing and bounded job start constraints.
For a broad family of instances excluding degenerate instances with many very long jobs, we derive a 1.985-approximation ratio.
For general instances, we require machine augmentation to obtain better than 2-approximate schedules. %}
%Finally, 
\change{In the presence of uncertain job processing times,} we exploit machine augmentation and lexicographic optimization, which is useful for $P||C_{\max}$ under uncertainty, to propose a two-stage robust optimization approach for bounded job start scheduling under uncertainty aiming in 
%good trade-offs in terms of makespan and number of used machines.
\change{a low number of used machines. Given a collection of schedules of makespan $\leq D$, this approach allows distinguishing which are the more robust.We substantiate both the heuristics and our recovery approach numerically using Royal Mail data. We show that, for the Royal Mail application, machine augmentation, i.e.\ short-term van rental, is especially relevant.}
\end{abstract}

\section{Introduction}

%\paragraph{General Setting}

Royal Mail provides mail collection and delivery services for all United Kingdom (UK) addresses. 
With a small van fleet (as of January 2019) of 37,000 vehicles and 90,000 drivers delivering to 27 million locations in UK, efficient resource allocation is essential to guarantee the business viability.
The backbone of the Royal Mail distribution %network 
is a three-layer hierarchical network with 6 regional distribution centers serving 38 mail centers. 
Each mail center receives, processes, and distributes mail for a large geographically-defined area via 1,250 delivery offices, each serving disjoint sets of neighboring post codes. 
Mail is collected in mail centers, sorted by region, and forwarded to an appropriate onward mail center, making use of the regional distribution centers for cross-docking purposes.
From the onward mail center it is transferred to the final delivery office destination.
This process has to be completed within 12 to 16 hours for 1st class post and 24 to 36 hours for 2nd class post depending on when the initial collection takes place.
%\todo{Are these Jeremy-approved numbers?}

% \paragraph{Scheduling Problem Motivation}

In a delivery office, post is sorted, divided into routes, and delivered to addresses using a combination of small fleet vans and walked trolleys. 
%Then, postal workers drive vans stationed at the delivery office to their daily routes and deliver the mail to recipients.
%The delivery routes are statically optimized and fixed a priori.
Allocating delivery itineraries to vans is critical.
%Each delivery office has an van exit gate which imposes an upper bound on the number of vehicles that may depart per unit of time.
%{\color{red} 
Each delivery office has a van exit gate which gives an upper bound the number of vehicles that can depart per unit of time.
Thus, we deal with scheduling a set $\mathcal{J}$ of jobs (delivery itineraries), each associated with an integer processing time $p_j$, on $m$ parallel identical machines (vehicles), s.t.\ the makespan, i.e.\ the last job completion time, is minimized.
%}
%Thus, we are faced with scheduling a set $\mathcal{J}$ of jobs (delivery itineraries) where each job is associated with an integer processing time $p_j$, on $m$ parallel identical machines (vehicles), where $m$ is non-constant, to minimize the makespan, i.e.\ the time when the last job completes.
Parameter $g$ imposes an upper bound on the number of jobs that may simultaneously begin per unit of time.
Each job has to be executed non-preemptively, i.e.\ by a single machine in a continuous time interval without interruptions.
We refer to this problem as the \emph{Bounded Job Start Scheduling Problem (BJSP)}.

\change{
Our contribution is twofold:
First, we derive greedy constant-factor approximation algorithms, i.e.\ simple heuristics adoptable by Royal Mail practitioners, for effectively solving BJSP instances with perfect knowledge.
% That is, we provide simple heuristic policies 
% %(with analytically proven performance guarantees)
% %for Royal Mail delivery scheduling, 
% which can be effectively adopted in practice by Royal Mail practitioners.
Second, we propose a two-stage robust optimization approach, based on Royal Mail practices, which evaluates the robustness of BJSP schedules under uncertainty.
% because Royal Mail delivery scheduling is often subject to uncertainty, we take a first step in this direction and develop a robust scheduling approach mitigating the effect of imprecise processing times. 
Using real data, we computationally validate the performance of both the heuristics and two-stage robust optimization approach.
%Next, 
Section~\ref{Section:Relation_Makespan} discusses the relationship between BJSP and the fundamental makespan scheduling problem, a.k.a.\ $P||C_{\max}$. Section~\ref{Section:Related_Work} presents relevant literature. Section~\ref{Section:Contributions} summarizes the paper's organization and our contributions.

% In a typical Royal Mail delivery office, departures are planned in the beginning of the day, but may concide due to various reasons.
% In these case, the order at which the vehicles depart occurs empirically by the drivers.
% To deal with this issue, we provide them good heuristic policies.
% We also analyze approximation algorithms with mathematical rigor.
% Often, there is uncertainty in the duration of the deliveries.
% We make the first steps by provide invistigating robust optimization approaches.
% To this end, we develop approximation algorithms and a robust optimization approach.
}

\subsection{Relation to $P||C_{\max}$}
\label{Section:Relation_Makespan}

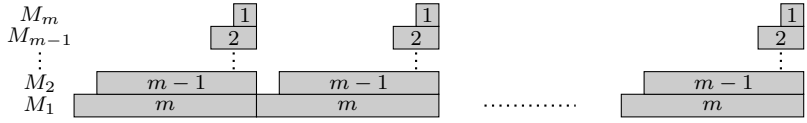
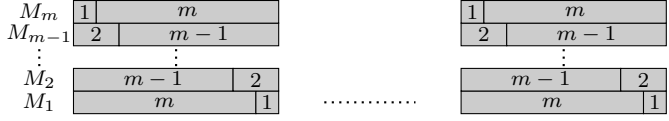
\begin{figure}[t]
\begin{subfigure}[t]{\textwidth}
\begin{center}
\input{opt_bjs.tex}
\end{center}
\caption{Optimal BJSP schedule of makespan $T_B=km$.}
\end{subfigure}
\begin{subfigure}[t]{\textwidth}
\begin{center}
\input{opt_pcmax.tex}
\end{center}
\caption{Optimal $P||C_{\max}$ schedule, without the BJSP constraint, of makespan $T_M=k(m+1)/2$.}
\end{subfigure}
\caption{\change{
BJSP instance with $m$ machines, $k$ jobs of processing time $i$, for each $i=1,\ldots,m$, and $g=1$. For an odd number $k=o(m)$, $T_B\simeq 2T_M$.}}
\label{Figure:PCmax}
\end{figure}

\change{ With perfect knowledge, BJSP is strongly related to 
%the fundamental makespan scheduling problem, a.k.a.\ $P||C_{\max}$ \cite{Graham1969}.
$P||C_{\max}$ which is defined similarly, but drops the BJSP constraint \cite{Graham1969}.}
Broadly speaking, BJSP is harder as a generalization of $P||C_{\max}$ and the two problems become equivalent when $g=m$. 
%Further, $P||C_{\max}$ is the BJSP relaxation obtained by dropping the BJSP constraint.
\change{
%The main additional difficulty in designing efficient algorithms for BJSP is bounding the idle machine time which occurs due to the bounded job start constraint.
$P||C_{\max}$ is strongly $\mathcal{NP}$-hard as a straightforward generation of the \emph{3-Partition} problem \cite{Garey1979}.
The well-known \emph{List Scheduling (LS)} and \emph{Longest Processing Time First (LPT)} algorithms achieve tight approximation ratios for $P||C_{\max}$ equal to 2 and 4/3, respectively.
Further, $P||C_{\max}$ admits a polynomial-time approximation scheme (PTAS).
}

\change{
%\paragraph{Technical Challenge}
Technically, good solutions for both BJSP and $P||C_{\max}$ must attain low imbalance $\max_i\{T-T_i\}$, where $T$ and $T_i$ are the makespan and completion time of machine $i$, respectively.
However, BJSP exhibits the additional difficulty of managing and bounding idle machine time
%, because of the BJSP constraint, 
during the time interval $[0,\min_i\{T_i\}]$.
To this end, we develop an algorithm that schedules long jobs in parallel with short jobs and bounds idle time with a concave relaxation. 
%analysis attempt to deal with this issue.
Figure~\ref{Figure:PCmax} shows a job set where the minimum makespan schedules with ($T_B$) and without ($T_M$) the bounded job start constraint differ by a factor of 2.
%feasible schedules may require idle machine time before all jobs have begun.
%Therefore, BJSP exhibits the additional difficulty of effectively bounding the total idle period $\sum_{t\leq r}(m-|\mathcal{A}_t|)$, where $r$ and $\mathcal{A}_t$ are the last job start time and set of jobs executed during $[t,t+1)$, respectively.

Approximate $P||C_{\max}$ solutions can be converted into feasible solutions for BJSP.
On the negative side, $P||C_{\max}$ optimal solutions are a factor $\Omega(m)$ from the BJSP optimum, in the worst case.}
To see this, take an arbitrary $P||C_{\max}$ instance and construct a BJSP one with $g=1$, by adding a large number of unit jobs.
The BJSP optimal schedule requires time intervals during which $m-1$ machines are idle at each time, while the $P||C_{\max}$ optimal schedule is perfectly balanced and all machines are busy until the last job completes.
On the positive side, we may easily convert any $\rho$-approximation algorithm for $P||C_{\max}$ into a $2\rho$-approximation algorithm for BJSP using na\"ive bounds. 
Given that $P||C_{\max}$ admits a PTAS, %polynomial-time approximation scheme, 
we obtain an $O(n^{1/\epsilon}\cdot poly(n))$-time $(2+\epsilon)$-approximation algorithm for BJSP.
Our main goal is to obtain tighter approximation ratios.

\subsection{Related Work}
\label{Section:Related_Work}

\change{
Next, we present related work to designing BJSP approximation algorithms and robust optimization approaches, relevant to Royal Mail delivery offices.}

% {\color{red} 

% We show that LPT achieves low idle machine time after the last job begins, while SPT achieves low idle time before the last job begins. Determining the best trade-off between the two so that the overall makespan of the resulting schedule is minimized is a challenge. In addition, the length of the jobs affects the quality of the resulting schedule in BJSP. Efficient algorithms for BJSP require carefully managing the idle machine machine time and developing appropriate bounds for analysis. The proposed LSM algorithm, which reduces the overall idle time by scheduling long jobs in parallel with long jobs, and our concave relaxation, for bounding the idle time before the last job begins (in the case of LPT for long instances), are key contributions in this direction.
% }

%{\color{red} \emph{Needs work: Discussion of related problems with job starting constrains and other relevant applications.}}
%BJSP is related to several problems investigated in the literature.
%BJSP is related to various scheduling problems investigated in the literature. 
\paragraph{Approximation Algorithms.}
BJSP relaxes the scheduling problem with forbidden sets, i.e.\ non-overlapping constraints, where subsets of jobs cannot run in parallel \cite{Schaffter1997}.
For the latter problem, better than 2-approximation algorithms are ruled out, unless $\mathcal{P}=\mathcal{NP}$ \cite{Schaffter1997}.
Even when there is a strict order between jobs in the same forbidden set, the scheduling with forbidden sets problem is equivalent to the precedence-constrained scheduling problem $P|prec|C_{\max}$ and cannot be approximated by a factor lower than $(2-\epsilon)$, assuming a variant of the unique games conjecture \cite{Svensson2011}.
% The scheduling problem with forbidden sets in which there is also a strict order between incompatible jobs in the same forbidden set is the well-known precedence-constrained scheduling problem and cannot be approximated by a factor $(2-\epsilon)$ assuming a variant of the unique games conjecture \cite{Svensson2011}. 
%This last result matches the performance of the well-known list scheduling algorithm by \citet{Graham1969}. 
Also, BJSP relaxes the scheduling with forbidden job start times problem, where no job may begin at certain time points, which does not admit constant-factor approximation algorithms \cite{Billaut2009,Gabay2016, Mnich2018,Rapine2013}.
%Scheduling with forbidden start times does not admit a constant-factor approximation algorithm even in the single machine case.
Despite the commonalities with the aforementioned literature, to the authors' knowledge, there is a lack of approximation algorithms for scheduling problems with bounded job starts.

\paragraph{Robust Optimization.}
\change{Royal Mail delivery times may be imprecise. % can be predicted from historical data?
Once a delivery has begun, it might finish earlier or later than its anticipated completion time.
Because of uncertain job completion times, Royal Mail vans attempting pre-computed schedules may not be able to complete all deliveries during working hours.
To this end, robust optimization provides a useful framework for structuring uncertainty, e.g.\ box uncertainty sets, and incorporating it in the decision-making process \cite{BenTal2009,Bertsimas2011,Goerigk2016,Kouvelis2013}.
Typically, Royal Mail delivery schedules are computed in two-stages: (i) the first stage computes a feasible, efficient schedule for an initial nominal problem instance before the day begins, and (ii) the second stage recovers the initial schedule by accounting for uncertainty during the day.
This setting is naturally captured by two-stage robust optimization with recovery \cite{BenTal2004,Bertsimas2010,Hanasusanto2015,Liebchen2009}.
%Prior literature investigates two-stage robust counterparts of discrete optimization problems, including timetabling and vehicle routing \cite{Chassein2016,Liebchen2009}. 

A common way of measuring solution robustness of a given discrete optimization problem instance is by comparing the final solution objective value after uncertainty realization with the solution objective value that we could have achieved if we had a crystal ball that accurately predicts the future.
From this perspective, we recently proposed a two-stage robust scheduling approach for $P||C_{\max}$, based on lexicographic optimization \cite{Letsios2018}.
If processing times are perturbed by a $(1+\epsilon)$ factor, the lexicographic optimization approach yields schedules a factor $(2+O(\epsilon))$ far from achievable optima with perfect knowledge.
But this common way of meansuring solution robustness is less applicable for our application because the makespan, i.e.\ the number of working hours in the day, is fixed.
So, in the Royal Mail context, certain decisions can be irrevocable during the recovery process.
For Royal Mail delivery offices, job start times can be irrevocable during the day, to reduce delivery delays and overtimes.
Further, resource augmentation (or constraint violation), e.g.\ backup machines, might be essential to ensure the resulting solutions' feasibility \cite{BenTal2000,kalyanasundaram2000speed}.
For this application, we measure a solution's robustness with the level of resource augmentation, e.g.\ number of short-term rental vans, required after uncertainty realization.
Robust bounded job start scheduling with resource augmentation is an intriguing open question.
}

% {\color{red}

% \paragraph{Robust Optimization Approaches.}
% \begin{itemize}
%     \item Two stages.
%     \item Motivate robust optimization (uncertainty set). provides a nice framework for solving optimization problems under uncertainty.
% \end{itemize}

% We consider robust optimization, i.e.\ hedge against worst-case realizations of imprecise parameter values.
% For $P||C_{\max}$ under uncertainty, \citet{Letsios2018} design a robust optimization approach. 

% A two-stage robust optimization setting under uncertainty requires an uncertainty set and a performance guarantee.
% A two-stage robust optimization approach requires an exact method for producing the initial solution and a recovery strategy.

% }

\subsection{Paper Organization and Contributions}
\label{Section:Contributions}

%Section~\ref{Section:Problem_Definition} defines BJSP and formulates it as an integer program. %, and introduces relevant notation.
%{\color{red} 
Section~\ref{Section:Problem_Definition} formally defines BJSP, proves the problem's $\mathcal{NP}$-hardness, and derives an $O(\log n)$ integrality gap for a natural integer programming formulation. %}
Section~\ref{Section:LPT} investigates \emph{Longest Processing Time First (LPT)} algorithm
%, i.e.\ the most natural BJSP option, 
and derives a tight $2$-approximation ratio.
We thereafter explore improving this ratio for the special case $g=1$.
Section~\ref{Section:Problem_Definition} shows that BJSP is strongly $\mathcal{NP}$-hard even when $g = 1$.
Several of our arguments can be extended to arbitrary $g$, but
%Furthermore, we conjecture all our algorithms and arguments can be extended to the more general case with arbitrary $g$.
focusing on $g=1$ avoids many floors, ceilings, and simplifies our presentation.
%From an application viewpoint, 
Furthermore, any Royal Mail instance can be converted to $g = 1$ using small discretization. 
% Since any $\rho$-approximation algorithm for $P||C_{\max}$ can be converted into a $2\rho$-approximation algorithm for BJS, we get an $8/3$-approximation ration for Longest Processing Time First (LPT) algorithm.
% Section~\ref{Section:LPT} shows that LPT is tightly $2$-approximate.
%Similarly, to the LPT case, our generic argument for converting any $\rho$-approximate $P||C_{\max}$ schedule into a $2\rho$-approximate schedule for the BJS problem implies a 4-approximation ratio for the Shortest Processing Time First (SPT) algorithm.

Section~\ref{Section:Long_Short} distinguishes \emph{long} versus \emph{short} instances.
An instance $\langle m,\mathcal{J}\rangle$ is \emph{long} if $p_j\geq m$ for each $j\in\mathcal{J}$ and \emph{short} if $p_j<m$ for all $j\in\mathcal{J}$.
This distinction comes from the observation that idle time occurs mainly because of (i) simultaneous job completions for long jobs and (ii) limited allowable parallel job executions for short jobs.
%Section~\ref{Section:Long_Short_LPT} 
Section~\ref{Section:Long_Short} proves that LPT is 5/3-approximate for long instances and optimal for short instances.
A key ingredient for establishing the ratio in the case of long instances is a concave relaxation for bounding idle machine time, \change{before the last job start}.
% Observe that an arbitrary instance $\langle m,\mathcal{J}\rangle$ admits a partitioning into a long instance $\langle m,\mathcal{J}^L \rangle$ and a short instance $\langle m,\mathcal{J}^S\rangle$, where $\mathcal{J}=\mathcal{J}^L\cup\mathcal{J}^S$.
% Therefore, algorithms for long and short instances can be converted into algorithms for arbitrary instances.
% Section~\ref{Section:SPT} considers \emph{Long Job Shortest Processing Time First Algorithm (LSPT)}, which schedules long jobs in non-decreasing order of processing times.
% For long instances, we show that LSPT achieves a better approximation ratio than LPT when the largest processing time $p_{\max}$ is relatively small.
%A key observation is that LSPT avoids any idle periods due to the BJSP constraint when scheduling long instances, but at the price of attaining low imbalance.
%{\color{red} 
Section~\ref{Section:Long_Short} also obtains an improved approximation ratio for long instances, when the maximum job processing time is relatively small, using the \emph{Shortest Processing Time First (SPT)} algorithm.
\change{For long instances, our analysis shows that LPT and SPT achieve low idle machine time after and before, respectively, the last job begins.
%, while SPT achieves low idle time before the last job begins. 
We leave determining the best trade-off between the two in order to %minimize the makespan 
%of the resulting schedule is minimized 
%is the challenge. 
obtain a better approximation ratio as an open question.}
%}

% Interestingly, SPT scheduling avoids any idle time until the last job begins in the case of long instances.
% Using this observation, Section~\ref{Section:SPT} obtains a tight 2-approximation ratio for a variant of SPT algorithm in the more general case with arbitrary jobs.
% Our findings demonstrate that LPT achieves low imbalance and possibly significant idle machine time, while the SPT performance is the other way round.

Greedy scheduling, e.g.\ LPT and SPT, which sequences long jobs first and short jobs next, or vice versa, cannot achieve an approximation ratio better than 2.
%This is the case for the LPT and SPT algorithms.
Section~\ref{Section:LSM} proposes \emph{Long-Short Mixing (LSM)}, which devotes a certain number of machines to long jobs and uses all remaining machines for short jobs.
%{\color{red} 
By executing the two job types in parallel, LSM \change{reduces the idle time before the last job begins and} achieves a 1.985-approximation ratio for a broad family of instances.
\change{
Carefully bounding idle time before the last job start by accounting for the parallel execution of long jobs with short job starts is the main technical difficulty behind our analysis.}
For degenerate instances with many very long jobs, we require constant-factor machine augmentation, i.e.\ $fm$ machines where $f>1$ is constant, to achieve a strictly lower than 2-approximation ratio.

Because Royal Mail delivery scheduling is subject to uncertainty, Section~\ref{Section:Uncertainty} exploits machine augmentation and lexicographic optimization for $P||C_{\max}$ under uncertainty \cite{Letsios2018,Skutella2016}
%a recent method for constructing robust solutions to $P||C_{\max}$ based on lexicographic optimization \cite{Letsios2018} 
to construct a two-stage robust optimization approach for the BJSP under uncertainty.
\change{We measure robustness based on the resource augmentation required for the final solution feasibility.
Our approach distinguishes which among different solutions is more robust.}
Section~\ref{Section:Numerical_Results} substantiates our \change{algorithms and} robust optimization approach empirically using Royal Mail data.
%data from Royal Mail delivery offices.
Section~\ref{Section:Conclusion} concludes with a collection of intriguing future directions.

%%%%%%%%%%%%%%%%%%%%%%%%%%%%%%%%%%%%%%%%%%%%%%%%%%%%%%%%%%%%%%%%

\section{Problem Definition and Preliminary Results}
\label{Section:Problem_Definition}

% This section describes and formulates the mail delivery scheduling problem as an integer program (IP).
% Section~\ref{Section:Numerical_Results} uses the proposed formulation for numerical evaluations.

% Section \ref{Section:Problem_Definition} describes the mail delivery scheduling problem. 
% Section \ref{Section:Integer_Program} presents an integer programming formulation.
% This section describes the mail delivery scheduling problem and presents an integer programming formulation two MILP formulations: a discrete-time MILP model and a continuous-time MILP model.

% \subsection{Problem Description}
% \label{Section:Problem_Definition}

An instance $I=\langle m,\mathcal{J}\rangle$ of the \emph{Bounded Job Start Scheduling Problem (BJSP)} is specified by a set $\mathcal{M}=\{1,\ldots,m\}$ of parallel identical machines and a set $\mathcal{J}=\{1,\ldots,n\}$ of jobs.
A machine may execute at most one job per unit of time.  
Job $j\in \mathcal{J}$ is associated with an integer processing time $p_j$.
Each job should be executed non-preemptively, i.e.\ in a continuous time interval without interruptions, by a single machine.
\emph{BJSP parameter} $g$ imposes an upper bound on the number of jobs that may begin per unit of time. 
The goal is to assign each job $j\in \mathcal{J}$ to a machine and decide its starting time so that this BJSP constraint is not violated and the makespan, i.e.\ the time at which the last job completes, is minimized. 
Consider a feasible schedule $\mathcal{S}$ with makespan $T$. 
We denote the start time of job $j$ by $s_j$.
Each job $j$ must be entirely executed during the interval $[s_j,C_j)$, where $C_j=s_j+p_j$ is the completion time of $j$.
So, $T=\max_{j\in\mathcal{J}}\{C_j\}$.
%We say that 
Job $j$ is \emph{alive} at time $t$ if $t\in[s_j,C_j)$.
Let $\mathcal{A}_t=\{j:[s_j,C_j)\cap[t-1,t)\neq\emptyset,j\in\mathcal{J}\}$ and $\mathcal{B}_t=\{j:s_j\in[t-1,t),j\in\mathcal{J}\}$ be the set of alive and beginning jobs during time unit $t$, respectively.
Schedule $\mathcal{S}$ is feasible only if $|\mathcal{A}_t|\leq m$ and $|\mathcal{B}_t|\leq g$, for all $t$. 
%Table \ref{tbl:notation} in Appendix \ref{s:notation} summarizes the notation.

BJSP is strongly $\mathcal{NP}$-hard because it becomes equivalent with $P||C_{\max}$ in the special case where $g=\min\{m,n\}$.
Theorem \ref{Theorem:NP_hardness} shows that BJSP is strongly $\mathcal{NP}$-hard also when $g=1$.

\begin{theorem}
\label{Theorem:NP_hardness}
BJSP is strongly $\mathcal{NP}$-hard in the special case $g=1$.
\end{theorem}
\begin{proof}
We present an $\mathcal{NP}$-hardness proof from 3-Partition.
Given a set $\mathcal{A}=\{a_1,\ldots,a_{3m}\}$ and a parameter $B\in\mathbb{Z}^+$ s.t.\ $a_j\in\mathbb{Z}^+$, $B/4\leq a_j\leq B/2$, for $j\in\{1,\ldots,3m\}$, and $\sum_{j=1}^{3m} a_j=mB$, the 3-Partition problem asks whether there exists a partition of $\mathcal{A}$ into $m$ subsets $\mathcal{S}_1,\ldots,\mathcal{S}_m$ s.t. $\sum_{j\in \mathcal{S}_i}a_j=B$ for each $i\in\{1,\ldots,m\}$.
Given an instance of 3-Partition, construct a BJSP instance $I=\langle m,\mathcal{J}\rangle$ with $n=3m$ jobs of processing time $p_j=n^2a_j$, for $j\in \{1,\ldots,3m\}$, and BJSP parameter $g=1$.
W.l.o.g., $n^2>3n$.
We show that $\mathcal{A}$ admits a 3-Partition iff there exists a feasible schedule $\mathcal{S}$ of makespan $T<n^2B+n^2$ for $I$.

%For the first direction, 
$\implies$: Suppose that $\mathcal{A}$ admits a 3-Partition $\mathcal{S}_1,\ldots,\mathcal{S}_m$.
Because $B/4< a_j< B/2$, for $j\in\{1,\ldots,3m\}$, $\mathcal{S}_i$ contains exactly three elements, i.e.\ $|\mathcal{S}_i|=3$, for each $i\in\{1,\ldots,m\}$.
We fix some arbitrary order $1,\ldots,3m$ of all jobs and construct a schedule $\mathcal{S}$ for $I$ where all jobs in $\mathcal{S}_i$ are executed by machine $i\in\mathcal{M}$.
The job starting times are decided greedily. 
In particular, let $T_i$ be the last job completion time in machine $i$ just before assigning job $j$. 
If no job has been assigned to $i$, then $T_i=0$.
We set $s_j$ equal to the earliest time slot after $T_i$ at which no job begins in any machine, i.e.\ $\min\{t:|\mathcal{B}_t|<1,t>T_i\}$. 
Now, let $T_i$ be the last completion time in machine $i$, once the greedy procedure has been completed. 
% The job starting times in $\mathcal{S}$ are decided in a greedy round-robin manner.
% For $i=1,\ldots,m$, the first job of the $i$-th machine is scheduled in the earliest time slot so that the BJS constraint is not violated. 
% This greedy policy is repeated for the second and third jobs of the machines.
Consider any job $j\in\mathcal{S}_i$ and let $j'\in\mathcal{S}_i$ be the last job executed before $j$ in machine $i$.
If no job is executed before $j$, then $s_j\leq n$.
Otherwise, by construction, we have $|\mathcal{B}_t|=1$, for every $t\in[C_{j'}+1,s_j-1]$.
Hence, $s_{j}-C_{j'}\leq n-1$.
Since $|\mathcal{S}_i|=3$ and $|\{t:|\mathcal{B}_t|=1, t\in\mathcal{D}\}|\leq n$, we conclude 
% Because of the BJS constraint $g=1$, at most one job must begin per unit of time.
% So, there can be at most $n$ idle slots before the execution of each job, due to the greedy round-robin policy.
% Since $\sum_{j\in \mathcal{S}_i}a_j=B$, we get that machine $i\in\mathcal{M}$ has completion time 
$T_i\leq \sum_{j\in \mathcal{S}_i}p_j+3n=n^2\left(\sum_{j\in \mathcal{S}_i}a_j\right)+3n= n^2B+3n< n^2B+n^2$.
So schedule $\mathcal{S}$ attains makespan $T<n^2B+n^2$.

%To the opposite direction, 
$\impliedby$: Suppose that there exists a feasible schedule $\mathcal{S}$ of makespan $T<n^2B+n^2$ for $I$.
We argue that each machine executes exactly three jobs.
Suppose for contradiction that machine $i\in\mathcal{M}$ executes a subset $\mathcal{S}_i$ of jobs with $|\mathcal{S}_i|\geq 4$.
Denote by $T_i=\max\{C_j:j\in\mathcal{S}_i\}$ the last job completion time in machine $i$.
Then, $T_i\geq\sum_{j\in\mathcal{S}_i}p_j=n^2\left(\sum_{j\in\mathcal{S}_i}a_j\right)$.
Because $a_j\in\mathbb{Z}^+$ and $a_j>B/4$, it must be the case that $\sum_{j\in\mathcal{S}_i}a_j\geq B+1$.
Hence, $T_i\geq n^2B+n^2$, which is a contradiction on the fact that $T_i\leq T$.
Thus, schedule $\mathcal{S}$ defines a partitioning of the jobs into $m$ subsets $\mathcal{S}_1,\ldots,\mathcal{S}_m$ s.t.\ $|\mathcal{S}_i|=3$, for each $i\in\mathcal{M}$.
We claim that $\sum_{j\in \mathcal{S}_i}a_i=B$.
Otherwise, there would be a machine $i\in\mathcal{M}$ with $\sum_{j\in \mathcal{S}_i}a_i\geq B+1$ and we would obtain a contradiction using similar reasoning to before.
We conclude that $\mathcal{A}$ admits a 3-Partition.
\end{proof}

% \subsection{Discrete-Time MILP Formulation}
% \label{Section:Integer_Program}

Next, we investigate the integrality gap of a natural integer programming formulation.
To obtain this integer program, we partition the time horizon into a set $D=\{1,\ldots, \tau\}$ of unit-length discrete time slots. 
Time slot $t\in D$ corresponds to time interval $[t-1,t)$. 
We may na\"ively choose $\tau=\sum_{j\in\mathcal{J}}p_j$, but smaller $\tau$ values are possible using tighter makespan upper bounds.
For simplicity, this manuscript assumes discrete time intervals $[s,t]=\{s,s+1,\ldots,t-1,t\}$, i.e.\ of integer length.
Interval $[1,\tau]$ is the \emph{time horizon}.
% For each $i\in I$, duration $d_i$ is rounded to $\lceil d_i/\lambda\rceil$.
%The main idea of 
In integer programming Formulation (\ref{Eq:DT_Model}), binary variables decide a starting time for each job.
Binary variable $x_{j,s}$ is 1 if job $j\in \mathcal{J}$ begins at time slot $s\in D$, and 0 otherwise. 
Continuous variable $T$ corresponds to the makespan.
If job $j$ starts at $s$, then it is performed exactly during the time slots $s,s+1,\ldots,s+p_j-1$.
Hence, job $j$ is alive at time slot $t$ iff it has begun at one among the time slots in the set $A_{j,t}=\{t-p_j+1,t-p_j+2,\ldots,t\}$.
To complete before the time horizon ends, job $j$ must begin at a time slot in the set $F_j=\{1,2,\ldots,\tau-p_j+1\}$.
Finally, denote by $\mathcal{J}_s=\{j:s\in \mathcal{F}_j,j\in\mathcal{J}\}$ the eligible subset of jobs at $s$, i.e.\ the ones that may be feasibly begin at time slot $s$ without exceeding the time horizon.
Formulation (\ref{Eq:DT_Model}) models the BJSP problem.
%We denote by $c=\delta/g$, the throughput capacity. 
%For each time slot $t$, at most $\lceil c\rceil$ vehicles may depart during the time slots $C_t=\{t,t+1,\ldots,t-1+\lceil1/c\rceil\}$. 
%It suffices to distinguish two cases, based on whether $c\geq 1$, or $c<1$.
%
\begin{subequations}
\label{Eq:DT_Model}
\begin{align}
\min_{x_{j,s},\;T} \quad & T \label{Eq:D_Objective} \\
& T \geq x_{j,s}(s+p_j) & j\in\mathcal{J}, s\in D \label{Eq:D_Makespan} \\
& \sum_{j\in \mathcal{J}} \sum_{\substack{s\in A_{j,t}}} x_{j,s} \leq m & t\in D \label{Eq:D_Vehicles_Count} \\
& \sum_{s\in F_j} x_{j,s} = 1 & j\in \mathcal{J} \label{Eq:D_Itinerary_Assignment} \\
& \sum_{j\in \mathcal{J}_s} x_{j,s} \leq g & s\in D \label{Eq:D_Throughput} \\ 
& x_{j,s}\in\{0,1\} & j\in \mathcal{J}, s\in F_j \label{Eq:Integrality}
\end{align}
\end{subequations} 

Expression (\ref{Eq:D_Objective}) minimizes makespan.
Constraints (\ref{Eq:D_Makespan}) enforce that the makespan is equal to the last job completion time.
Constraints (\ref{Eq:D_Vehicles_Count}) ensure that at most $m$ machines are used at each time slot $t$. 
Constraints (\ref{Eq:D_Itinerary_Assignment}) requires that each job $j$ is scheduled.
Constraints (\ref{Eq:D_Throughput}) express the BJSP constraint.

%As a preliminary result, 
Theorem~\ref{Thm:Integrality_Gap} shows that the fractional relaxation obtained by replacing Eq.\ (\ref{Eq:Integrality}) with the constraints $0\leq x_{j,s}\leq 1$, for $j\in\mathcal{J}$ and $s\in F_j$, has a non-constant integrality gap.
Thus, stronger linear programming (LP) relaxations are required for obtaining constant-factor approximation algorithms with LP rounding.

\begin{theorem}
\label{Thm:Integrality_Gap}
The fractional relaxation of integer programming formulation (\ref{Eq:DT_Model}) has integrality gap $\Omega(\log n)$.
\end{theorem}

%\begin{proof} %{Appendix \ref{Appendix:Fractional_IntegralityGap}}
\begin{proof}
Consider an instance with $m$ machines,  $n=m$ jobs of processing time $p_j=1$ for each $j\in\mathcal{J}$, and BJSP parameter $g=m$.
For this instance, the LP solution sets $x_{j,s}=1/(s\cdot\sum_{t=1}^{\tau}\frac{1}{t})$ for each $j,s$.
The LP fractional solution is feasible as at each time, no more than $m$ job pieces are feasibly executed (and begin), while the cost is $\max\{sx_{j,s}\}=1/\sum_t\frac{1}{t}$.
On the contrary, the optimal integral solution has makespan 1. 
\end{proof}

\section{LPT Algorithm}
\label{Section:LPT}

Longest Processing Time first algorithm (LPT) schedules the jobs on a fixed number $m$ of machines w.r.t.\ the order $p_1\geq\ldots\geq p_n$.
Recall that $|\mathcal{A}_t|$ and $|\mathcal{B}_t|$ is the number of alive and beginning jobs, respectively, at time slot $t\in D$.
We say that time slot $t\in D$ is \emph{available} if $|\mathcal{A}_t|<m$ and $|\mathcal{B}_t|<g$.
LPT schedules the jobs greedily w.r.t.\ their sorted order. 
Each job $j$ is scheduled in the earliest available time slot, i.e.\ at $s_j=\min\{t:|\mathcal{A}_t|<m, |\mathcal{B}_t|<g, t\in D\}$.
Theorem \ref{Thm:Naive_LPT} proves a tight approximation ratio of 2 for LPT.

\begin{theorem}
LPT is 2-approximate for minimizing makespan and this ratio is tight. 
\label{Thm:Naive_LPT}
\end{theorem}
\begin{proof}
Denote by $\mathcal{S}$ and $\mathcal{S}^*$ the LPT and a minimum makespan schedule, respectively.
Let $\ell$ be the job completing last in $\mathcal{S}$, i.e.\ $T=s_{\ell}+p_{\ell}$.
For each time slot $t\leq s_{\ell}$, either $|\mathcal{A}_t|=m$, or $|\mathcal{A}_t|<m$.
Since $\ell$ is scheduled at the earliest available time slot, for each $t\leq s_{\ell}$ s.t. $|\mathcal{A}_t|<m$, we have $|\mathcal{B}_t|=g$.
Let $\lambda$ be the total length of time s.t.\ $|\mathcal{A}_t|<m$ in $\mathcal{S}$.
Because of the BJSP constraint, exactly $g$ jobs begin per unit of time, which implies that $\lambda\leq \lceil\frac{\ell}{g}\rceil$. 
Therefore, schedule $\mathcal{S}$ has makespan:
\begin{equation*}
T = s_{\ell} + p_{\ell} \leq \frac{1}{m}\sum_{j\neq \ell} p_j + \lambda + p_{\ell} \leq \frac{1}{m}\sum_{j=1}^n p_j + 
\left(\left\lceil\frac{\ell}{g}\right\rceil+p_{\ell}\right).
\end{equation*}
Denote by $s_j^*$ the starting time of job $j$ in $\mathcal{S}^*$ and let $\pi_1,\ldots,\pi_n$ the job indices ordered in non-decreasing schedule $\mathcal{S}^*$ starting times, i.e.\ 
$s_{\pi_1}^*\leq\ldots\leq s_{\pi_n}^*$.
Because of the BJSP constraint, $s_{\pi_j}^*\geq\lceil j/g\rceil$. 
In addition, there exists $j'\in[j,n]$ s.t.\ $p_{\pi_{j'}}>p_j$.
Thus, $\max_{j'=j}^n\{s_{\pi_{j'}^*}+p_{\pi_{j'}}\}\geq \lceil j/g\rceil+p_j$, for $j=1,\ldots,n$.
%and the LPT ordering, $s_{\pi_j}^*+p_{\pi_j}\geq \lceil j/g\rceil+p_j$ for each $j\in\mathcal{J}$.
Hence, $\mathcal{S}^*$ has makespan:
\begin{equation*}
T^* \geq \max\left\{\frac{1}{m}\sum_{j=1}^np_j,\max_{j=1}^n\left\{\left\lceil\frac{j}{g}\right\rceil+p_j\right\}\right\}.
\end{equation*}
We conclude that $T\leq 2T^*$.

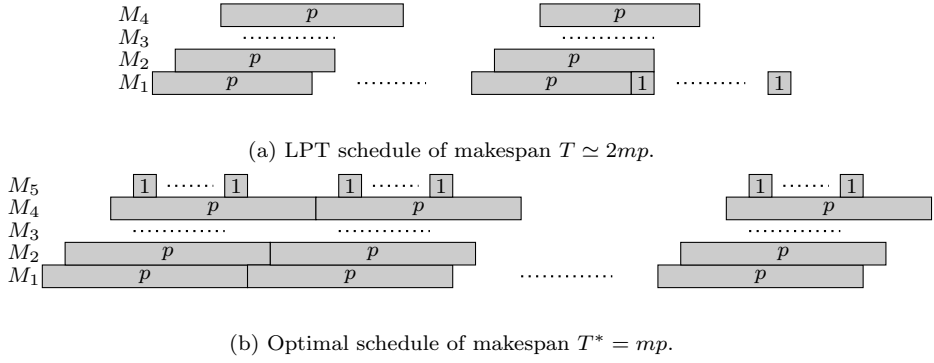
\begin{figure}[t]
\begin{subfigure}[t]{\textwidth}
\begin{center}
\input{2_tightness_lpt.tex}
\end{center}
\caption{LPT schedule of makespan $T\simeq 2mp$.}
\end{subfigure}
\begin{subfigure}[t]{\textwidth}
\begin{center}
\input{2_tightness_opt.tex}
\end{center}
\caption{Optimal schedule of makespan $T^*=mp$.}
\end{subfigure}
\caption{\change{
BJSP instance with $g = 1$ for the tightness of the LPT 2-approximation ratio, with $m$ machines, $\Theta(m^2)$ jobs of processing time $p$ and $\Theta(mp)$ unit-length jobs such that $p=\omega(m)$ and $m=\omega(1)$.}}
\label{Figure:LPT_Tightness}
\end{figure}

\change{ Figure~\ref{Figure:LPT_Tightness} illustrates a tightness example of our analysis for LPT.}
Consider an instance $I=\langle m,\mathcal{J}\rangle$ with $m(m-1)$ long jobs of processing time $p$, where $p=\omega(m)$ and $m=\omega(1)$, $m(p-m)$ unit jobs, and BJSP parameter $g=1$.
LPT schedules the long jobs into $m-1$ \change{ groups}, each one with exactly $m$ . 
All jobs of a \change{ group} are executed in parallel for their greatest part. 
In particular, the $i$-th job of the $k$-th \change{ groups} is executed by machine $i$ starting at time slot $(k-1)p+i$.
All unit jobs are executed sequentially by machine $1$ starting at $(m-1)p+1$.
Observe that $\mathcal{S}$ is feasible and has makespan $T=(m-1)p+m(p-m)=(2m-1)p-m^2$.
The optimal solution $\mathcal{S}^*$ schedules all jobs in $m$ \change{ groups}.
The $k$-th \change{ group} contains $(m-1)$ long jobs and $(p-m+1)$ unit jobs.
Specifically, the $i$-th long job is executed by machine $i$ beginning at $(k-1)p+i$, while all short jobs are executed consecutively by machine $m$ starting at $(k-1)p+m$ and completing at $kp$.
Schedule $\mathcal{S}^*$ is feasible and has makespan $T^*=mp$.
Because $\frac{m}{p}\rightarrow0$ and $\frac{1}{m}\rightarrow0$, i.e.\ both approach zero, $T\rightarrow 2T^*$.
%Since $\frac{m}{p}\simeq 0$ and $\frac{1}{m}\simeq 0$, we conclude that $T\simeq 2T^*$.
%Appendix~\ref{Appendix:LPT_Tightness} illustrates this tightness example.
\end{proof}

\section{Long and Short Instances}
\label{Section:Long_Short}

This section assumes that $g = 1$, but several of the arguments can be extended to arbitrary $g$.
From an application viewpoint, any Royal Mail instance can be converted to $g=1$ using small discretization. 

\subsection{Longest Processing Time First}
\label{Section:Long_Short_LPT}

%With the aim of bounding the total idle time of BJS feasible schedules, 
We consider two natural classes of BJSP instances for which LPT achieves an approximation ratio better than 2.
Instance $\langle m,\mathcal{J}\rangle$ is (i) \emph{long} if $p_j\geq m$ for each $j\in \mathcal{J}$ and (ii) \emph{short} if $p_j<m$ for every $j\in \mathcal{J}$.
This section proves that LPT is 5/3-approximate for long instances and optimal for short instances.
\change{ Intuitively, LPT schedules for long instances contain a significant amount of time without job starts, where all machines execute long jobs in parallel. LPT schedules for short instances have no time where all machines simultaneously execute jobs in parallel, because of the BJSP constraint and the fact that all jobs are short. In this case, the number of job starts is significant compared to the overall makespan. Using these observations, we are able to obtain better than 2-approximate schedules for these two classes of instances.}

Consider a feasible schedule $\mathcal{S}$ and let $r=\max_{j\in\mathcal{J}}\{s_j\}$ be the last job start time.
%suppose that job $\ell\in\mathcal{J}$ begins last, i.e.\ $s_{\ell}\geq s_j$ for each $j\in\mathcal{J}$, and let $B=s_{\ell}.$ 
%Recall that $S$ is feasible if it satisfies IP constraints (\ref{Eq:D_Vehicles_Count})-(\ref{Eq:Integrality}).
We say that $\mathcal{S}$ is a \emph{compact schedule} if it holds that either (i) $|\mathcal{A}_t|=m$, or 
(ii) $|\mathcal{B}_t|=1$, for each $t\in[1,r]$.
%Based on Lemma~\ref{Lemma:Initial_Idle}, we may restrict our attention to the class of compact schedules. 
%Furthermore, Lemma~\ref{Lemma:Initial_Idle} derives a lower bound on the optimal makespan.
Lemma~\ref{Lemma:Initial_Idle} shows the existence of an optimal compact schedule and derives a lower bound on the optimal makespan.

% \begin{figure}[t]
% \begin{center}
% \input{idle_bound.tex}
% \end{center}
% \caption{\color{red} Structure of a compact schedule. At each time $t$ until the last job starting time $b$, either all machines process a job, or a job begins.}
% \label{Figure:Compact_Schedule}
% \end{figure}

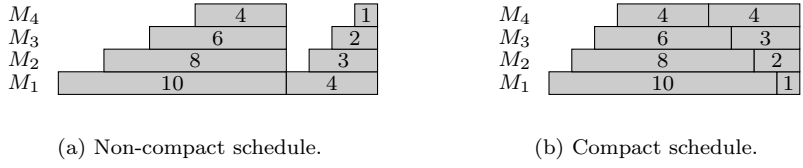
\begin{figure}[t]
\begin{subfigure}[t]{0.5\textwidth}
\begin{center}
\input{compact_initial.tex}
\end{center}
\caption{Non-compact schedule.}
\end{subfigure}
\begin{subfigure}[t]{0.5\textwidth}
\begin{center}
\input{compact_final.tex}
\end{center}
\caption{Compact schedule.}
\end{subfigure}
\caption{\change{ Converting a non-compact schedule to a compact one, by shifting jobs back in increasing order of their starting times.}}
\label{Figure:Compact_Schedule}
\end{figure}

%The Lemma~\ref{Lemma:Initial_Idle} proof requires the following definition for comparing schedules lexicographically.
%in the remainder of the section.

\begin{lemma}
\label{Lemma:Initial_Idle}
For each instance $I=\langle m,\mathcal{J}\rangle$, 
% and let $\mathcal{J}^L=\{j:p_j\geq m,j\in\mathcal{J}\}$ be the set of long jobs. Then, (i) 
there exists a feasible compact schedule $\mathcal{S}^*$ which is optimal. 
Let $\mathcal{J}^L=\{j:p_j\geq m,j\in\mathcal{J}\}$.
%and (ii) 
If $|\mathcal{J}^L|\geq m$, then $\mathcal{S}^*$ has makespan
$T^*\geq \frac{1}{m} \left( \frac{m(m-1)}{2}+\sum_{j=1}^np_j \right)$.
\end{lemma}
\begin{proof} 
% Part (i) requires Definition~\ref{Def:LexicographiC_Comparison} for lexicographically comparing schedules.
% \begin{definition}
% \label{Def:LexicographiC_Comparison}
% %Let $\pi$ be the order of jobs with in non-increasing starting times, i.e.\ $s_{\pi_1}\geq\ldots\geq s_{\pi_n}$, in a feasible schedule $\mathcal{S}$.
% Consider two feasible schedules $\mathcal{S}$ and $\mathcal{S}'$ with corresponding orders $\pi$ and $\pi'$ of jobs in non-increasing start times. 
% That is, for $j\in\{1,\ldots,n\}$, $s_{\pi_j}$ and $s_{\pi_j'}'$ is the $j$-th greatest start time in $\mathcal{S}$ and $\mathcal{S}'$, respectively. %, breaking ties arbitrarily.
% %$s_{\pi_1}\geq\ldots\geq s_{\pi_n}$ and $s_{\pi_1'}'\geq\ldots\geq s_{\pi_n'}'$, respectively.
% %That is, the job starting times in $S$ and $\mathcal{S}'$ satisfy $s_{\pi_1}\geq\ldots\geq s_{\pi_n}$ and $s_{\pi_1'}'\geq\ldots\geq s_{\pi_n'}'$.
% We say that $\mathcal{S}$ is lexicographically smaller than $\mathcal{S}'$, and write $\mathcal{S}\leq_{\text{lex}}\mathcal{S}'$, if $s_{\pi_{\ell}}< s_{\pi_\ell'}'$, where $\ell=\min\{j:s_{\pi_j}\neq s_{\pi_j'}',j\in\{1,\ldots,n\}\}$.
% \end{definition}
For the first part, among the set of all optimal schedules, pick the schedule $\mathcal{S}^*$ lexicographically minimizing\footnote{\change{ Let $s_{j_1}^*<\ldots<s_{j_n}^*$ be order of job start times in $\mathcal{S}^*$.
Moreover, denote by $s_{j_1}<\ldots<s_{j_n}$ the order of job starts in another arbitrarily chosen optimal schedule $\mathcal{S}$.
If $k\in[1,n]$ is the smallest index such that $s_{j_k}^*>s_{j_k}$, then there exists $k'\in[1,k)$ such that $s_{j_k'}^*<s_{j_k'}$.}} the vector of job start times sorted in non-decreasing order.
% Let $\pi^*$ and $\pi$ be the order of jobs in non-increasing starting times in $\mathcal{S}^*$ and a different optimal schedule
% let $\mathcal{S}$ be an optimal schedule with a vector $\vec{s}=(s_1,\ldots,s_n)$ of job starting times and suppose that $\mathcal{S}$ is different from $\mathcal{S}^*$.
% Then, $\mathcal{S}^*<_{\text{lex}}\mathcal{S}$, i.e.\ $s_{\ell}^*<s_{\ell}$, where $\ell=\min\{j:s_j^*\neq s_{j},j\in\mathcal{J}\}$.
We claim that $\mathcal{S}^*$ is compact.
Assume that this is not the case.
Then, there exists a time $t\in [1,r)$ such that $|\mathcal{A}_t|<m$ and $|\mathcal{B}_t|<1$.
Consider the earliest such time $t$.
Moreover, let $t'$ be the earliest time $t'>t$ satisfying either $|\mathcal{A}_t|=m$, or $|\mathcal{B}_t|=1$.
Clearly, there exists a job $j\in\mathcal{J}$ such that $s_j=t'$.
If we decrease the job $j$ start time by one unit of time, we obtain a feasible compact schedule $\mathcal{S'}$ with makespan 
$T'\leq T^*$, where $T^*$ is the schedule $\mathcal{S}^*$ makespan, and $\mathcal{S}'$ has a lexicographically smaller vector of job start times than $\mathcal{S}^*$ which is a contradiction. %the fact that $\mathcal{S}$ lexicographically minimizes the job start times.
\change{ Figure~\ref{Figure:Compact_Schedule} illustrates how to convert a non-compact schedule to a compact one.}

The second part requires lower bounding the total idle machine time $\Lambda=\sum_{t=1}^{r}(m-|\mathcal{A}_t|)$ for each feasible schedule $\mathcal{S}$.
%\begin{definition}
%\label{Def:Idle_Time}
%In a feasible schedule $\mathcal{S}$, the idle machine time is $\Lambda=\sum_{t=1}^{\tau}(m-|\mathcal{A}_t|)$. 
%\end{definition}
%Consider any feasible compact schedule $\mathcal{S}$ and let $\gamma=\lceil m/g\rceil-1$. 
%By the BJS constraint, $|\mathcal{B}_t|=g$, hence $m-|\mathcal{A}_t|\geq m-tg$, for $t=1,\ldots,\gamma$.
By the BJSP constraint \change{ and the fact that $|\mathcal{J}^L|\geq m$}, $|\mathcal{B}_t|=1$ and, hence, $m-|\mathcal{A}_t|\geq m-t$, for each $t\in\{1,\ldots,m\}$.
This is because all machines are idle before the first time slot.
%Hence, $m-|\mathcal{A}_t|\geq(m-t)g$, for $t=1,\ldots,a$.
\change{ So, the idle machine during time interval $[1,m]$ is lower bounded by $\Lambda\geq\sum_{t=1}^m(m-t)=\frac{m(m-1)}{2}$}.
% \begin{equation*}
% \Lambda\geq\sum_{t=1}^{\gamma}(m-tg)=am-\frac{a(a+1)g}{2}
% \end{equation*}
A simple packing argument \change{ (including both idle and non-idle machine time)} implies that schedule $\mathcal{S}$ has makespan $T\geq\frac{1}{m}\left(\frac{m(m-1)}{2}+\sum_{j=1}^np_j\right)$.
\end{proof}

Next, we analyze LPT in the case of long instances.
Similar to the Lemma \ref{Lemma:Initial_Idle} proof, we may show that LPT produces a compact schedule $\mathcal{S}$.
%Consider an arbitrary feasible compact schedule $S$ and let $i\in\mathcal{J}$ be the job that completing last, i.e.\ $C_i=T$. 
%the job that completes last, i.e.\ $C_i=\max_{j=1}^n\{C_j\}$ or equivalently $C_i=T(S)$.
We partition the interval $[1,r]$ into a sequence $P_1,\ldots,P_k$, \change{ where $k\leq r$}, of maximal periods satisfying the following invariant:
for each $q\in\{1,\ldots,k\}$, either (i) $|\mathcal{A}_t|<m$ for each $t\in P_q$ or (ii) $|\mathcal{A}_t|=m$ for each $t\in P_q$.
That is, there is no pair of time slots $s,t\in P_q$ such that $|\mathcal{A}_s|<m$ and $|\mathcal{A}_t|=m$.
We call $P_q$ a \emph{slack period} if $P_q$ satisfies (i),
otherwise, $P_q$ is a \emph{full} period. 
%Define a \emph{slot} as a pair $(i,t)$ associated with a machine $i$ and a time $t$.
For a given period $P_q$ of length $\lambda_q$, denote by $\Lambda_q=\sum_{t\in P_q}(m-|\mathcal{A}_t|)$ the idle machine time.
Note that $\Lambda_q=0$, for each full period $P_q$.
Lemma \ref{Lemma:Idle_Slots} upper bounds the total idle machine time of slack periods in the LPT schedule $\mathcal{S}$, except the very last period $P_k$. 
When $P_k$ is slack, the length $\lambda_k$ of $P_k$ is upper bounded by Lemma~\ref{Lem:Last_Slack_Period}.

\begin{lemma}
\label{Lemma:Idle_Slots}
Let $k'=k-1$.
Consider a long instance $I=\langle m,\mathcal{J}\rangle$, with $|\mathcal{J}|\geq m$, and the LPT schedule $\mathcal{S}$. 
It holds that (i) $\lambda_q\leq m-1$ and (ii) $\Lambda_q\leq\frac{\lambda_q(\lambda_q-1)}{2}$ for each slack period $P_q$, where 
$q\in\{1,\ldots,k'\}$.
Furthermore, (iii) $\sum_{q=1}^{k'}\Lambda_q\leq \frac{nm}{2}$.
\end{lemma}
\begin{proof}
For (i), let $P_q=[s,t]$ be a slack time period in $\mathcal{S}$ and assume for contradiction that $\lambda_q\geq m$, i.e.\ $t\geq s+m-1$.
Given that $p_j\geq m$ for each $j\in\mathcal{J}$, we have $\{j:s_j\in[s,s+m-1],j\in\mathcal{J}\}\subseteq \mathcal{A}_{s+m-1}$. 
That is, all jobs starting during $[s,s+m-1]$ are alive at time $s+m-1$.
Since $P_q$ is a slack period, $|\mathcal{A}_u|<m$ holds for each $u\in [s,s+m-1]$.
Because $\mathcal{S}$ is compact, $|\mathcal{B}_u|=1$, i.e.\ exactly $g=1$ jobs begin, at each $u\in [s,s+m-1]$.
This implies $|\mathcal{A}_{s+m-1}|\geq m$, contradicting the fact that $P_q$ is a maximal slack period. 

For (ii), consider the partitioning $\mathcal{A}_u=\mathcal{A}_u^-+\mathcal{A}_u^+$ for each time slot $u\in P_q=[s,t]$, where $\mathcal{A}_u^-$ and $\mathcal{A}_u^+$ is the set of alive jobs at time $u$ completing inside $P_q$ and after $P_q$, i.e.\ $C_j\in[s,t]$ and $C_j>t$, respectively.
Since $\lambda_q\leq m-1$, every job $j$ beginning during $P_q$, i.e.\ $s_j\in[s,t]$, must complete after $P_q$, i.e.\ $C_j>t$.
We modify schedule $\mathcal{S}$ by removing every job $j$ completing inside $P_q$, i.e.\ $C_j\in P_q$.
Clearly, the modified schedule $\mathcal{S}'$ has increased idle time $\Lambda_q'$ during $P_q$, i.e.\ $\Lambda_q\leq \Lambda_q'$.
Further, no job $j$ with $s_j\in P_q$ is removed.
Because $|\mathcal{B}_u|=1$ for each $u\in P_q$, we have $|\mathcal{A}_u|=|\mathcal{A}_{u+1}|-1$ for $u=s,\ldots,t$.
Furthermore, $|\mathcal{A}_{t+1}|=m$.
So:
%, it must be the case that: 
% and $|\mathcal{A}_u|=|\mathcal{A}_{u+1}|-1$ for $u=s,\ldots,t$ in $S'$, where $e_u$ is the number of empty, i.e.\ idle, machines at $u$.
% We conclude that
\begin{equation*}
\Lambda_q' = \sum_{u=s}^t(m-|\mathcal{A}_u|) = \sum_{u=s}^t[m-(t-s+1)] = \sum_{u=1}^{\lambda_q-1}u = \frac{\lambda_q(\lambda_q-1)}{2}.
\end{equation*}
% Therefore, $\sum_{u=s-1}^tf_u=\sum_{u=s}^tb_u$, where $b_u$ and $f_u$ is the number of jobs beginning and finishing, respectively, at time slot $u$.
% Since exactly one job begins at each time slot $u\in P_q$, we conclude that $\lambda_q=\sum_{u=s}^tb_u$.
% 
% In $S'$, no job completes during $P_q$.
% Since $b_u=1$ for each $u\in P_q$, we have that $e_u=e_{u+1}+1$ for $u=s,\ldots,t$, i.e.\ the number of idle periods increases by one.
% The length of the idle period remains $\lambda_q$ and exactly one job begins per unit of time in $P_q$.
% In fact, the length of the period is equal to the number of jobs completing at $s-1$ and the number of jobs completing during $P_q$.
% The number of machines is still $m$ and $m-\lambda_q$ in the idle period.
% Because we have removed the jobs, the number of idle slots is exactly $\lambda_q-u$ in the $u$-th time slot of the period.
% Hence, we conclude the required bound.
% Since $a_{u+1}=m$ and $\lambda_q=u-t+1$, we get that $a_{t'}=m-(u-t')$, for $t'\in P_q$.
% We conclude that
% \begin{equation*}
% L_q = \sum_{t'=t}^u(m-a_{t'}) = \sum_{t'=t}^u(u-t'+1) = \sum_{k=1}^{\lambda_q-1}k = \frac{\lambda_q(\lambda_q-1)}{2}.
% \end{equation*}
%
For (iii),
consider slack period $P_q$, for $q\in\{1,\ldots,k'\}$.
By (i), $\lambda_q\leq m-1$.
%Recall that at most one job may begin per unit of time, because of the BJSP constraint.
%Since $\mathcal{S}$ is compact, we get that $\sum_{q=1}^{k}\lambda_q\leq n-1$.
%By the BJSP constraint, 
Since at most $g=1$ jobs begin at each $t\in P_q$, 
$\sum_{q=1}^{k'}\lambda_q\leq n-1$.
By (ii), Concave Program (\ref{Eq:Idle_Model}) upper bounds $\sum_{q=1}^{k'}\Lambda_q$.
\begin{subequations}
\label{Eq:Idle_Model}
\begin{align}
\max_{\lambda \in \{ 0,1\}^{k'}} \quad & \sum_{q=1}^{k'} \frac{\lambda_q(\lambda_q-1)}{2} \label{Eq:Idle_Objective} \\
%& \lambda_q \leq m-1 & q\in \{1,\ldots,k\} \label{Eq:Idle_Individual_Length} \\
%& \lambda_q \leq m & q\in \{1,\ldots,k\} \label{Eq:Idle_Individual_Length} \\
%& \sum_{q=1}^{k}\lambda_q \leq n-1 \label{Eq:Total_Length} \\
& 1\leq \lambda_q \leq m & q\in \{1,\ldots,k'\} \label{Eq:Idle_Individual_Length} \\
& \sum_{q=1}^{k'}\lambda_q \leq n \label{Eq:Total_Length} 
\end{align}
\end{subequations} 
%Let $\lambda^*$ be an optimal solution to (\ref{Eq:Idle_Model}).
Assume w.l.o.g.\ that $n/m$ is integer.
If $k'\leq n/m$, by setting $\lambda_q=m$, for $q\in\{1,\ldots,k'\}$, we obtain $\sum_{q=1}^{k'}\lambda_q(\lambda_q-1)/2 \leq k'm(m-1)/2\leq nm/2$.
If $k'> n/m$, we argue that the solution $\lambda_q=m$, for $q\in\{1,\ldots,n/m\}$, and $\lambda_q=0$, otherwise, is optimal for Concave Program (\ref{Eq:Idle_Model}).
In particular, for any solution $\lambda$ with $0<\lambda_q,\lambda_{q'}<m$ such that $q\neq q'$, we may construct a modified solution with higher objective value by applying Jensen's inequality $f(\lambda)+f(\lambda')\leq f(\lambda+\lambda')$ for any $\lambda,\lambda'\in[0,\infty)$, with respect to the single variable, convex function $f(\lambda)=\lambda(\lambda-1)/2$. 
If $\lambda_q+\lambda_{q'}\leq m$, we may set $\tilde{\lambda}_q=\lambda_q+\lambda_{q'}$ and $\tilde{\lambda}_{q'}=0$.
Otherwise, $m<\lambda_q+\lambda_{q'}\leq 2m$ and we set $\tilde{\lambda}_q=m$ and $\tilde{\lambda}_{q'}=\lambda_q+\lambda_{q'}-m$.
In both cases, $\lambda_{q''}=\tilde{\lambda}_{q''}$, for each $q''\in\{1,\ldots,k'\}\setminus\{q,q'\}$.
Clearly, $\tilde{\lambda}$ attains higher objective value than $\lambda$, for Concave Program (\ref{Eq:Idle_Model}).
\end{proof}

\begin{lemma}
\label{Lem:Last_Slack_Period}
Suppose that $P_k$ is a slack period and let $\mathcal{J}_k$ be the set of jobs beginning during $P_k$.
Then, it holds that $\lambda_k\leq\frac{1}{m}\sum_{j\in\mathcal{J}_k}p_j$.
\end{lemma}
\begin{proof}
Because $P_k$ is a slack period, it must be the case that $|\mathcal{B}_u|=1$, for each $u\in P_k$.
Since we consider long instances, $p_j\geq m$ for each $j\in\mathcal{J}$.
Therefore, $\lambda_k\leq\frac{1}{m}\sum_{j\in\mathcal{J}_k}p_j$.
\end{proof}

\change{
\begin{theorem}
\label{Thm:LPT_Long}
LPT achieves an approximation ratio 
$\rho\in \left[\frac{3}{2},\frac{5}{3} \right]$ in the case of long instances.
%LPT achieves an approximation ratio $\rho\in[\frac{3}{2},\frac{5}{3}]$, for long instances.
\end{theorem}
}
\begin{proof}
Denote the LPT and optimal schedules by $\mathcal{S}$ and $\mathcal{S}^*$, respectively.
Let $\ell\in\mathcal{J}$ be a job completing last in $\mathcal{S}$, i.e.\ $T=s_{\ell}+p_{\ell}$.
Recall that LPT sorts jobs s.t.\ $p_1\geq\ldots\geq p_n$.
W.l.o.g.\ we assume that $\ell=\arg\min_{j\in\mathcal{J}}\{p_j\}$.
Indeed, we may discard every job $j>\ell$ and bound the algorithm's performance w.r.t.\ instance 
$\tilde{I}=\langle m,\mathcal{J}\setminus\{j:j>\ell,j\in\mathcal{J}\}\rangle$.
Let $\tilde{S}$ and $\tilde{S}^*$ be the LPT and an optimal schedule attaining makespan $\tilde{T}$ and $\tilde{T}^*$, respectively, for instance $\tilde{I}$.
Showing that $\tilde{T}\leq(5/3)\tilde{T}^*$ is sufficient for our purposes because $T=\tilde{T}$ and 
$\tilde{T}^*\leq T^*$.
We distinguish two cases based on whether $p_n> T^*/3$, or $p_n\leq T^*/3$.

Case 1 ($p_n> T^*/3$): We claim $T\leq (3/2)T^*$.
Initially, observe that $n\leq 2m$.
Otherwise, there would be a machine $i\in\mathcal{M}$ executing at least $|\mathcal{S}_i^*|\geq 3$ jobs, say $j,j',j''\in\mathcal{J}$, in $\mathcal{S}^*$.
This machine would have last job completion time $T_i^*\geq p_j+p_{j'}+p_{j''}>T^*$, a contradiction.
% Suppose that machine $i\in\mathcal{M}$ executes jobs $j,j',j''\in\mathcal{}$ 
% In $\mathcal{S}^*$, $i$ would have completion time $T_i^*>T(\mathcal{S}^*)$ which is a contradiction.
If $n\leq m$, LPT clearly has makespan $T=T^*$.
So, consider $n>m$.
Then, some machine executes at least two jobs in $\mathcal{S}^*$, i.e.\ $p_n\leq T^*/2$.
To prove our claim, it suffices to show $s_n\leq T^*$, i.e.\ job $n$ starts before or at $T^*$.
Let $c=\max_{1\leq j\leq m}\{C_j\}$ be the time when the last among the $m$ biggest jobs completes. 
If $s_n\leq c$, then $s_n\leq\max_{1\leq j\leq m}\{j+p_j\}\leq T^*$.
Otherwise, let $\lambda=s_n-c$.
Because $n\leq 2m$, it must be the case that $\lambda\leq m$. 
Furthermore, $|\mathcal{A}_t|<m$ and, thus, $|\mathcal{B}_t|=1$, for each $t\in[c+1,s_n-1]$.
That is, exactly one job begins per unit of time during $[c+1,s_n]$.
Due to the LPT ordering, these are exactly the jobs $\{n-\lambda,\ldots,n\}$.
Since $\lambda\leq m$ and $p_j\geq m$, at least $m-h$ units of time of job $n-h$ are executed from time $s_n$ and onwards, for $h\in\{1,\ldots,\lambda\}$.
Thus, the total processing load which executed not earlier than $s_n$ is $\mu\geq\sum_{h=1}^{\lambda}(m-h)$.
On the other hand, at most $m-h$ machines are idle at time slot $c+h$, for $h\in\{1,\ldots,\lambda\}$.
So, the total idle machine time during $[m+1,s_n-1]$ is $\Lambda\leq\sum_{h=1}^{\lambda-1}(m-h)$.
We conclude that $\mu\geq\Lambda$ which implies that $s_n\leq\frac{m(m-1)}{2}+\frac{1}{m}\sum_{j\in\mathcal{J}}p_j$.
By Lemma~\ref{Lemma:Initial_Idle}, our claim follows.

Case 2 ($p_n\leq T^*/3$):
In the following, Equalities (\ref{Eq:Long_Last_Job}) hold because job $n$ completes last and by the definition of alive jobs.
%The first equality holds since job $n$ completes last and the second by the definition of the set of alive jobs.
Inequalities (\ref{Eq:Long_Packing})-(\ref{Eq:Long_Bounding}) use a simple packing argument with job processing times and machine idle time taking into account: (\ref{Eq:Long_Packing}) Lemma~\ref{Lem:Last_Slack_Period}, (\ref{Eq:Long_Gap_Bounding}) Lemma~\ref{Lemma:Idle_Slots} property $\sum_{q=1}^{k'}\Lambda_q\leq \frac{nm}{2}$, and (\ref{Eq:Long_Bounding}) the bound $T^* \geq \max\{\frac{1}{m}\sum_{j\in\mathcal{J}}p_j,n+p_n,3p_n\}$.
% Recall that $T(S)=b_i+p_i$.
% W.l.o.g.\ $p_i=\min_{j=1}^n\{p_j\}$.
% For the $\ell$ slack periods, it must be the case that $\sum_{q=1}^\ell\lambda_q\leq n$ because at most one job begins per unit of time.
% Furthermore, by Lemma \ref{Lemma:Idle_Bound}, $\sum_{q=1}^{\ell}L_q\leq \frac{n}{m}\frac{m(m-1)}{2}\leq\frac{nm}{2}$.
% \begin{equation*}
% T(S^*) \geq \frac{1}{m} \sum_{j=1}^n p_j \geq \frac{n\min_{j=1}^np_j}{m} \geq r\min_{j=1}^n p_j
% \end{equation*}
% Hence, $\min_{j=1}^n p_j\leq \frac{1}{r}T(S^*)$.
\begin{subequations}
\begin{align}
T & = s_n + p_n \label{Eq:Long_Last_Job} 
 = \frac{1}{m}\left( \sum_{t=1}^{s_n}|\mathcal{A}_t| + \sum_{t=1}^{s_n} \left(m-|\mathcal{A}_t| \right) \right) + p_n %\min_{j=1}^n\{p_j\}
%\label{Eq:Long_Alive} 
\\
& \leq \frac{1}{m}\left(\sum_{i=1}^np_i + \sum_{q=1}^{k'}\Lambda_q\right) + p_n % \min_{j=1}^n\{p_j\} 
\label{Eq:Long_Packing} \\
& \leq \frac{1}{m}\sum_{i=1}^np_i + \frac{n}{2} + p_n \label{Eq:Long_Gap_Bounding} \\
%\leq T(\mathcal{S}^*) + \frac{1}{2}T(\mathcal{S}^*) + \frac{1}{6}T(\mathcal{S}^*)
& \leq \frac{5}{3} T^*. \label{Eq:Long_Bounding}
\end{align}
\end{subequations}
\end{proof}

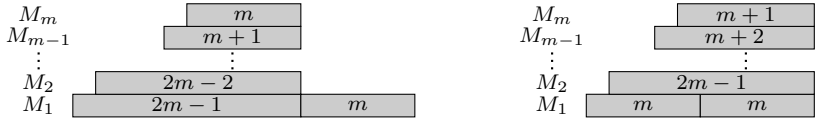
\begin{figure}[t]
\begin{subfigure}[t]{0.5\textwidth}
\begin{center}
\input{43_tightness_lpt.tex}
\end{center}
\caption{LPT schedule of makespan $T=3m-1$.}
\end{subfigure}
\begin{subfigure}[t]{0.5\textwidth}
\begin{center}
\input{43_tightness_opt.tex}
\end{center}
\caption{Optimal schedule of makespan $T^*=2m$.}
\end{subfigure}
\caption{\change{BJSP \emph{long instance} with a 3/2 lower bound of LPT.}}
\label{Figure:LPT_Tightness_Long}
\end{figure}

We complement Theorem~\ref{Thm:LPT_Long} with a long instance $I=\langle m,\mathcal{J}\rangle$ where LPT is 3/2-approximate and leave closing the gap between the two as an open question.
Instance $I$ \change{ is illustrated in Figure~\ref{Figure:LPT_Tightness_Long} and} contains $m+1$ jobs, where $p_j=2m-j$, for $j\in\{1,\ldots,m\}$, and $p_{m+1}=m$.
In the LPT schedule $\mathcal{S}$, job $j$ is executed at time $s_j=j$, for $j\in\{1,\ldots,m\}$, and $s_{m+1}=2m-1$.
Hence, $T=3m-1$.
But an optimal schedule $\mathcal{S}^*$ assigns job $j$ to machine $j+1$ at time $s_j=j+1$, for $j\in\{1,\ldots,m-1\}$. 
Moreover, jobs $m$ and $m+1$ are assigned to machine $1$ beginning at times $s_m=1$ and $s_{m+1}=m$, respectively. 
Clearly, $\mathcal{S}^*$ has makespan $T^*=2m$. 

% \begin{lemma}
% Let $\ell\in\mathcal{J}$ be the job completing last in the LPT schedule.
% Then, it holds that $p_j\leq T^*/2$.
% \end{lemma}
% \begin{proof}
% Suppose that $p_{\ell}>T^*/2$.
% Then, we claim that $\ell\leq m$.
% If $\ell>m$, then $j$ must begin in the same machine with another job $p_j\geq p_{\ell}$.
% However, for each instance $\langle m,\mathcal{J}\rangle$ such that $|\mathcal{J}|\leq m$, LPT produces obviously an optimal schedule.
% \end{proof}

% \begin{lemma}
% Suppose that 
% \end{lemma}
% \begin{proof}
% \end{proof}

% \paragraph{Short Instances}

Theorem \ref{Lem:LPT_Short} completes this section with a simple argument on the LPT performance for short instances.

\begin{theorem}
\label{Lem:LPT_Short}
LPT is optimal for short instances.
\end{theorem}
\begin{proof}
Let $p_1\geq\ldots\geq p_n$ be the LPT job ordering and suppose that job $\ell$ completes last in LPT schedule $\mathcal{S}$. 
%i.e.\ $C_{\ell}=T(\mathcal{S})$.
We claim that job $\ell$ begins at time slot $s_{\ell}=\lceil\ell/g\rceil$ in $\mathcal{S}$.
Due to the BJSP constraint, we have $s_{\ell}\geq \lceil\ell/g\rceil$.
Assume that the last inequality is strict.
Then, $|\mathcal{A}_t|=m$ for some time slot $t<s_{\ell}$.
So, there exists job $j\in \mathcal{A}_t$ with $s_j=t-m+1$.
Since $p_j\leq m-1$, we get $C_j<t$ which contradicts $j\in\mathcal{A}_t$.
Because $T^*\geq\lceil j/g\rceil+j$ for each $j\in\mathcal{J}$, the theorem follows.
\end{proof}

\subsection{Shortest Processing Time First}
\label{Section:SPT}

This section investigates the performance of \emph{Long Job Shortest Processing Time First Algorithm (LSPT)}.
LSPT orders the jobs as follows: (i) each long job precedes every short job, (ii) long jobs are sorted according to \emph{Shortest Processing Time First (SPT)}, and (iii) short jobs are sorted as in LPT.
LSPT schedules jobs greedily, in the same vein as LPT, with this new job ordering.
For long instances, when the largest processing time $p_{\max}$ is relatively small compared to the average machine load, Theorem~\ref{Thm:SPT} shows that LSPT achieves an approximation ratio better than the 5/3, i.e.\ the approximation ratio Theorem~\ref{Thm:LPT_Long} obtains for LPT.
From a worst-case analysis viewpoint, the main difference between LSPT and LPT is that the former requires significantly lower idle machine time until the last job start, but at the price of much higher difference between the last job completion times in different machines.

\begin{theorem}
\label{Thm:SPT}
LSPT is 2-approximate for minimizing makespan.
For long instances, LSPT is $(1+\min\{1,1/\alpha\})$-approximate, where $\alpha=(\frac{1}{m}\sum_{j\in\mathcal{J}}p_j)/p_{\max}$.
\end{theorem}
\begin{proof}
Let $\mathcal{S}$ be the LSPT schedule and suppose that it attains makespan $T$.
Moreover, denote by $\ell\in\mathcal{S}$ the job completing last, i.e.\ $C_{\ell}=T$.
Similarly to the Lemma~\ref{Lemma:Initial_Idle} proof, $\mathcal{S}$ is compact.
% %Suppose that $|\mathcal{J}^L|\geq m$ and let $i\in\mathcal{J}$ be the job completing last.
% Thus, if $\ell\in\mathcal{J}^S$, by arguing along the same lines with the Lemma~\ref{Thm:Naive_LPT} proof for analyzing LPT, schedule $\mathcal{S}$ has makespan $T\leq \frac{1}{n}\sum_{j=1}^np_j+n+p_{\min}$. 
% %similarly to the LPT analysis, because $\mathcal{S}$ is compact.
% Next, consider that $\ell\in\mathcal{J}^L$.
We distinguish two cases based on whether $|\mathcal{J}^L|< m$, or $|\mathcal{J}^L|\geq m$.
In the former case, every job $j\in\mathcal{J}^L$ satisfies $s_j\leq p_{\max}$.
Using similar reasoning to the Theorem~\ref{Thm:Naive_LPT} proof, we get $T\leq p_{\max} + n + p_{\min}$.

In latter case, let $t'=\min\{t:|\mathcal{A}_t|<m, t>m\}$.
That is, $t'$ is the earliest time slot strictly after time $t=m$ when at least one machine is idle.
We claim that $s_j<t'$ for each $j\in\mathcal{J}^L$, i.e.\ every long job begins before $t'$ and there is no idle machine time during 
$[m,t']$. %, where $r=\max_{j\in\mathcal{J}^L}\{s_j\}$ is the last long job start time.
Assume for contradiction that there is some job $j\in\mathcal{J}^L$ such that $s_j>t'$.
Since, there is an idle machine at $t'$ and long jobs remain to begin after $t'$, at least two jobs $j',j''\in\mathcal{J}^L$ complete simultaneously at $t'-1$, i.e.\ $C_{j'}=C_{j''}=t'-1$.
Because of the BJSP constraint $|\mathcal{B}_t|\leq 1$, it must by the case that $s_{j'}\neq s_{j''}$.
W.l.o.g.\ $s_{j'}< s_{j''}$.
By the SPT ordering, we also have that $p_{j'}\leq p_{j''}$.
Thus, we get the contradiction $C_{j'}<C_{j''}$.
Our claim implies that $t'\leq\frac{m(m-1)}{2}+\frac{1}{n}\sum_{j\in\mathcal{J}}p_j$.
%, and (ii) if $s_j\geq t'$, then $j\in\mathcal{J}^S$.

Next, consider two subcases based on whether $\ell\in\mathcal{J}^L$, or $\ell\in\mathcal{J}^S$.
If $\ell\in\mathcal{J}^L$, then $T\leq t'+p_{\max}$.
Otherwise, if $\ell\in\mathcal{J}^S$, then $T\leq t'+n+p_{\min}$. 
In both subcases,
\begin{equation*}
T\leq \frac{m(m-1)}{2}+\frac{1}{n}\sum_{j=1}^np_j+\max\left\{p_{\max},n+p_{\min}\right\}.
\end{equation*}
%
% Now consider the case $|\mathcal{J}^L|< m$.
% Clearly, every job $j\in\mathcal{J}^L$ satisfies $s_j\leq p_{\max}$.
% Hence, we have that $T\leq p_{\max} + n + p_{\max}$.
%
Obviously, the optimal solution satisfies:
\begin{equation*}
T^*\geq \max\left\{\frac{1}{m}\sum_{j=1}^np_j,p_{\max},n+p_{\min}\right\}.
\end{equation*}
In all cases, $T\leq 2T^*$.
For long instances, i.e.\ the case $\ell\in\mathcal{J}^L$, LSPT is $(1+\min\{1,1/\alpha\})$-approximate.
\end{proof}

\section{Parallelizing Long and Short Jobs}
\label{Section:LSM}

This section proposes the \emph{Long Short Job Mixing Algorithm (LSM)} to compute 1.985-approximate schedules for a broad family of instances, e.g.\ with at most $\lceil 5m/6\rceil$ jobs of processing time (i) $p_j>(1-\epsilon)(\sum_{j'}p_{j'})$, or (ii) $p_j>(1-\epsilon)(\max_{j'}\{j'+p_{j'}\})$ assuming non-increasing $p_{j}$'s, for sufficiently small constant $\epsilon>0$.
For degenerate instances with more than $\lceil5m/6\rceil$ jobs of processing time $p_j>T^*/2$, where $T^*$ is the optimal objective value, LSM requires constant machine augmentation to achieve an approximation ratio lower than 2.
There can be at most $m$ such jobs.
%using 1.2-machine augmentation.
In the Royal Mail application, machine augmentation \cite{davis2011survey,kalyanasundaram2000speed,phillips2002optimal} adds more delivery vans.
For simplicity, we also assume that $m=\omega(1)$, but the approximation ratio can be adapted for smaller values of $m$.
However, we require that $m\geq 7$.

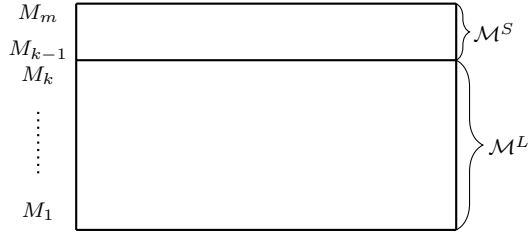
\begin{figure}[t]
\begin{center}
\input{lsm.tex}
\end{center}
\caption{\change{ Structure of schedules produced by LSM. Long jobs are prioritized in the subset $\mathcal{M}^L$ of the machines. The subset $\mathcal{M}^S$ of remaining machines execute only short jobs.}}
\label{Figure:LSM_Structure}
\end{figure}

LSM attempts to construct a feasible schedule where long jobs are executed in parallel with short jobs, \change{as depicted in Figure~\ref{Figure:LSM_Structure}}.
%short job starts occur at time slots where many long jobs are executed in parallel.
LSM uses $m^L<m$ machines for executing long jobs.
The remaining $m^S=m-m^L$ machines are reserved for short jobs. 
%Furthermore, LSM assigns short jobs once all long jobs have been scheduled.
Carefully selecting $m^L$ allows to obtain a good trade-off in terms of (i) delaying long jobs and (ii) achieving many short job starts at time slots when many long jobs are executed in parallel.
Here, we set $m^L=\lceil 5m/6\rceil$.
% %Despite the fact that strictly less than $m$ are used for long jobs, 
% Let $\mathcal{S}(m^L)$ be the LSM schedule when $m^L$ machines are dedicated to the long jobs and denote by $T(m^L)$ the corresponding makespan.
% LSM exhaustively constructs all schedules, for $m^L\in[1,m-1]$, and selects the one with the smallest makespan.
% To derive the LSM approximation ratio, we elaborate on concrete values of $m^L$ for different cases encountered in the analysis.
% LSM returns a schedule $\mathcal{S}$ of makespan $\min\{T(m^L):1\leq m^L\leq m-1\}$.
% %To select the $m^L$ value, LSM constructs a feasible schedule.
% %The remainder of the section proves an approximation ratio for LSM.
% %Theorem~\ref{Thm:LSM} proves a $(\frac{m^L}{m}+\frac{2}{3})$-approximation ratio for LSM.
% %When $m^L=\lceil3m/4\rceil$, LSM is 1.79-approximate.
Before formally presenting LSM, we modify the notions of long and short jobs by setting $\mathcal{J}^L=\{j:p_j\geq m^L, j\in\mathcal{J}\}$ and $\mathcal{J}^S=\{j:p_j<m^L,j\in\mathcal{J}\}$, %the sets of long and short jobs, 
respectively. 
%Similarly to before, the sets of long and short jobs are defined as $\mathcal{J}^L=\{j:p_j\geq m, j\in\mathcal{J}\}$ and $\mathcal{J}^S=\{j:p_j<m,j\in\mathcal{J}\}$.
Both $\mathcal{J}^L$ and $\mathcal{J}^S$ are sorted in non-increasing processing time order.
%LSM maintains two different lists for $\mathcal{J}^L$ and $\mathcal{J}^S$ which are both sorted according to LPT.
%LSM algorithm considers the partitioning $\mathcal{J}=\mathcal{J}^L\cup\mathcal{J}^S$ of jobs into the sets $\mathcal{J}^L=\{j\in\mathcal{J}:p_j\geq m\}$ and $\mathcal{J}^S=\{j\in\mathcal{J}:p_j<m\}$ of long and short jobs, respectively.
%Both sets of jobs are sorted in non-increasing order of processing times.
LSM schedules jobs greedily by traversing time slots in increasing order. %from the earliest to the latest.
Let $\mathcal{A}_t^L$ be the set of alive long jobs at time slot $t\in D$.
For $t=1,\ldots,\tau$, LSM proceeds as follows:
(i) if $|\mathcal{A}_t^L|<m^L$, then the next long job begins at $t$,
(ii) if $|J^L|=0$ or $m^L\leq|\mathcal{A}|< m$, LSM schedules the next short job to start at $t$,
else (iii) LSM considers the next time slot.
%For analysis purposes, if at some point $|J^L|=0$, then the remaining short jobs begin in the next available time slots.
%On the other hand, if it happens that $|J^S|=0$, then the long jobs are executed greedily on the $m^L$ machines and $m^S$ machines remain idle.
From a complementary viewpoint, LSM partitions the machines $\mathcal{M}$ into $\mathcal{M}^L=\{i:i\leq m^L,i\in\mathcal{M}\}$ and $\mathcal{M}^S=\{i>m^L,i\in\mathcal{M}\}$.
LSM prioritizes long jobs on machines $\mathcal{M}^L$ and assigns only short jobs to machines $\mathcal{M}^S$.
A job may undergo processing on machine $i\in\mathcal{M}^S$ only if all machines in $\mathcal{M}^L$ are busy.
The Algorithm~\ref{Alg:LSM} pseudocode describes LSM. %is a more formal description of LSM.

\begin{algorithm}
\caption{Long-Short Mixing (LSM)}
\begin{algorithmic}
\State Sort $\mathcal{J}^L=\{j\in\mathcal{J}:p_j\geq m^L\}$ in non-increasing order.
\State Sort $\mathcal{J}^S=\{j\in\mathcal{J}:p_j<m^L\}$ in non-increasing order.
\For {$t=1,\ldots,\tau$}
\If {$|\mathcal{A}_t|<m$}
\If {$|\mathcal{A}_t^L|<m^L$ and $|\mathcal{J}^L|>0$}
\State $j'=\arg\max_{j\in\mathcal{J}^L}\{p_j\}$
\State $\mathcal{J}^L=\mathcal{J}^L\setminus\{j'\}$
\ElsIf {$|\mathcal{J}^S|>0$}
\State $j'=\arg\max_{j\in\mathcal{J}^S}\{p_j\}$
\State $\mathcal{J}^S=\mathcal{J}^S\setminus\{j'\}$
\EndIf
\State In either of the above cases, set $s_{j'}=t$.
\EndIf
\EndFor
\end{algorithmic}
\label{Alg:LSM}
\end{algorithm}

Theorem \ref{Thm:LSM} shows that LSM achieves a better approximation ratio than LPT, i.e.\ strictly lower than 2, for a broad family of instances.

%Theorem \ref{Thm:LSM} shows that, with bounded machine augmentation, LSM achieves a better approximation ratio than LPT, i.e.\ strictly lower than 2.

%proves a better than 2 approximation ratio for LSM and captures the trade-off between prioritizing long or short jobs.
%Note that $m^L<m$ machines using LPT, the analysis uses arguments in the proof of Theorem~\ref{Thm:LPT_Long} for analyzing LPT in the case of long instances.

\begin{theorem}
\label{Thm:LSM}
LSM is 1.985-approximate (i) for instances with no more than $\lceil 5m/6\rceil$ jobs s.t.\ $p_j>(1-\epsilon)\max\{\frac{1}{m}\sum_{j'}p_{j'},\max_{j'}\{j'+p_{j'}\}\}$ for sufficiently small $\epsilon>0$, and (ii) for general instances using 1.2-machine augmentation.
\end{theorem}

%}

\begin{proof}
Let $\mathcal{S}$ be the LSM schedule and $\ell=\arg\max\{C_j:j\in\mathcal{J}\}$ the job completing last.
That is, $\mathcal{S}$ has makespan $T=C_{\ell}$. 
For notational convenience, given a subset $P\subseteq D$ of time slots, denote by $\lambda(P)=|P|$ and $\mu(P)=\sum_{t\in P}|\mathcal{A}_t|$ 
%and $\sigma(P)=|\{t:|\mathcal{B}_t|\geq 1,t\in P\}|$ 
the number of time slots and executed processing load, 
%and number of job starts, 
respectively, during $P$.
Furthermore, let $n^L=|\mathcal{J}^L|$ and $n^S=|\mathcal{J}^S|$ be the number of long and short jobs, respectively.
We distinguish two cases based on whether (i) $\ell\in\mathcal{J}^S$ or (ii) $\ell\in\mathcal{J}^L$.
%set $r=\max_{j\in\mathcal{J}}\{s_j\}$  equal to the last job start.
%In the former case, we loose in the length of the last job.
%In the latter case, we loose in the number of short job starts.
% The analysis proceeds similarly to the Theorem \ref{Thm:LPT_Long} analysis for the LPT performance for short instances, but requires bounding the LPT performance on a smaller number of machines.
% Similar reasoning with the Lemma~\ref{Lemma:Initial_Idle} proof implies that $\mathcal{S}$ is compact.
%
\paragraph{Case \bm{$\ell\in\mathcal{J}^S$}}
We partition time slots $\{1,\ldots,T\}$ into five subsets. %individual pieces for which we use distinct arguments.
Let $r^L=\max_{j\in\mathcal{J}^L}\{s_j\}$ and $r^S=\max_{j\in\mathcal{J}^S}\{s_j\}$ be the maximum long and short job start time, respectively, in $\mathcal{S}$.
Since $\ell\in \mathcal{J}^S$, it holds that $r^L<r^S$.
%Denote by $\mathcal{A}_t^L$ the subset of long jobs running at time slot $t$.
For each time slot $t\in[1,r^L]$ in $\mathcal{S}$, either $|\mathcal{A}_t^L|=m^L$ long jobs simultaneously run at $t$, or not.
In the latter case, it must be the case that $t=s_j$ for some $j\in\mathcal{J}^L$. %, i.e.\ a long job $j$ starts at $t$.
On the other hand, for each time slot $t\in[r^L+1,s_{\ell}]$, either $|\mathcal{A}_t|=m$, %and all machines are busy, 
or $t=s_j$ %and some short job $j\in \mathcal{J}^S$ begins.
for some $j\in \mathcal{J}^S$.
Finally, $[s_{\ell},T(\mathcal{S})]$ is exactly the interval during which job $\ell$ is executed.
%Figure~\ref{Fig:LSM_Time} illustrates this partitioning.
%In particular, let
Then, we may define: 
\begin{itemize}
  \item $F^L=\{t:|\mathcal{A}_t^L|=m^L\}$, 
  \item $B^L=\{t:|\mathcal{A}_t^L|<m^L,t=s_j,j\in\mathcal{J}^L\}$, 
  \item $F^S=\{t:t>r^L,|\mathcal{A}_t|=m\}$, and
  \item $B^S=\{t:t>r^L,|\mathcal{A}_t|<m,t=s_j,j\in\mathcal{J}^S\}$.
\end{itemize}
%Denote the time length, i.e.\ number of time slots, that cases (i)-(iii) occur by $\lambda,\sigma^L$, and $\sigma^S$, respectively.
%In fact, we denote by $\sigma^S$ only the time slots at which long jobs have been completed.
%Denote by $\lambda^F=|\{t:|\mathcal{B}_t|=0,t\in\mathcal{D}\}|$, $\lambda^L=|\{t:s_j=t,j\in\mathcal{J}^L,t\in\mathcal{D}\}|$, and $\lambda^S=|\{t:s_j=t,j\in\mathcal{J}^S,t\in\mathcal{D})\}|$ the corresponding length of time in $\mathcal{S}$.
Clearly, %sets $F^L$, $B^L$, $F^S$, $B^S$ are pairwise disjoint and
\begin{align}
\label{Eq:Alg_Makespan_Bound}
T \leq \lambda(F^L) + \lambda(B^L) + \lambda(F^S) + \lambda(B^S) + p_{\ell}.
\end{align}
Next, we upper bound a linear combination of $\lambda(F^L)$, $\lambda(B^S)$, and $\lambda(F^S)$ taking into account the fact that certain short jobs begin during a subset $\hat{B}^S\subseteq F^L\cup F^S$ of time slots. %certain number of short jobs begin during $[1,r^L]$.
%Let $\hat{B}^S$ be the set of time slots at which a small job starts during $[1,r^L]$.
%Clearly, $\hat{B}^S\subseteq F^L$ %. % these job starts occur during the time slots $F^L$.
By definition, $\lambda(B^S)\leq n^S-\lambda(\hat{B}^S)$.
We claim that $\lambda(\hat{B}^S)\geq(m^S/m^L)(\lambda(F^L)+\lambda(F^S))$.
For this, consider the time slots $F^L\cup F^S$ as a continuous time period by disregarding intermediate $B^L$ and $B^S$ time slots.
Partition this $F^L\cup F^S$ time period into subperiods of equal length $m^L$.
Note that no long job begins during $F^L\cup F^S$ and the machines in $\mathcal{M}^S$ may only execute small jobs in $\mathcal{S}$.
Because of the greedy nature of LSM and the fact that $p_j< m^L$ for $j\in \mathcal{J}^S$, there are at least $m^S$ short job starts in each subperiod.
Hence, our claim is true and we obtain that $\lambda(B^S)\leq n^S - (m^S/m^L)(\lambda(F^L)+\lambda(F^S))$, or equivalently:
\begin{align}
\label{Eq:Small_Starts}
m^S\lambda(F^L) + m^S\lambda(F^S) + m^L\lambda(B^S) \leq m^Ln^S.
\end{align}
%
% We claim that $\lambda(\hat{B}^S)\geq\lambda(F^L)/3$.
% For this, consider the time slots $F^L$ as a continuous time period by discarding intermediate $B^L$ time slots.
% Partition this $F^L$ time period into subperiods of equal length $m/4$.
% Because $p_j\leq 3m/4$ for $j\in J^S$, in each machine, there is at least one small job start every three subperiods.
% Hence, our claim is true and we obtain that
% \begin{align}
% \label{Eq:Small_Start_Bound}
% \lambda(B^S)\leq |J^S| - \lambda(F^L)/3.
% \end{align}
%
Subsequently, we upper bound a linear combination of $\lambda(F^L)$, $\lambda(B^L)$, and $\lambda(F^S)$ using a simple packing argument. 
The part of the LSM schedule for long jobs is exactly the LPT schedule for a long instance with $n^L$ jobs and $m^L$ machines.
If $|\mathcal{A}_{r^L}^L|<m$, we make the convention that $\mu(B^L)$ does not contain any load of jobs beginning in the maximal slack period completed at time $r^L$.
Observe that $\mu(F^L)=m^L\lambda(F^L)$ and $\mu(F^S)=m\lambda(F^S)$.
%let $\mu^L$ and $\mu^S$ be the processing load executed during the $\sigma^L$ and $\lambda^S$ time slots, respectively.
%Because $m^L$ machines are dedicated to the jobs in $\mathcal{J}^L$, 
Additionally, by Lemma~\ref{Lemma:Idle_Slots}, we get that $\mu(B^L)\geq m^L\lambda(B^L)/2$, except possibly the very last slack period.
Then, %by Lemma~\ref{Lem:Last_Slack_Period},
$\mu(F^L)+\mu(B^L)+\mu(F^S)\leq\sum_{j\in J}p_j$.
%Given that $m^L<m$,
% The load $\mu^F$ executed during full periods is clearly upper bounded by the total load $\frac{1}{m}\sum_{j\in\mathcal{J}}p_j$ after subtracting the load $\mu^L+\mu^S$ executed during the long and short job starting periods.
% Hence,
Hence, we obtain:
\begin{align}
\label{Eq:Packing}
m^L\lambda(F^L) + \frac{1}{2}m^L\lambda(B^L) + m\lambda(F^S) \leq \sum_{j\in\mathcal{J}}p_j.
\end{align}
% \begin{align}
% \label{Eq:Full_Load_1}
% \lambda(F^L) + \lambda(F^S) 
% %= \frac{\mu(F^L)}{m^L} + \frac{\mu(F^S)}{m} 
% \leq \frac{1}{m^L}\left(\sum_{j\in J}p_j - \mu(B^L) - \mu(F^S)\right) + \frac{\mu(F^S)}{m}
% \leq \frac{1}{m^L}\sum_{j\in J^L}p_j - \frac{\lambda(B^L)}{2}
% \end{align}
%
Summing Expressions (\ref{Eq:Small_Starts}) and (\ref{Eq:Packing}),
\begin{align*}
\left( m^L+m^S \right)\lambda(F^L) + \frac{1}{2}m^L\lambda(B^L) + \left( m+m^S \right)\lambda(F^S) + m^L\lambda(B^S) \leq \sum_{j\in\mathcal{J}}p_j + m^Ln^S.
\end{align*}
Because $m=m^L+m^S$, if we divide by $m$, the last expression gives:
\begin{align}
\label{Eq:Short_Bound}
\lambda(F^L) + \frac{1}{2}\left( \frac{m^L}{m^{\phantom{L}}} \right)\lambda(B^L) + \lambda(F^S) + \left( \frac{m^L}{m^{\phantom{L}}} \right)\lambda(B^S)  \leq 
\frac{1}{m}\sum_{j\in\mathcal{J}}p_j + \left( \frac{m^L}{m^{\phantom{L}}} \right)n^S.
\end{align}
We distinguish two subcases based on whether $\lambda(F^S)+\lambda(F^L)\geq 5n^S/6$ or not.
Obviously, $\lambda(B^L)\leq n^L$.
In the former subcase, Inequality (\ref{Eq:Small_Starts}) gives $\lambda(B^S)\leq(1-\frac{5m^S}{6m^L})n^S$.
Using Inequality (\ref{Eq:Short_Bound}), Expression (\ref{Eq:Alg_Makespan_Bound}) becomes:
\begin{align*}
T 
& \leq \frac{1}{m}\sum_{j\in\mathcal{J}}p_j + \left( 1-\frac{m^L}{2m} \right)n^L
+ \left[\left( \frac{m^L}{m^{\phantom{L}}} \right)+\left( 1-\frac{m^L}{m^{\phantom{L}}} \right) \left(1-\frac{5m^S}{6m^L} \right)\right]n^S + p_{\ell}.
% & \leq \frac{1}{m}\sum_{j\in\mathcal{J}}p_j + (1-\frac{m^L}{2m})n^L + 
% (1-\frac{m^S}{2m^L}+\frac{m^S}{2m})n^S
\end{align*}
For $m^L=\lceil 5m/6\rceil$, we have (i) $5/6\leq m^L/m\leq 5/6+1/m$ and (ii) $m^S/m^L\geq \frac{1/6 - 1/m}{5/6+1/m}$.
Given $m=\omega(1)$,
%, and (iii) $m^S/m\leq$. 
%Furthermore,
%because $m^S\leq m^L\leq m$,
% \begin{align*}
% T & \leq \frac{1}{m}\sum_{j\in\mathcal{J}}p_j + (1-\frac{m^L}{2m})n^L
% + \frac{1}{2}\left( 1-\frac{m^L}{m^{\phantom{L}}} \right)n^S + p_{\ell}
% \end{align*}
% For $m^L=\lceil 5m/6\rceil$, it clearly holds that $m^L\geq m^S$. 
% By using the standard bounds $a/b\leq\lfloor a/b\rfloor\leq a/b+(b-1)/b$ for each pair of integers $a,b\in\mathbb{Z}^+$,
\begin{align*}
T & \leq \frac{1}{m}\sum_{j\in\mathcal{J}}p_j + \left(1-\frac{5}{12} \right)n^L + \left[\frac{5}{6}+ \left(1-\frac{5}{6} \right) \left(1-\frac{1}{5} \right)\right]n^S + p_{\ell}.
\end{align*}
% \begin{align*}
% T(\mathcal{S}) & \leq \frac{1}{m}\sum_{j\in\mathcal{J}}p_j + \frac{m^L}{m}n^S + (1-\frac{m^L}{2m})\lambda(B^L)  
% + \left( 1-\frac{m^L}{m^{\phantom{L}}} \right)\lambda(B^S) + p_{\ell} \\
% & \leq \frac{1}{m}\sum_{j\in\mathcal{J}}p_j + \left( 1-\frac{m^L}{m^{\phantom{L}}} \right)\lambda(B^S) + \left( 1-\frac{m^L}{m^{\phantom{L}}} \right)\lambda(B^L) + p_{\ell}
% \end{align*}
Note that an optimal solution $\mathcal{S}^*$ has makespan:
\begin{align*}
T^* \geq \max\left\{ \frac{1}{m}\sum_{j\in\mathcal{J}}p_j, n^L+n^S+p_{\ell} \right\}.
\end{align*}
Because the instance is mixed with long and short jobs and we consider the case $\ell\in\mathcal{J}^S$, we have $p_{\ell}\geq T^*/2$.
Therefore, $T\leq(1+\frac{29}{30}+(\frac{29}{30})\frac{1}{2})T^*\leq 1.985 T^*$.
Now, consider the opposite subcase where $\lambda(F^L)+\lambda(F^S)\leq 5n^S/6$. 
Given $\lambda(B^L)\leq n^L$ and $\lambda(B^S)\leq n^S$, expression (\ref{Eq:Alg_Makespan_Bound}) becomes
$T \leq \frac{11}{6}(n^S+n^L+p_{\ell}) \leq 1.835\cdot T^*$.
\paragraph{Case \bm{$\ell\in\mathcal{J}^L$}}
%The analysis bears similarities with the Theorem \ref{Thm:LPT_Long} proof for analyzing LPT in the case of long instances.
Recall that $\mathcal{A}_t^L$ and $\mathcal{B}_t^L$ are the sets of long jobs which are alive and begin, respectively, at time slot $t$.
Furthermore, $r^L=\max\{s_j:j\in\mathcal{J}^L\}$ is the last long job starting time. 
Because LSM greedily uses $m^L$ machines for long jobs, either $|\mathcal{A}_t^L|=m^L$, or $|\mathcal{B}_t^L|=1$, for each $t\in[1,r^L]$.
So, we may partition time slots $\{1,\ldots,r^L\}$ into $F^L=\{t:|\mathcal{A}_t^L|=m\}$ and $B^L=\{t:|\mathcal{A}_t^L|<m,|\mathcal{B}_t^L|=1\}$ and obtain: 
%At each time slot $t\in[1,r^L]$, LSM uses $m^L$ machines  $|\mathcal{A}_t^L|<m^L$.
%Hence, either $|\mathcal{A}_t^L|\geq m^L$, or 
%Denote by $\lambda^L$ and $\sigma^L$ the corresponding amount of time.
%Then,
\begin{align*}
T \leq \lambda(F^L) + \lambda(B^L) +p_{\ell}.
\end{align*}
Because $m^L$ long jobs are executed at each time slot $t\in F^L$, 
\begin{equation*}
\lambda(F^L)\leq\frac{1}{m^L}\left[\sum_{j\in\mathcal{J}^L}p_j-\mu(B^L)\right].
\end{equation*}
% The part of the LSM schedule for long jobs is exactly the LPT schedule for a long instance with $n^L$ jobs and $m^L$ machines.
% If $|\mathcal{A}_{r^L}^L|<m$, we make the convention that $\mu(B^L)$ does not contain any load of jobs beginning in the maximal slack period completed at time $r^L$.
Then, Lemma~\ref{Lemma:Idle_Slots} implies that $\mu(B^L)\geq n^Lm^L/2$. 
%(except possibly the very last slack period).
Furthermore, $\lambda(B^L)\leq n^L$.
Therefore, by considering Lemma~\ref{Lem:Last_Slack_Period}, we obtain:
\begin{align*}
T \leq \frac{m}{m^L}\left(\frac{1}{m}\sum_{j\in\mathcal{J}}p_j\right)+\frac{1}{2}(n^L+p_{\ell})+\frac{1}{2}p_{\ell}.
\end{align*}
In the case $p_{\ell}\leq T^*/2$, since $T^*\geq n^L+p_{\ell}$, we obtain an approximation ratio of $(\frac{m}{m^L}+\frac{3}{4})\leq 1.95$, when $m^L=\lceil 5m/6\rceil$, given that $m=\omega(1)$.

Next, consider the case $p_{\ell}>T^*/2$.
% \subsection{Extending LSM to General Instances}
% \label{Section:LSM_Extension}
Let $\mathcal{J}^V=\{j:p_j>T^*/2\}$ be the set of very long jobs and $n^V=|\mathcal{J}^V|$. 
Clearly, $n^V\leq m$. 
By using resource augmentation, i.e.\ allowing LSM to use $m'=\lceil6m/5\rceil$ machines, we guarantee that LSM assigns at most one job $j\in\mathcal{J}^V$ in each machine.
The theorem follows.
% Otherwise, the optimal schedule $\mathcal{S}^*$ would have makespan strictly greater than $T^*$, i.e.\ a contradiction.
% To deal with this case, we use machine augmentation.
% In particular, we compare the LSM schedule using $m'=\lceil 6m/5\rceil$ machines to an optimal schedule with $m$ machines.
% Setting $m^L=\lceil5m'/6\rceil$ machines guarantees the approximation ratio.
\end{proof}

\paragraph{Remark}
If $\lceil5m/6\rceil<|\mathcal{J}^V|\leq m$, LSM does not achieve an approximation ratio better than 2, e.g.\ as illustrated by an instance consisting of only $\mathcal{J}^V$ jobs. 
Assigning two such jobs on the same machine is pathological.
Thus, better than 2-approximate schedules require assigning all jobs in $\mathcal{J}^V$ to different machines.

\section{Dealing with Uncertainty}
\label{Section:Uncertainty}
%\input{uncertainty}
%\newpage

\change{This section proposes a two-stage robust optimization approach for BJSP under uncertainty based on lexicographic optimization. 
We recently proposed a variant of this approach for two-stage robust $P||C_{\max}$ \cite{Letsios2018}.
Because job start times are irrevocable in the Royal Mail context, we adapt the $P||C_{\max}$ approach to BJSP, using resource augmentation.
That is, new machines are added once the uncertainty is realized. 
The robustness of a solution is measured by the level of resource augmentation, i.e.\ number of machines required for the final solution feasibility (instead of the actual makespan objective value we adopt in \cite{Letsios2018} with a different recovery strategy).
% which computes provably near-optimal solutions in Section~\ref{Section:LSM}, and %constructs 
% Lexicographic Optimation (LexOpt) \cite{Letsios2018}, for robust solutions to $P||C_{\max}$, to design a robust optimization approach for BJSP under uncertainty.
% Section~\ref{Section:Guidelines} summarizes a collection of guidelines from the Royal Mail experience shaping our investigations to account for uncertainty.
Section~\ref{Section:Uncertainty_Setting} formally describes our uncertainty setting, including the uncertainty set structure and investigated robustness measure.
Section~\ref{Section:LexOpt} presents the proposed approach for solving two-stage robust BJSP instances, i.e.\ the first and second stages.
Given a collection of schedules of makespan $\leq D$, our approach determines which are the more robust.}

\subsection{Uncertainty Setting}
\label{Section:Uncertainty_Setting}

%This section specifies the uncertainty setting of robust mail delivery scheduling.
%\paragraph{Two-Stage Robust Optimization Model}

%An instance of BJSP under uncertainty is a set $\mathcal{J}$ of jobs each one associated with a 

%Based on the Section~\ref{Section:Guidelines} guidelines, 
Figure~\ref{Fig:Two_Stage_Model} illustrates the two-stage setting for solving BJSP under uncertainty. The Figure~\ref{Fig:Two_Stage_Model} setting is most similar to the \cite{Liebchen2009} recoverable robustness setting, but also has connections to other two-stage optimization problems under uncertainty \cite{BenTal2004,Bertsimas2010,Hanasusanto2015}.
Stage 1 computes a feasible, efficient schedule $\mathcal{S}$ for an initial nominal instance $I$ of the problem with a set $\mathcal{J}$ of jobs and vector of processing times $p$.
After stage 1, there is uncertainty realization and a different, perturbed vector $\tilde{p}$ of processing times becomes known.
Stage 2 transforms $\mathcal{S}$ into a feasible, efficient solution $\tilde{\mathcal{S}}$ for the perturbed instance $\tilde{I}$ with vector $\tilde{p}$ of processing times, \change{using machine augmentation}.
Designing and analyzing a two-stage robust optimization method necessitates (i) defining the uncertainty set structure
%(iii) initial solution flexibility, 
and (ii) quantifying the solution robustness.

% %Figure~\ref{Fig:Two_Stage_Model}.
% %Figures~\ref{Fig:Two_Stage_Model}-\ref{Fig:Recovery} illustrates our setting.
% Stage 1 computes a feasible, efficient schedule $\mathcal{S}$ for an initial nominal instance $I$ of the problem. 
% Instance $I$ consists of a job set $\mathcal{J}$ with vector of processing times $\bm{p}$.
% To attain low fixed, global deadline $D$.
% Schedule $\mathcal{S}$ should be robust in satisfying $D$ while containing a low number of machines.
% After stage 1, there is uncertainty realization and a different, perturbed vector $\bm{\tilde{p}}$ of processing times becomes known.
% %Instance $\tilde{I}$ belongs to the uncertainty set $\mathcal{U}_F(I)$ of instance $I$.
% The second stage transforms $\mathcal{S}$ into a feasible, efficient solution $\tilde{\mathcal{S}}$ for the perturbed instance $\tilde{I}$ with vector of processing times $\bm{\tilde{p}}$.
% %$\tilde{I}=\langle\mathcal{J},\bm{\tilde{p}}\rangle$. 

\begin{figure}
\begin{subfigure}[t]{\textwidth}
\centering
\input{two_stage_model}
\caption{Two-Stage Model} 
\vspace*{0.5cm}
\label{Fig:Two_Stages}
\end{subfigure}
\begin{subfigure}[t]{\textwidth}
\centering
\input{rescheduling}
\caption{Rescheduling}
\label{Fig:Rescheduling}
\end{subfigure}

\caption{Uncertainty Setting. Figures obtained from \cite{Letsios2018}.}
\label{Fig:Two_Stage_Model}
\end{figure}
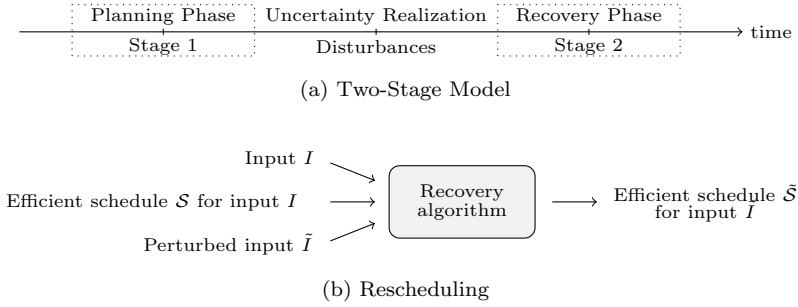

% Perturbations involve input parameters subject to uncertainty.
% {\color{gray} We consider (i) duration variations, (ii) gate capacity fluctuations, and (iii) vehicle failures.}
% Itinerary duration $d_i$ may be modified and become $\tilde{d}_i=f_id_i$, where $f_i\geq 0$ is the perturbation factor of itinerary $i\in I$.
% Similarly, for a given time length $\lambda$, the gate capacity $g^u(\lambda)$ may become $\tilde{g}^u(\lambda)=g^u(\lambda)+ a(\lambda)$.
% Finally, vehicle $k\in V$ might fail during the delivery of itinerary $i\in I$.
% However, such a failure is a complex issue because it involves (i) restoring the broken vehicle, (ii) retrieving its content, and (iii) proceeding the suspended itinerary.

Scheduling under uncertainty may involve different perturbation types. %, i.e.\ parameters subject to uncertainty.
We study processing time variations, i.e.\ $p_j$ may be perturbed by $f_j> 0$
%modified 
to become $\tilde{p}_j=f_jp_j$.
%, where $f_j> 0$ is the perturbation factor of job $j\in \mathcal{J}$.
If $f_j>1$, then job $j$ is delayed.
If $f_j<1$, job $j$ completes early.
Instance $I$ is modified by a perturbation factor $F > 1$ when $1/F\leq f_j\leq F$, for each $j\in\mathcal{J}$.
Uncertainty set $\mathcal{U}_F(I)$ contains every instance $\tilde{I}$ that can be obtained by disturbing $I$ w.r.t.\ perturbation factor $F$.

\change{Our two-stage robust optimization approach aims to achieve a low number of machines once the recovery stage 2 is completed.
%good trade-offs between the number of used machines and makespan after uncertainty realization.
% To this end, given stage 1 imposes a global deadline $D$ for upper bounding makespan and computes a robust solution with a small number of vehicles, thus resource augmentation.
% Stage 2 imposes the stage 1 job starts to avoid many deadline violations and dynamically allocates jobs to machines. Stage 2 may require more machines than the initial schedule. 
% To quantify the initial solution robustness, 
% we aim for a low \emph{price of robustness} \cite{Bertsimas2004,Bertsimas2011,Goerigk2016,NIPS2006_3053}.
%For a perturbed instance $\langle m,\tilde{\mathcal{J}}\rangle$, 
%In particular, 
Specifically, let $\mathcal{S}$ be an initial schedule, of makespan $\leq D$, for a nominal BJSP instance $I$ and $\tilde{\mathcal{S}}$ be a recovered schedule, obtained from $\mathcal{S}$ after uncertainty realization, for a perturbed instance $\tilde{I}\in\mathcal{U}_F(I)$.
Denote by $\tilde{\mathcal{S}}^*$ a feasible schedule for $\tilde{I}$ with makespan $\leq D$ and minimal number of machines.
% be the final obtained solution and the nominal optimal solution obtained with full input knowledge, respectively.
Schedule $\mathcal{S}$ is robust if the ratio $V(\tilde{\mathcal{S}})/V(\tilde{\mathcal{S}}^*)$, where $V(\cdot)$ denotes the number of machines in a given schedule for $\tilde{I}$, is as low as possible.
By slightly abusing standard robust optimization terminology, we refer to this ratio as the \emph{price of robustness} \cite{Bertsimas2004,Bertsimas2011,Goerigk2016,NIPS2006_3053}.}
%The stage 1 global deadline and stage 2 fixed decisions both avoid large makespan constraint deviations.

% The following definition formalizes the notion of robustness.
% That is, a robust solution should be near-optimal and near-feasible after uncertainty realization. 

% \begin{definition}
% Schedule is $\alpha$-approximate $\beta$-feasible if it holds that (i) $m(\mathcal{S})\leq \alpha m(\mathcal{S}^*)$, and (ii) $T(\mathcal{S})\leq\beta T$.
% \end{definition}

% The main result of the paper is an $O(1)$-approximate $O(1)$-feasible schedule.

% A solution is robust if it minimizes the number of misses.
% A more generic definition assigns a weight to each itinerary and minimizes the total weight of itineraries that are not completed.
% The weight could an estimation of the remaining load.

% \begin{itemize}
%     \item The duration of an itinerary might be different than its nominal value, i.e.\ either greater or smaller. There is the issue of how to select the nominal values: should we overestimate, or underestimate? An uncertain parameter $f_i$ specifies the perturbation factor of a duration so that $\tilde{d}_i=\phi_i\cdot d_i$. This parameter allows to incorporate either itinerary delays, or departure delays.
%     \item A vehicle might break down and then we need new ones to repair the schedule of the day. Should we wait for the one to return, or use a new one? This should depend on the delay. A binary uncertain paramerer $\beta_j$ specifies whether vehicle $j$ breaks down.
% \end{itemize}

\subsection{Two-Stage Robust Scheduling Approach}
\label{Section:LexOpt}

This section proposes a two-stage robust optimization for solving BJSP under uncertainty with: (i) an exact method producing initial solution $\mathcal{S}$ and (ii) a recovery strategy restoring $\mathcal{S}$ after instance $\tilde{I}$ is revealed. 
%This section presents the proposed approach. %our approach for computing the stage 1 solution and a recovery strategy for the stage 2 solution.

\subsubsection{Exact LexOpt Scheduling with Machine Augmentation (Stage 1)}

\change{To compute robust first-stage schedules, we develop an integer programming formulation minimizing a characteristic value, which is motivated by lexicographic optimization and the fact that machine augmentation is required in the recovery stage.}
% the notion of a
%for producing the initial solution ,
%Before stating the integer program, 
%we define a schedule's 
%a characteristic value.}

\change{Recent work shows that two-stage robust $P||C_{\max}$ schedules can be obtained by lexicographically minimizing the machine completion times $T_1\geq\ldots\geq T_m$, i.e.\ $\text{lex}\min\{T_1,\ldots,T_m\}$, where $T_i$ corresponds to the $i$-th greatest machine completion time \cite{Letsios2018}.
That is, we minimize $T_1$ and, among all schedules minimizing $T_1$, we select a schedule minimizing $T_2$, then $T_3$ etc.
Here, the proposed two-stage robust optimization approach lexicographically minimizes the job completion times $C_1\geq\ldots\geq C_n$, i.e.\ $\text{lex}\min\{C_1,\ldots,C_n\}$, where $C_j$ refers to the $j$-th greatest completion time.
%Note that $T_i$ and $C_j$ refer to the $i$-th greatest machine completion time and $j$-th grea
By considering all jobs instead of only the ones completing last in each machine, we enforce robustness at the price of extra computational effort.
%, because the number of objectives increases.
In particular, we are faced with a multi-objective optimization problem, where the number of objectives is $n>m$.
Based on standard weighting lexicographic optimization methods, this problem can be reformulated as the mono-objective problem $\min\{\sum_{j=1}^nB^{C_j}\}$, where $B>1$ is a sufficiently large scalar \cite{Sherali1982}. 
%Section \ref{Section:Robustness_Assessment} 
We empirically select $B=2$.
% in multi-objective lexicographic optimization, the $\text{lex}\min\{C_1,\ldots,C_n\}$ problem can be heuristically solved via the monobjective optimization problem $\sum_{j=1}^n2^{C_j}$.

To achieve low resource augmentation at the stage 2 schedule $\tilde{\mathcal{S}}$, we incorporate the number $V$ of machines in the characteristic value for obtaining the stage 1 schedule $\mathcal{S}$.
Because the job starts are not modified in stage 2, if a minimal of machines is used in $\mathcal{S}$, many new simultaneous job executions may occur in the stage 2 schedule $\tilde{\mathcal{S}}$ after uncertainty realization, due to job delays.
%So, the number of machines might be high in the final solution $\tilde{\mathcal{S}}$.
On the other hand, if a large number of machines is used in $\mathcal{S}$, a small number of new job overlaps will occur in the stage 2 schedule $\tilde{\mathcal{S}}$, but the number of machines in the final schedule is already high.
An empirically chosen parameter $\theta>0$ specifies the contribution of $V$ in the characteristic value.

% Minimizing $W(\mathcal{S})=\sum_{j\in\mathcal{J}}w_j(\mathcal{S})$ lexicographically minimizes the sum of job completion times and produces robust schedules in terms of minimizing makespan \cite{Letsios2018}.
% Minimizing the characteristic value generates a robust schedule with small machine augmentation if the parameter $\theta$ is carefully selected. 
}

Denote by $V(\mathcal{S})$ be the number of machines in $\mathcal{S}$.
In addition, associate the weight $w_j(\mathcal{S})=2^{C_j(\mathcal{S})}$ to each job $j\in\mathcal{J}$, where $C_j(\mathcal{S})$ is the job $j$ completion time, and let 
$W(\mathcal{S})=\sum_{j\in\mathcal{J}}w_j(\mathcal{S})$ be the sum of job weights in $\mathcal{S}$.
% Given a schedule $\mathcal{S}$ with a number $V(\mathcal{S})$ of used machines, .
% Then, $W(\mathcal{S})=\sum_{j\in\mathcal{J}}w_j(\mathcal{S})$ is the sum of job weights in schedule $\mathcal{S}$
%For sufficienlty large $M$, minimizing $w(\mathcal{S})$ lexicographically minimizes the job completion times.
The characteristic value $F(\mathcal{S})$ of schedule $\mathcal{S}$ is the weighted sum:
\begin{equation*} 
F(\mathcal{S}) \triangleq V(\mathcal{S}) + \theta\cdot W(\mathcal{S}), 
\end{equation*}
where $\theta>0$ is a parameter specifying the relevant importance between $V(\mathcal{S})$ and $W(\mathcal{S})$.
Section~\ref{Section:Numerical_Results} selects the $\theta$ value empirically.
A schedule of minimal characteristic value can be computed with integer programming formulation (\ref{Eq:LexOpt_Model}).
Variable $v$ corresponds to the number of machines and parameter $w_{j,t}=2^t$ is the weight of job $j$ if it completes at time slot $t$.
%
% Insights revealed by our approximation analysis in the setting under uncertainty:
% \begin{itemize}
%   \item If in the optimal solution, many job starts occur during full periods, then this solution might be unstable because delays might cause may jobs to simultaneously begin.
%   \item An small additional number of machines allows to obtain a significantly better schedule.
% \end{itemize}
%
% There is a trade-off between the number of used vehicles and the solution robustness.
% If a job is significantly delayed, then an important number of other jobs might have to be implemented by a different machine than the originally planned one.
% If a job is slightly delayed, then delays might be propagated due to the gate capacity constraint.
%
{\allowdisplaybreaks
\begin{subequations}
\label{Eq:LexOpt_Model}
\begin{align}
\min_{x_{j,s},\;v,w} \quad & v+\theta\left(\sum_{j\in\mathcal{J}}\sum_{s\in\mathcal{F}_j}x_{j,s}w_{j,s+p_j}\right) \label{Eq:LO_Objective} \\
& v \geq x_{j,s} & j\in\mathcal{J}, s\in D \label{Eq:LO_Machines} \\
& x_{j,s}(s+p_j)\leq D & j\in\mathcal{J}, s\in D \label{Eq:LO_Deadline} \\
& \sum_{j\in \mathcal{J}} \sum_{\substack{s\in A_{j,t}}} x_{j,s} \leq v & t\in \mathcal{D} \label{Eq:LO_Vehicles_Count} \\
& \sum_{s\in F_j} x_{j,s} = 1 & j\in \mathcal{J} \label{Eq:LO_Itinerary_Assignment} \\
& \sum_{j\in \mathcal{J}_s} x_{j,s} \leq g & s\in D \label{Eq:LO_Throughput} \\ 
& x_{j,s}\in\{0,1\} & j\in \mathcal{J}, s\in F_j \label{Eq:LO_Integrality}
\end{align}
\end{subequations}}
\noindent
Expression (\ref{Eq:LO_Objective}) minimizes the characteristic value.
Constraints (\ref{Eq:LO_Machines}) limit the active machines to the total number of machines.
Constraints (\ref{Eq:LO_Deadline}) force all jobs to complete before the makespan $D$.
Constraints (\ref{Eq:LO_Vehicles_Count}) ensure that at most $v$ machines are used at each time slot $t$. 
Constraints (\ref{Eq:LO_Itinerary_Assignment}) require that each job $j$ is scheduled.
Constraints (\ref{Eq:LO_Throughput}) express the BJSP constraint.

\change{Large exponents provoke numerical issues when solving Integer Program (\ref{Eq:LexOpt_Model}).
To circumvent this issue, we
%We may 
reduce the number of objectives in the underlying lexicographic optimization problem by rounding job completion times.
% to avoid numerical instabilities.
%To avoid many lexicographic objectives, round the completion times.
In particular, we divide the time horizon into a set of $\ell$ time periods.
A job $j$ completing at time period $q=1,\ldots,\ell$ has weight $w_j=2^q$.
By minimizing $\sum_{j\in \mathcal{J}}w_j$, we compute near-lexicographically optimal schedules.}
%which ar the job completion times.}
%In this way, we ensure that not many jobs deviate from the global deadline.}

\subsubsection{Rescheduling Strategy (Stage 2)}

%\paragraph{Recovery strategy}
A rescheduling strategy  %specifies the actions for 
transforms an initial schedule $\mathcal{S}$ for the nominal problem instance $I$ into a final schedule $\tilde{\mathcal{S}}$ for the perturbed instance $\tilde{I}$.
To satisfy the requirement that schedule $\tilde{\mathcal{S}}$ should stay close to $\mathcal{S}$, we distinguish between binding and free optimization decisions similarly to \cite{Letsios2018}.
Let $x_{j,s}$ and $y_{i,j}$ be binary variables specifying whether job $j\in\mathcal{J}$ begins at time slot $t\in\mathcal{D}$ and is executed by machine $i\in\mathcal{M}$, respectively.
\change{Based on Royal Mail practices,} we consider rescheduling with restricted job start times and flexible job-to-machine assignments.
Definition~\ref{Def:Rescheduling} formalizes this fact with binding and free optimization decisions.

\begin{definition}
\label{Def:Rescheduling}
Let $\mathcal{S}$ be the initial schedule in BJSP under uncertainty.
\begin{itemize}
  \item \textbf{Binding decisions} $\{x_{j,t}:j\in\mathcal{J},t\in\mathcal{F}_j\}$ are variable evaluations determined from 
    $\mathcal{S}$ in the rescheduling process.
  \item \textbf{Free decisions} $\{y_{i,j}:i\in\mathcal{M},j\in\mathcal{J}\}$ are variable evaluations not determined from $\mathcal{S}$ but needed to recover feasibility.
\end{itemize}
\end{definition}

% \emph{Binding decisions} %$\{x_{j,s}:x_{j,s}(\mathcal{S})=1,j\in\mathcal{S},s\in\mathcal{F}_j\}$ 
% are variable evaluations determined from the initial solution after uncertainty realization.
% \emph{Free decisions} %$\{y_{i,j}:i\in\mathcal{M},j\in\mathcal{J}\}$ 
% are not determined from the initial solution, but are essential to ensure feasibility.
%Here, job start times, i.e.\ $x_{j,s}$, are binding decisions while job-to-machine assignments, i.e.\ $y_{i,j}$, are free decisions.
Enforcing binding decisions ensures a limited number of initially planned solution modifications.
Note that first-stage decisions remain critical in this context.
%This separation ensures that we stay close to the initially planned solution and first-stage decisions remain critical.
%In the proposed recovery strategy, job start times are binding decisions, while job-to-machine assignments are free decisions.
%We consider a simple recovery strategy with restrictive job start times and flexible job-to-machine assignments. 
%Following the \citet{Letsios2018} terminology, the former are binding decisions while the latter are free decisions.
%This approach is typically adopted in practice.
%and imposes restrictions due to modification costs.
%Algorithm~\ref{Alg:Rescheduling_Strategy} formally describes the proposed rescheduling strategy that maintains all binding decisions.
%In particular, 
The proposed recovery strategy sets $x_{j,s}(\mathcal{S})=x_{j,s}(\tilde{\mathcal{S}})$.
On the other hand, job-to-machine assignments are decided in an online manner.
For $t=1,\ldots,\tau$, the jobs $\{j:x_{j,t}(\mathcal{S})=1\}$ are assigned to the lowest-indexed available machines.
A machine is available at $t$ if it executes no jobs.
This assignment derives the $y_{i,j}(\tilde{\mathcal{S}})$ values.

\section{Numerical Results}
\label{Section:Numerical_Results}
%\input{numerical_results}
%\newpage

\change{This section computationally evaluates the proposed heuristics and robust optimization approach for BJSP with perfect knowledge and under uncertainty, respectively, using Royal Mail data.}
Section \ref{Section:Instance_Generation} discusses %our data analysis steps for deriving the 
the derivation of Royal Mail BJSP instances and historical schedules.
\change{Section \ref{Section:Long_Short_Jobs} presents information about the number and load of long and short jobs in these instances.
Section \ref{Section:Evaluation_Heuristics} compares the proposed LPT, LSPT and LSM heuristics.
}
Section \ref{Section:Evaluation_Historical_Schedules} evaluates \change{the historical schedules sensitivity with respect to the number of machines, i.e.\ Royal Mail vans.}
Finally, Section \ref{Section:Robustness_Assessment} presents a numerical assessment of the two-stage robust optimization approach. 
%{\color{red} 
We run all computations on a 2.5GHz Intel Core i7 processor with an 8 GB RAM memory running macOS Mojave 10.14.6.
Our implementations use Python 3.6.8, Pyomo 5.6.1, and solve the integer programming instances with CPLEX 12.8. A recent MEng thesis considers several of these contexts in greater detail \cite{suraj-g}.
%}

%\subsection{Data Analysis}
\subsection{Generation of Benchmark Instances and Historical Schedules}
\label{Section:Instance_Generation}

We use historical data from three Royal Mail delivery offices which we refer to as (i) \basingstoke, (ii) \colchester, and (iii) \northampton. 
Part of the data is encrypted for confidentiality protection.
We consider a continuous time period of 78, 111, and 111 working days for \basingstoke, \colchester, and \northampton, respectively.
A BJSP instance is the set of all jobs performed by a single delivery office during one date.
So, we examine 300 instances in total.
A job corresponds to a delivery in a set of neighboring postal codes.
The data is a list of jobs, each specified by
%, i.e.\ itineraries, performed during a specific time period.
%One job is associated with 
a (i) unique identifier, (i) delivery office, (iii) date, (iv) vehicle plate number, (v) begin time, and (v) completion time.
Below, we give more details for generating the benchmark instances and the actual schedules realized by Royal Mail. % from the historical data.
% All jobs performed on one date at one delivery office specify a BJSP instance and the actual realized schedule.
% The current section discusses details in generating these instances.

% We use mail delivery data from (i) Basingstoke DO during 28 July 2017 and 30 September 2017, (ii) Colchester DO during 01 July 2017 and 12 November 2017, and (iii) Northampton during 01 July 2017 and 12 November 2017. 

A BJSP instance is defined by a (i) time horizon, (ii) time discretization, (iii) number of available vehicles, (iv) set of jobs, and (v) BJSP parameter. 
%To generate a BJSP instance, we determine the (i) number of available vehicles, (ii) set of jobs, (iii) earliest start time, (iv) last completion time, and (v) gate capacity.
A simple data inspection shows that, among all jobs, 92\% run during [06:00,19:00] in \basingstoke, 91\% are executed during [05:00,19:30] in \colchester, and 93\% are implemented during [05:30,19:30] in \northampton.
A scatter plot illustration clearly demonstrates that these boundaries specify the time horizon for each delivery office after dropping outliers \cite{source_code}.
%This observation motivates the design of robust optimization methods.
The time horizon boundaries might be violated by both the historical schedules and our two-stage robust optimization method. 
We use a time discretization of $\delta=15$ minutes.
The number of available vehicles is the number of distinct plate numbers used by each delivery office.
% In \basingstoke, 97\% of the itineraries begin during 06:00 and complete no later than 19:00.
% In \colchester, 96\% of the itineraries begin during 05:00 and complete no later than 19:00.
% In \northampton, 95\% of the itineraries begin during 05:30 and complete no later than 19:00.
% Even the historical schedules have infeasbilities in terms of number of machines and makespan.
% The number of available vehicles is the number of distinct plate numbers used during the day.
% The set of jobs contains all distinct itineraries performed during the day.
We set the processing time of job $j\in \mathcal{J}$ equal to $p_j=\lceil(C_j-s_j)/\delta\rceil$, where $s_j$ and $C_j$ is start and completion time of $j$ in the corresponding historical schedule.
%The durations are rounded up to the closest integer multiple of $\delta$.
The minimal processing time is 30 minutes, but a job may last for a number of hours.
A scatter plot illustration shows that the distribution of processing times follows a similar pattern on a weekly basis \change{\cite{Chassein2019,source_code}}.
This observation supports using robust optimization to deal with the Royal Mail BJSP instances under uncertainty.
%The duration values follow the same pattern during each day.
%Figure illustrates the distribution of itinerary durations.
%This observation motivates the usage of robust optimization to deal with the mail delivery problem under interval uncertainty.
%The time horizon beginning is set equal to the earliest start time after dropping outlier points.
%The begin time bound $b^u$ is set equal to the time at X\% of itineraries have already departed.
%Similarly, the global deadline is equal to the earliest time at which the greatest part of itineraries have been completed.
%We manually set the discretization length $\delta=15$ minutes.
BJSP parameter $g$ is set equal to the maximum number of jobs beginning in a time interval $\delta$ units of time after ignoring few outliers.
%, where $\delta$ is the 
%Similarly to the time horizon length case, we set the delivery office gate capacity equal to the maximal number of vehicles, i.e.\ machines, departing per time slot by ignoring outliers. 
%Using this information we generate X instances for \basingstoke, Y instances for \colchester, and Z instances for \northampton. 
%{\color{red} For presentation purposes, we would like to also store the exact day (e.g.\ Monday) and time (e.g.\ 1pm) that events occur. } 
%The effect can be clearly illustrated using scatter plots, but we omit them.

The Royal Mail data include a historical schedule for each BJSP instance.
Such a schedule is associated with (i) job start times, (ii) makespan, and (iii) number of used machines.
Begin times are rounded down to the closest multiple of $\delta$ for time discretization.
After rounding, the makespan is the time at which the last job completes.
The number of vehicles is the maximal number of jobs running simultaneously.
We note that solutions in this form do not explicitly specify job-to-machine assignments.
However, once the job start times are known, feasible assignments can be computed with simple interval scheduling algorithms \cite{Kolen2007}.

\change{

\subsection{Long and Short Jobs}
\label{Section:Long_Short_Jobs}

\begin{figure}[t]
    \begin{subfigure}[t]{0.5\textwidth}
        \begin{center}
        \includegraphics[width=\textwidth]{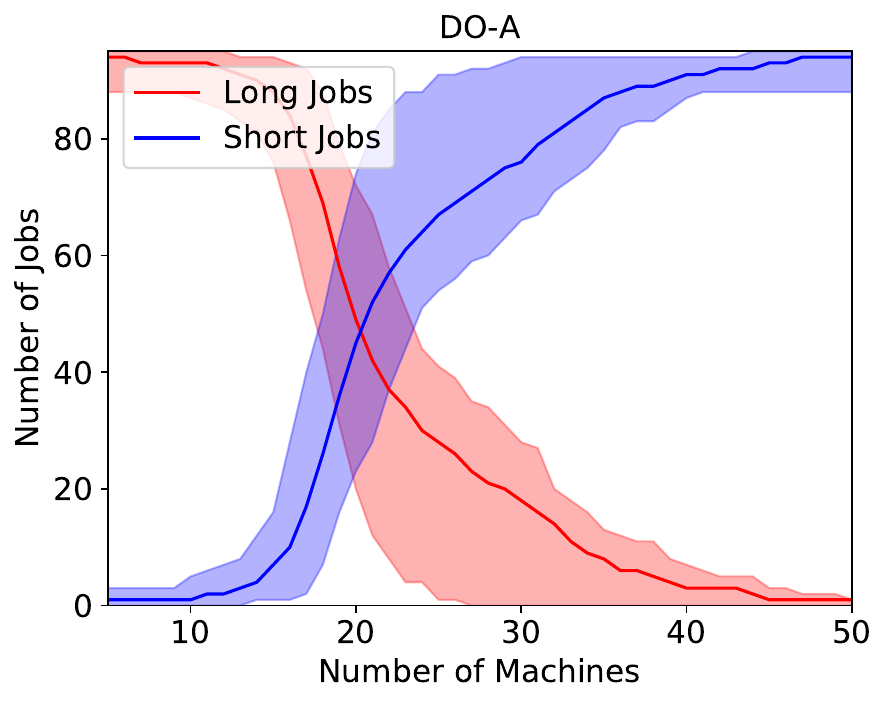}
        \end{center}
        \caption{\basingstoke}
        %\label{Figure:basingstoke_return_times_histogram}
    \end{subfigure}
    \begin{subfigure}[t]{0.5\textwidth}
        \begin{center}
        \includegraphics[width=\textwidth]{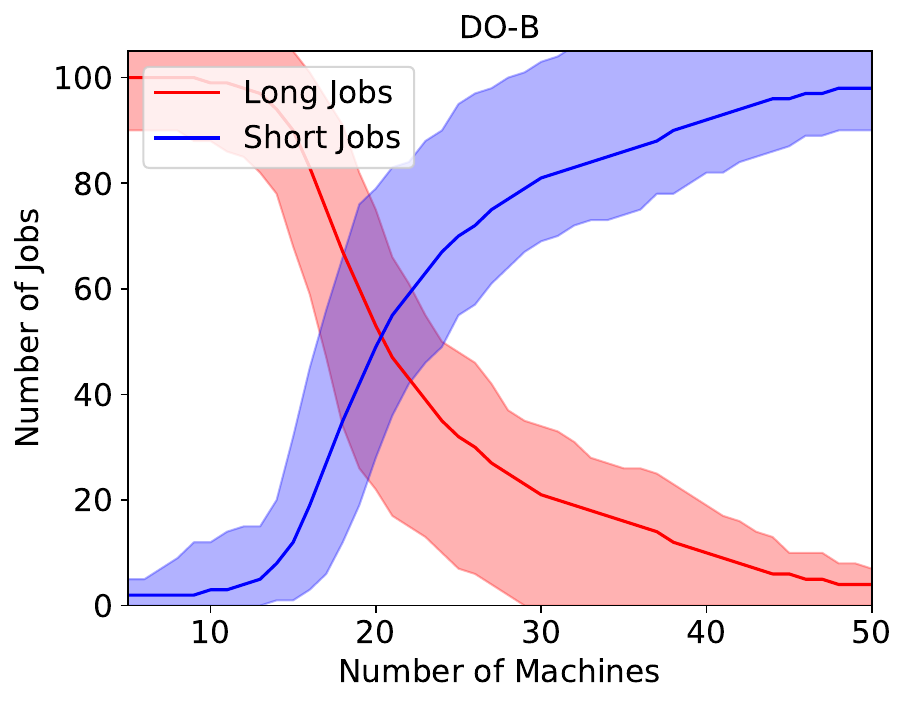}
        \end{center}
        \caption{\colchester}
        %\label{Figure:colchester_durations_histogram}
    \end{subfigure}
    \begin{subfigure}[t]{\textwidth}
        \begin{center}
        \includegraphics[width=0.5\textwidth]{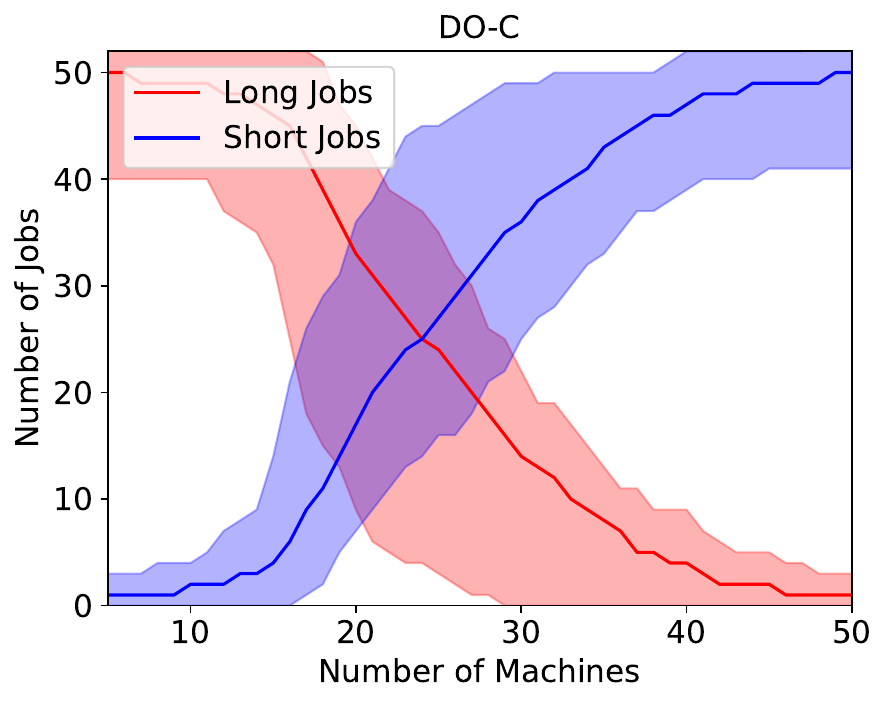}
        \end{center}
        \caption{\northampton}
        %\label{Figure:northampton_durations_histogram}
    \end{subfigure}
\caption{\change{Line chart comparing the cardinality of long and short jobs. Given a number $m$ machines, (i) the solid line plots the average cardinality and (ii) the shaded area shows the difference between the maximum and minimum cardinality, with respect to all job sets (days).}
%\todo[inline]{Can't read the figure fonts}
%schedules realized by Royal Mail and nominal schedules produced by solving MILP (\ref{Eq:DT_Model}) with CPLEX. Dotted vertical lines indicate Mondays.
}
\label{Figure:Long_Short_Cardinality}
\end{figure}

This section presents information about the number and processing load of long and short jobs in the instances of each delivery office.
Since the long-short job separation depends on the number $m$ of machines, we examine a range of $m$ values.
A job set, i.e.\ the jobs executed by a delivery office in one day, is solved for every $m\in[5,50]$.
% Recall that the data set contains the jobs executed by \basingstoke, \colchester, and \northampton\ during a period of 78, 111, and 111 days  respectively.
Let $K$ be the number of examined days for a delivery office, e.g.\ \basingstoke.
Moreover, denote by $N$ the number of all completed jobs during these $K$ days and by $(P_1,\ldots,P_N)$ the corresponding vector of processing times.
%The discussion about long and short jobs is based on these parameters.
%Using this information, we compute the average number of long and short jobs per day, for each delivery office.
%By definition, the classification of a job as long or short depends on the number of available machines. 
Then, $N^L(m)=|\{j:P_j\geq m, 1\leq j\leq N\}|$ and $N^S(m)=|\{j:P_j<m, 1\leq j\leq N\}|$ is the total number of long and short jobs, respectively, during all days, assuming $m$ machines. 
In addition, $\Lambda^L(m)=\frac{1}{m}\sum_{j:P_j\geq m}P_j$ and $\Lambda^S(m)=\frac{1}{m}\sum_{j:P_j< m}P_j$ is the mean load of long and short jobs, respectively, with $m$ machines.

Figure~\ref{Figure:Long_Short_Cardinality} plots the averaged number $\frac{N^L(m)}{K}$ and $\frac{N^S(m)}{K}$ of long and short jobs per day, with respect to $m$. 
% Figure~\ref{Figure:Long_Short_Cardinality} also shows the the deviation of the maximum and minimum number of long and short jobs from their mean, for each $m$.
% Similarly, let $\Lambda^L(m)=\frac{1}{m}\sum_{j:P_j\geq m}P_j$ and $\Lambda^S(m)=\frac{1}{m}\sum_{j:P_j< m}P_j$ be the mean load of long and short jobs, respectively, with $m$ machines.
Similarly, Figure~\ref{Figure:Long_Short_Load} illustrates the averaged mean loads $\frac{\Lambda^L(m)}{K}$ and $\frac{\Lambda^S(m)}{K}$ per day, with respect to $m$.
% We check the range $m\in[2,n_{\min}]$, where $n_{\min}$ is the minimum number of jobs that may appear during one day in the delivery office.
% We would also like to get an indication of the average load $\Lambda^L/K$ and $\Lambda^S/K$ of the long and short jobs, where $\Lambda^L$ and $\Lambda^S$ are obtained by summing all processing times $(P_1,\ldots,P_N)$.
Clearly, when $m$ increases, $\frac{N^L(m)}{K}$ and $\frac{P^L(m)}{K}$ decrease, while $\frac{N^S(m)}{K}$ and $\frac{P^S(m)}{K}$ increase.
% This fact is verified in Figures~\ref{Figure:Long_Short_Cardinality}-\ref{Figure:Long_Short_Load} for each delivery office.
Section~\ref{Section:LPT} implies that the most difficult instances for LPT arise when the averaged mean load $\frac{\Lambda^L(m)}{K}$ of long jobs per day tends to become equal to the averaged number $\frac{N^S(m)}{K}$ of short jobs per day.
Figures~\ref{Figure:Long_Short_Cardinality}-\ref{Figure:Long_Short_Load} show that this pathological situation occurs when $m$ belongs to $[18,25]$, $[20,27]$ and $[23,30]$ for \basingstoke, \colchester\ and \northampton, respectively.

% The analysis in Section~\ref{Section:LPT} implies that the approximation ratio of LPT is strictly better than two when the average mean load $\frac{\Lambda^L(m)}{K}$ of long jobs is significantly different than the average number $\frac{N^S(m)}{K}$ of short jobs, for each delivery office.
% Worst-case instances for LPT arise when these quantities tend to become roughly equal.
% This pathological situation occurs when $m$ belongs to $[18,25]$, $[20,27]$ and $[23,30]$ for \basingstoke, \colchester and \northampton, respectively.
% For instances with values of $m$ in these ranges, we observe that LSM produces better solutions than LPT.
}

\begin{figure}[t]
    \begin{subfigure}[t]{0.5\textwidth}
        \begin{center}
        \includegraphics[width=\textwidth]{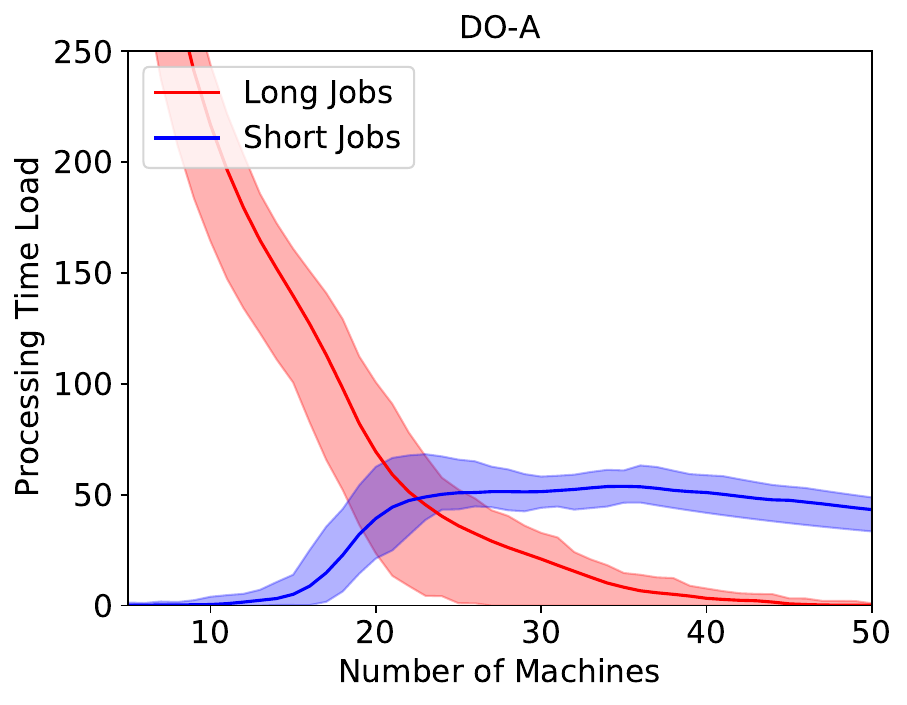}
        \end{center}
        \caption{\basingstoke}
        %\label{Figure:basingstoke_return_times_histogram}
    \end{subfigure}
    \begin{subfigure}[t]{0.5\textwidth}
        \begin{center}
        \includegraphics[width=\textwidth]{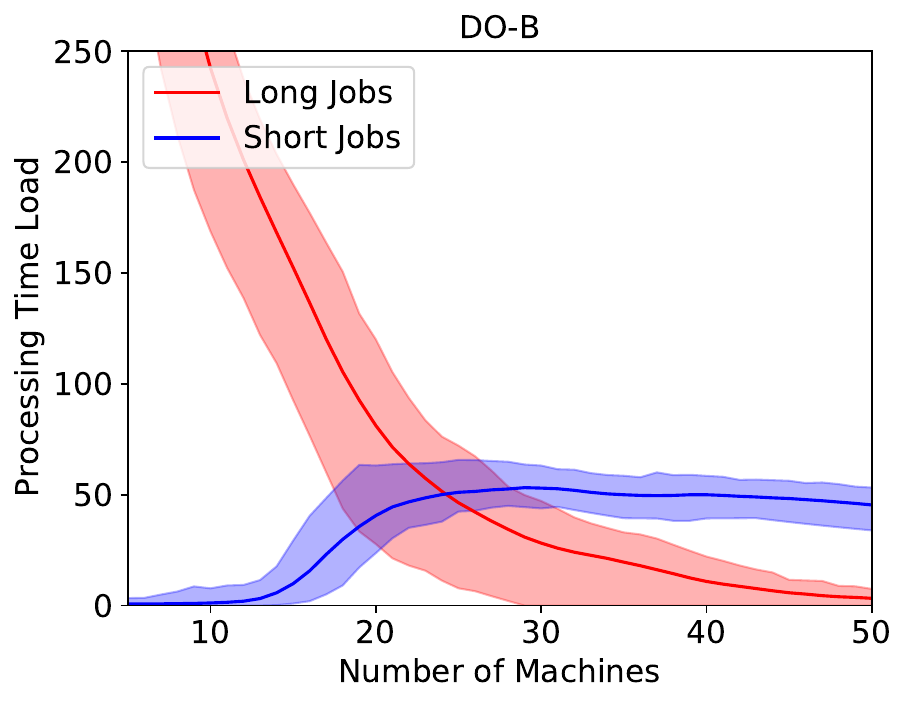}
        \end{center}
        \caption{\colchester}
        %\label{Figure:colchester_durations_histogram}
    \end{subfigure}
    \begin{subfigure}[t]{\textwidth}
        \begin{center}
        \includegraphics[width=0.5\textwidth]{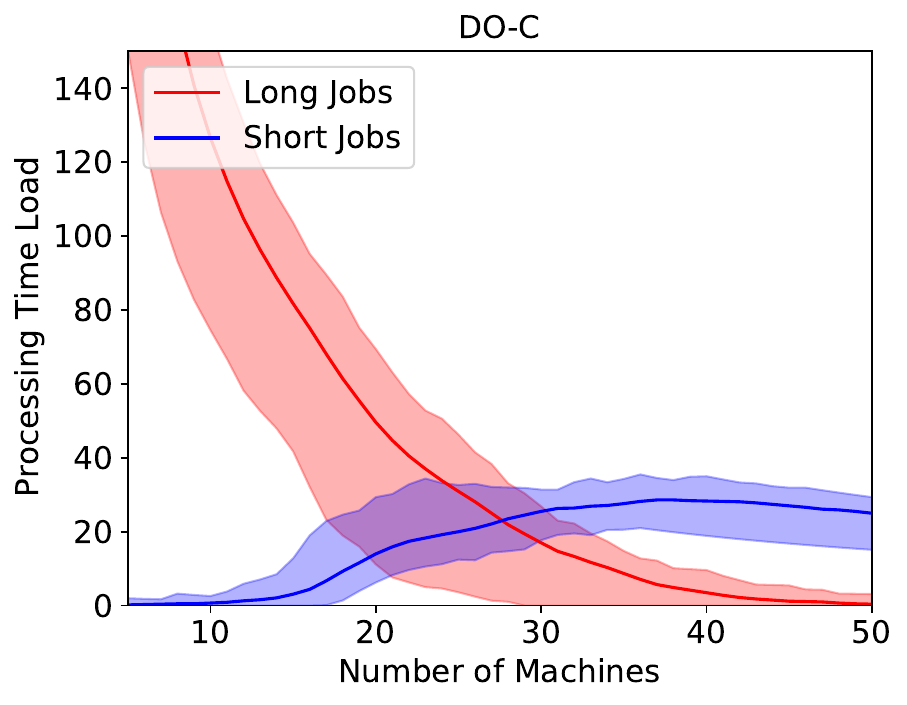}
        \end{center}
        \caption{\northampton}
        %\label{Figure:northampton_durations_histogram}
    \end{subfigure}
\caption{\change{Line chart comparing the processing time load $\frac{1}{m}\sum_{j}p_j$ of long and short jobs. Given a number $m$ machines, (i) the solid line plots the average load and (ii) the shaded area shows the difference between the maximum and minimum load, with respect to all job sets (days).}
%\todo[inline]{Can't read the figure fonts}
%schedules realized by Royal Mail and nominal schedules produced by solving MILP (\ref{Eq:DT_Model}) with CPLEX. Dotted vertical lines indicate Mondays.
}
\label{Figure:Long_Short_Load}
\end{figure}

\subsection{Evaluation of Heuristics}
\label{Section:Evaluation_Heuristics}

\begin{figure}[t]
    \begin{subfigure}[t]{0.5\textwidth}
        \begin{center}
        \includegraphics[width=\textwidth]{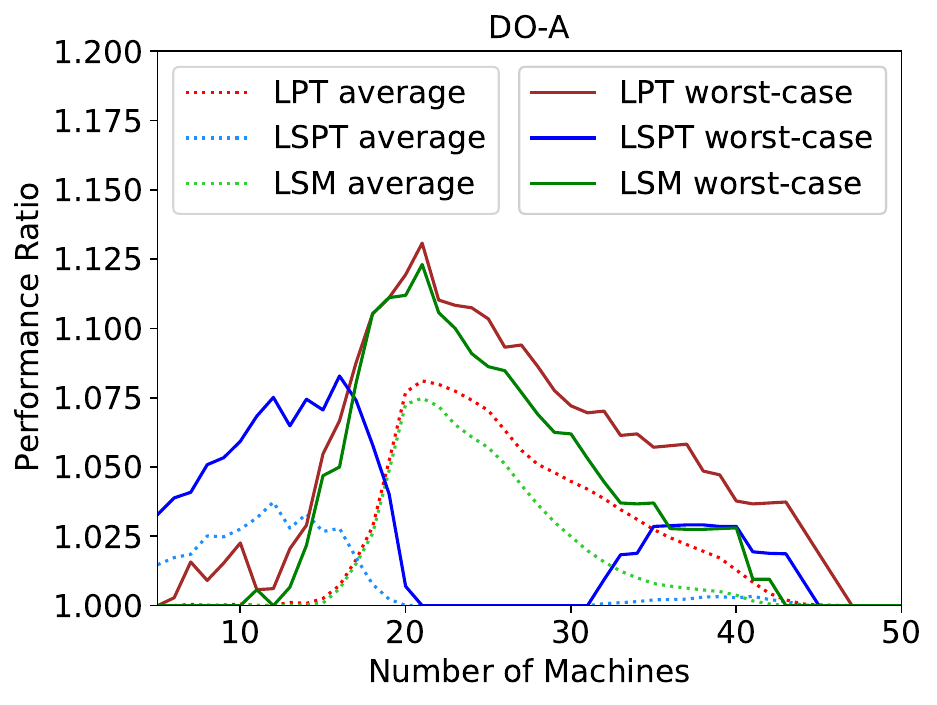}
        \end{center}
        \caption{\basingstoke}
        %\label{Figure:basingstoke_return_times_histogram}
    \end{subfigure}
    \begin{subfigure}[t]{0.5\textwidth}
        \begin{center}
        \includegraphics[width=\textwidth]{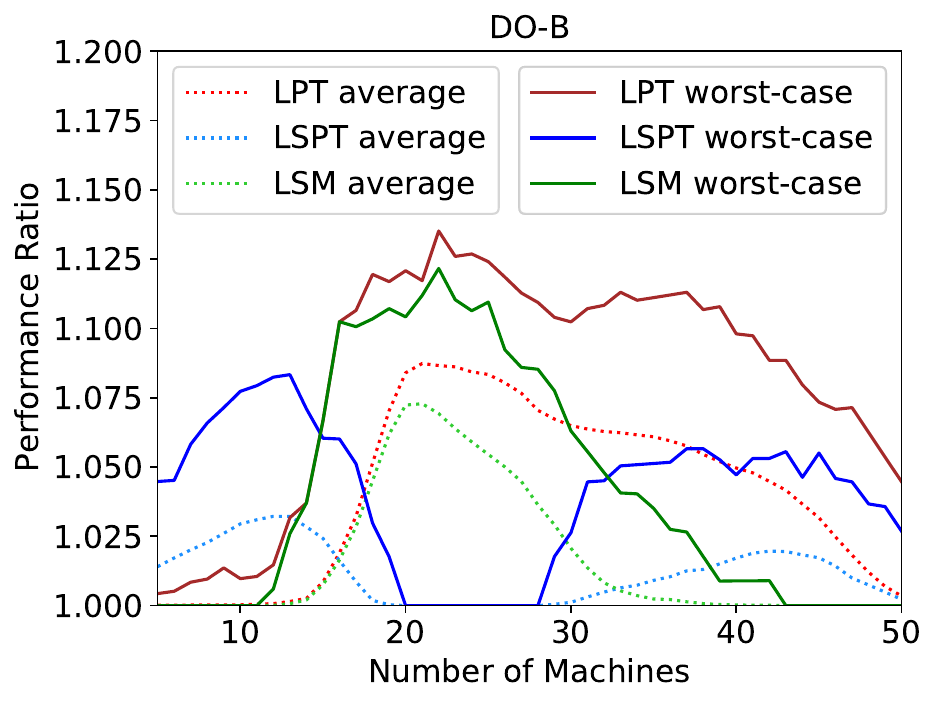}
        \end{center}
        \caption{\colchester}
        %\label{Figure:colchester_durations_histogram}
    \end{subfigure}
    \begin{subfigure}[t]{\textwidth}
        \begin{center}
        \includegraphics[width=0.5\textwidth]{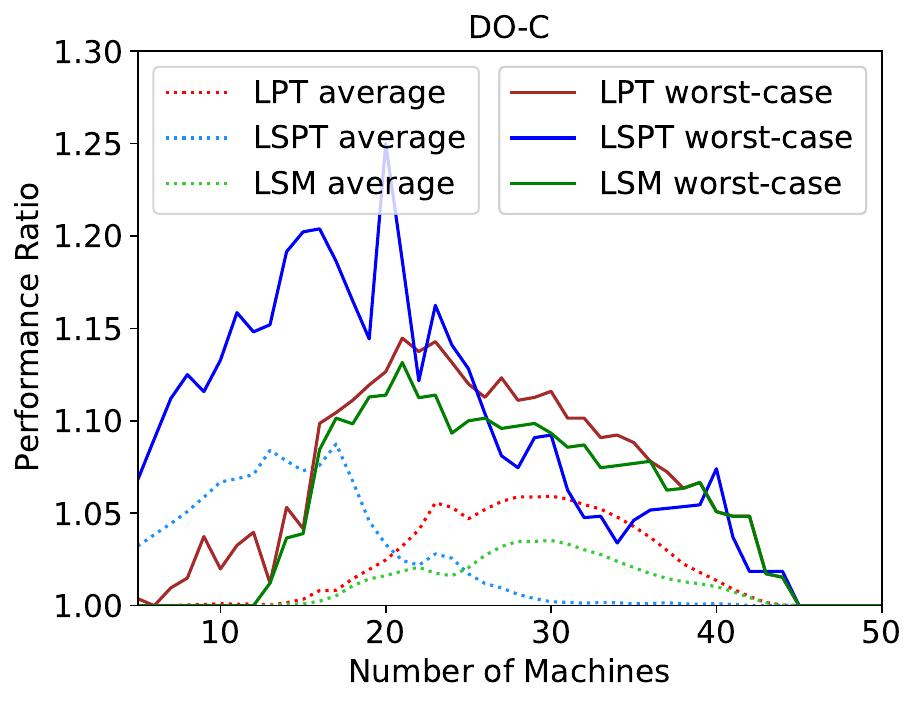}
        \end{center}
        \caption{\northampton}
        %\label{Figure:northampton_durations_histogram}
    \end{subfigure}
\caption{\change{Line chart comparing the performance of the heuristics. For a given number $m$ of machines, (i) solid lines plot the worst-case performance ratios and (ii) dotted lines plot the average performance ratios of the heuristics with respect to all jobs sets (days).}
%\todo[inline]{Can't read the figure fonts}
%schedules realized by Royal Mail and nominal schedules produced by solving MILP (\ref{Eq:DT_Model}) with CPLEX. Dotted vertical lines indicate Mondays.
}
\label{Figure:Heuristic_Performance}
\end{figure}

\change{
This section compares the performance of the LPT (Section~\ref{Section:LPT}), LSPT %(Section~\ref{Section:SPT}) 
(Section~\ref{Section:Long_Short}) and LSM (Section~\ref{Section:LSM}) heuristics
%with respect to optimal solutions 
for BJSP. 
For each job set $\mathcal{J}$ (i.e.\ collection of jobs executed by a delivery office in one day) and number $m\in[5,50]$ of machines, we create a BJSP instance with $g=1$. 
%job sets in the Royal Mail data and a range of values for the number of available machines.
%compares the makespan of the heuristic schedules using the ranges of $m$ values specified in Section~\ref{Section:Long_Short_Jobs}.
%obtained by optimal solving Integer Program~\ref{Eq:DT_Model} (Section~\ref{Section:Problem_Definition}) using CPLEX.
%We use the job sets in the Royal Mail data for evaluating our heuristics.
%Specifically, a job set $\mathcal{J}$ consists of all jobs executed by delivery office on a certain date.
%For each $\mathcal{J}$ in the data set, we evaluate the heuristics using every instance $(m,\mathcal{J})$, where $m\in[5,50]$. 
%Specifically, we consider every $m$ value in the range $[10,30]$ for \basingstoke, $[15,30]$ for \colchester\ and $[20,35]$ for \northampton.
%Table~\ref{Table:Delivery_Offices} reports information about the resulting BJSP instances.
%the range of $n$ values, i.e.\ the number of jobs executed during a day, for each delivery office and  range for the number of machines.
%Here, we examine all possible values in the range $[0.1n,0.3n]$ for the number of machines.
%In all cases, we set a unit BJSP parameter, i.e.\ $g=1$.
We solve every instance $I=(m,\mathcal{J})$ using LPT, LSPT and LSM.
Let $T(A,I)$ be the makespan of the schedule produced by heuristic $A$ for instance $I$.
Then, the performance ratio of heuristic $A$ for $I$ is $T(A,I)/T^*(I)$, where $T^*(I)$ is the best heuristically computed makespan for $I$.
Figure~\ref{Figure:Heuristic_Performance} plots the worst-case and average performance ratio of each heuristic, for each $m$ value.
For small $m$ values, the number $n^S$ of short jobs 
%starts and the idle period before the last job begins are 
is low compared to the mean load $\frac{1}{m}\sum_{j\in\mathcal{J}^L}p_j$ of long jobs and LPT produces good heuristic schedules, noticeably better than LSPT.
%LSPT schedules are worse because 
% Therefore, LPT produces good solutions LSPT schedules, which leaves big jobs in the end of the schedule.
As $m$ increases, the idle period before the last job begins becomes progressively more significant and LSPT tends to compute better schedules than LPT (recall that LSPT and LPT achieve low idle machine before and after, respectively, the last job start).
% The opposite situation is encountered for the larger $m$ values and LSPT dominates LPT. 
Interestingly, LSPT produces the best heuristic schedules in the pathological LPT case where $n^S$ tends to become equal to $\frac{1}{m}\sum_{j\in\mathcal{J}^L}p_j$.
Finally, LSM consistently produces better schedules than LPT.
Therefore, scheduling short jobs in parallel with long jobs early in a schedule is useful for achieving low makespan.
Our findings indicate that a LSM variant where long jobs are executed according to SPT, in parallel with short jobs, might be a possible alternative for obtaining better BJSP approximation algorithms. }
% In our implementations, we run LSM once for every $[\frac{m}{8},\frac{m}{4}]$
% The harder instances are encountered when $n^S$ is equal to $\frac{1}{m}\sum_{j\in\mathcal{J}^L}p_j$. 
% Finally, in all cases, that LSM is better than LPT, i.e.\ there is a benefit in executing short jobs in parallel with long jobs early in the schedule.
% Note, that we have implemented LSM using LSPT.

\begin{figure}[t]
    \begin{subfigure}[t]{0.5\textwidth}
        \begin{center}
        \includegraphics[width=\textwidth]{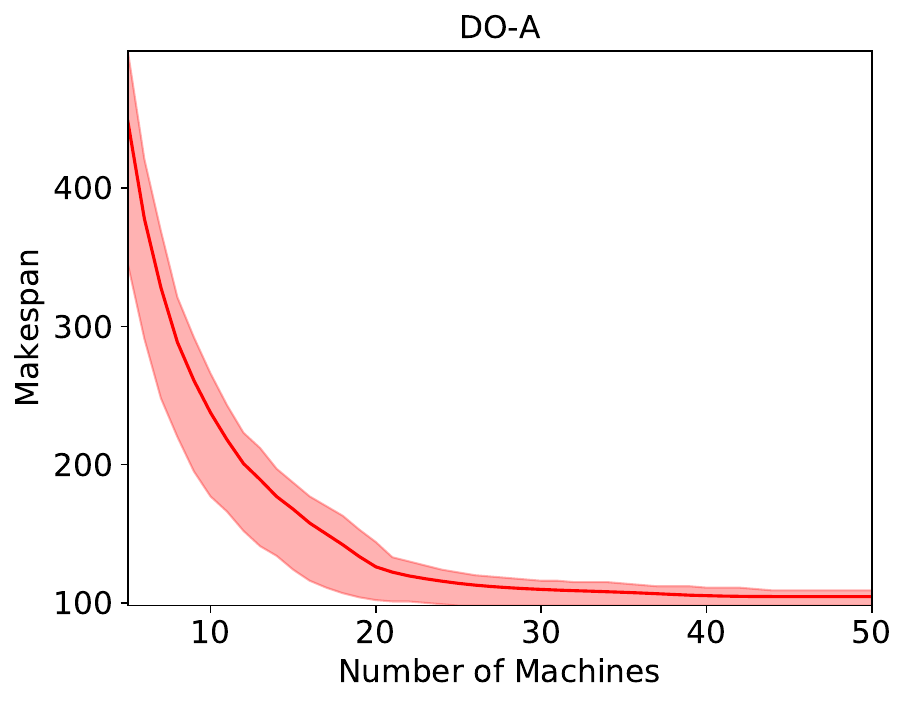}
        \end{center}
        \caption{\basingstoke}
        %\label{Figure:basingstoke_return_times_histogram}
    \end{subfigure}
    \begin{subfigure}[t]{0.5\textwidth}
        \begin{center}
        \includegraphics[width=\textwidth]{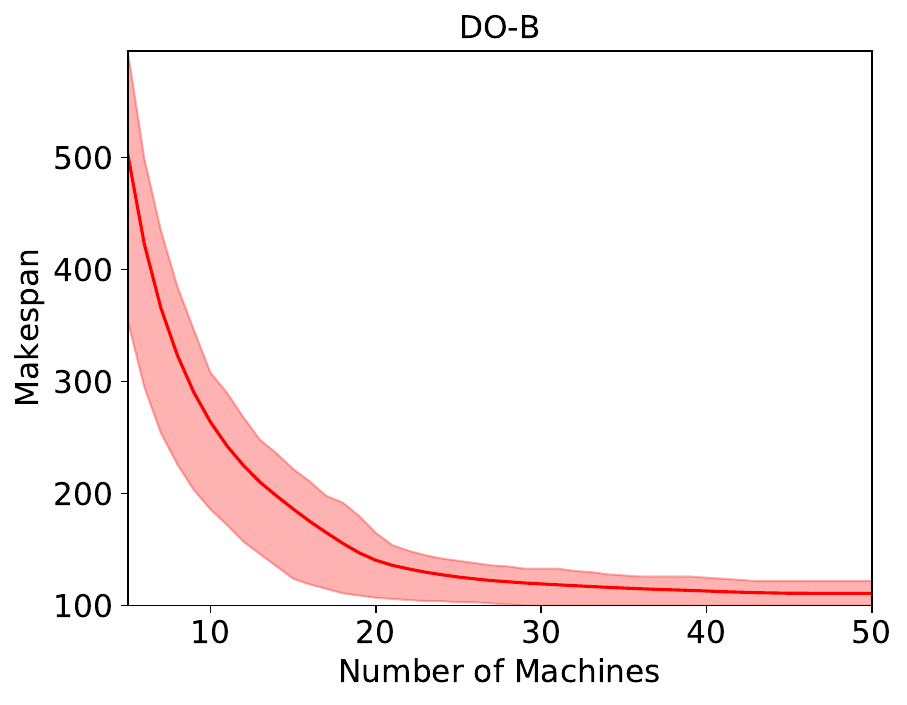}
        \end{center}
        \caption{\colchester}
        %\label{Figure:colchester_durations_histogram}
    \end{subfigure}
    \begin{subfigure}[t]{\textwidth}
        \begin{center}
        \includegraphics[width=0.5\textwidth]{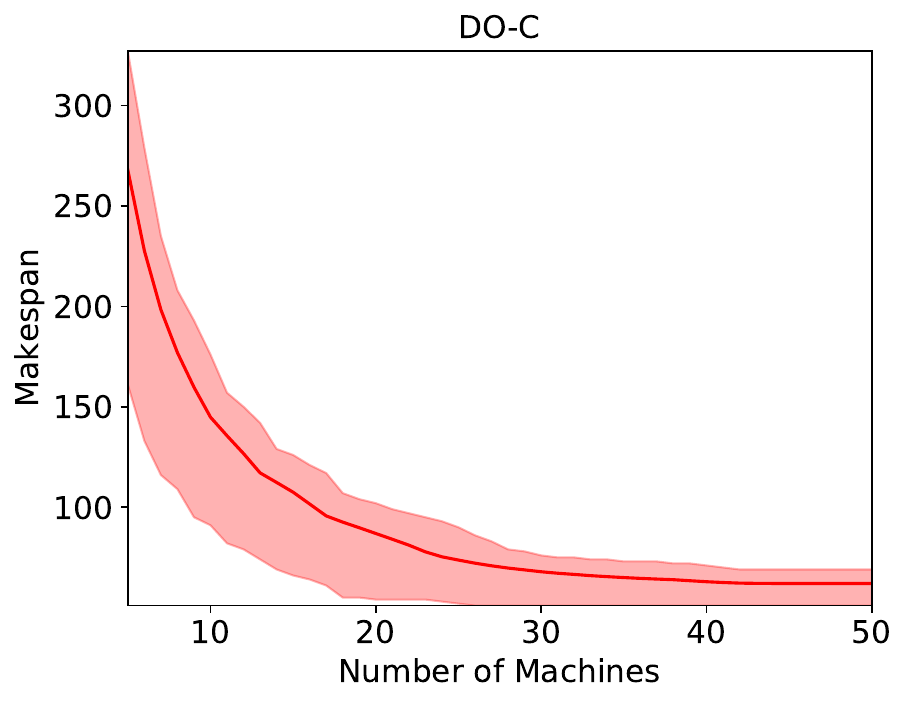}
        \end{center}
        \caption{\northampton}
        %\label{Figure:northampton_durations_histogram}
    \end{subfigure}
\caption{\change{Trade-off between the makespan and number of machines. Given a number $m$ machines, (i) the solid line plots the average makespan and (ii) the shaded area shows the difference between the maximum and minimum makespan, with respect to all job sets (days).}
%\todo[inline]{Can't read the figure fonts}
%schedules realized by Royal Mail and nominal schedules produced by solving MILP (\ref{Eq:DT_Model}) with CPLEX. Dotted vertical lines indicate Mondays.
}
\label{Figure:Heuristic_Tradeoff}
\end{figure}

\change{
For completeness, Figure 10 plots the trade-off between the best heuristically computed makespan $T$ with respect to the number $m$ of machines. 
Clearly, as $m$ increases, $T$ decreases.
This finding supports using machine augmentation for better makespan schedules in the presence of uncertainty.
}

\subsection{Evaluation of Historical Schedules}
\label{Section:Evaluation_Historical_Schedules}

This section evaluates the Royal Mail historical schedules (i) efficiency in number of used machines and (ii) sensitivity with respect to processing time and (iii) BJSP parameter variations.

For part (i), we solve each BJSP instance by feeding the corresponding MILP (\ref{Eq:LexOpt_Model}) model to CPLEX.
In these MILP (\ref{Eq:LexOpt_Model}) models, we set $\theta=0$ to minimize the number of used machines.
Figure \ref{Figure:Historical_Schedules} compares the number of machines in the Royal Mail historical schedules and the CPLEX solutions. 
%the minimal number of machines obtained by solving the corresponding MILP (\ref{Eq:LexOpt_Model}) with $\theta=0$ using CPLEX.
%Note that the job processing times are identical in the two schedules. 
%objective values obtained by minimizing the number of vehicles with full input knowledge.
We observe that nominal optimal solutions save at least 10, 25, and 10 vehicles per day compared to historical schedules for \basingstoke, \colchester, and \northampton, respectively.
This finding is a strong indication that more efficient fleet management might be possible in Royal Mail delivery offices.
%We note that, in given a mail delivery scheduling problem instance $I$, the number $n$ of itineraries is a trivial upper bound on the number $v(S)$ of vehicles in any feasible schedule $S$ for $I$, i.e.\ $v(S)\leq n$ $\forall S\in \mathcal{S}(I)$.
%Let $\mathcal{H}$ be the set of Royal Mail historical schedules.
%We identify as major weakness of the historical schedules in our data set the fact that $v(S)=n$ $\forall S\in\mathcal{H}$.

\begin{figure}[t]
    \begin{subfigure}[t]{0.5\textwidth}
        \begin{center}
        \includegraphics[width=\textwidth]{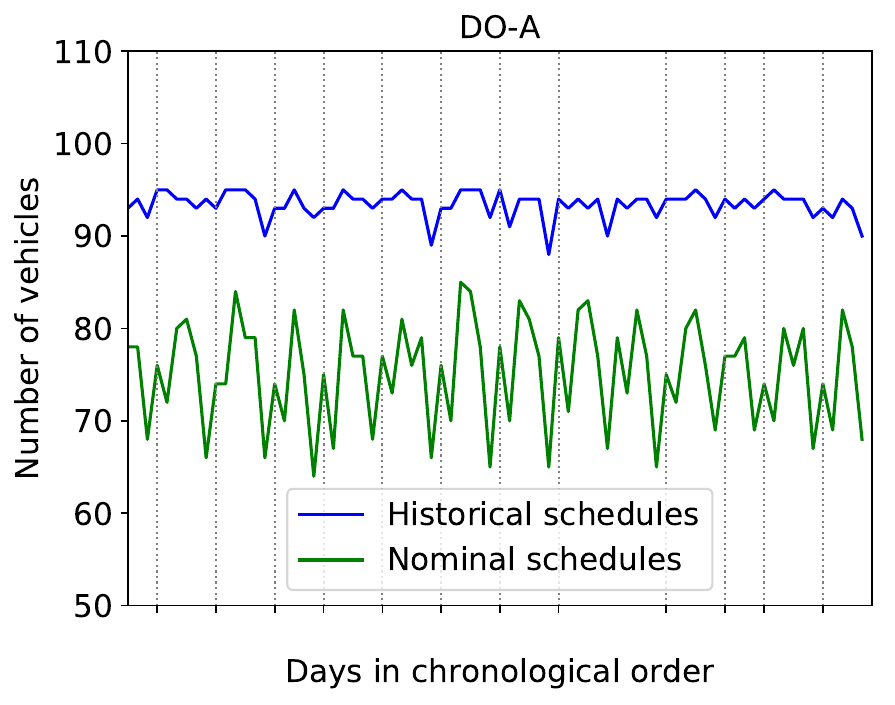}
        \end{center}
        \caption{\basingstoke}
        %\label{Figure:basingstoke_return_times_histogram}
    \end{subfigure}
    \begin{subfigure}[t]{0.5\textwidth}
        \begin{center}
        \includegraphics[width=\textwidth]{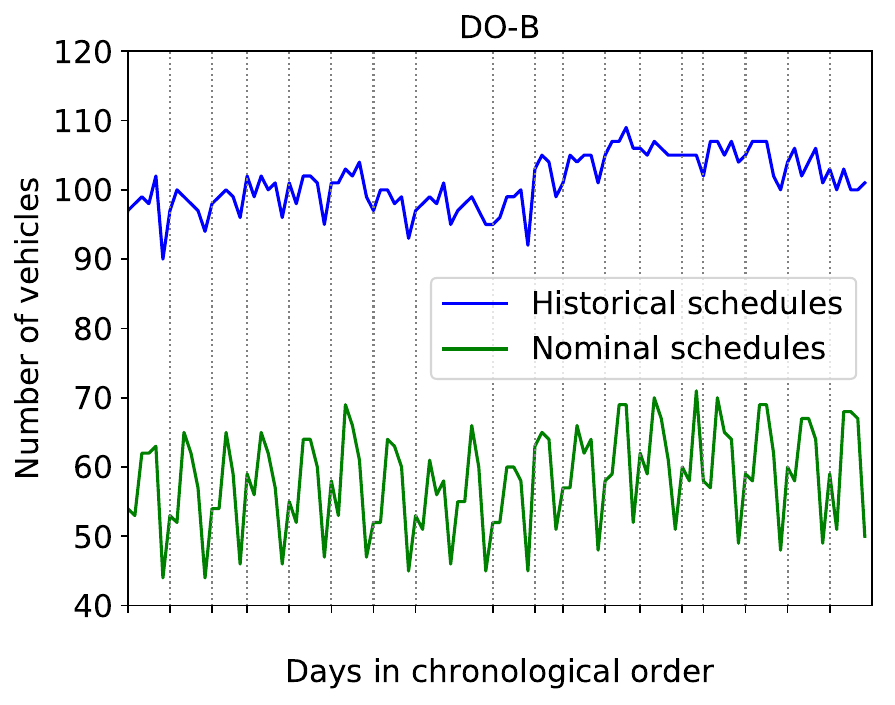}
        \end{center}
        \caption{\colchester}
        %\label{Figure:colchester_durations_histogram}
    \end{subfigure}
    \begin{subfigure}[t]{\textwidth}
        \begin{center}
        \includegraphics[width=0.5\textwidth]{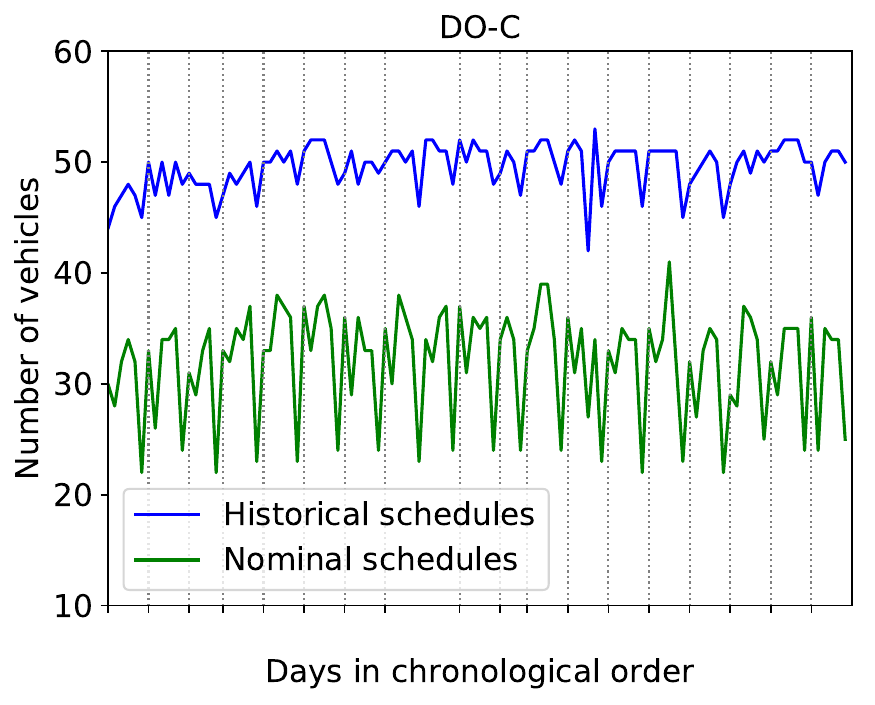}
        \end{center}
        \caption{\northampton}
        %\label{Figure:northampton_durations_histogram}
    \end{subfigure}
\caption{Line chart comparing the number of used machines in historical and nominal optimal schedules. 
%\todo[inline]{Can't read the figure fonts}
%schedules realized by Royal Mail and nominal schedules produced by solving MILP (\ref{Eq:DT_Model}) with CPLEX. Dotted vertical lines indicate Mondays.
}
\label{Figure:Historical_Schedules}
\end{figure}

%\subsection{Sensitivity Analysis}

For part (ii), we create a set of perturbed instances.
In particular, for each original instance $I$, we create one perturbed instance $\tilde{I}$ where the processing time of each job $j\in\mathcal{J}$ is decreased by a factor $f_j=0.5$.
We reduce the processing times to guarantee feasibility.
For both instances $I$ and $\tilde{I}$, we employ CPLEX to solve the corresponding MILP (\ref{Eq:LexOpt_Model}) formulations with $\theta=0$. 
Figure~\ref{Figure:Itinerary_Delays} compares the number of used machines obtained for the original and perturbed instances. %to the ones obtained for the perturbed instances. 
Not surprisingly, doubling itinerary durations results in a proportional increase on the number of used vehicles in the nominal optimal solution.
But, Figure~\ref{Figure:Itinerary_Delays} exhibits an important consequence of uncertainty in Royal Mail fleet management.
%displays the number of vehicles in chronological order week after a week.
Disturbances amplify the difference in number of used machines between different days for one delivery office.
This situation leads to inefficient machine utilization.

%shows that doubling itinerary durations results in a proportional increase on the number of used vehicles in the nominal solution.
%each BJSP instance twice: once with the original processing times and another with    

%This section evaluates the impact on the objective value of (i) itinerary duration uncertainty, (ii) gate capacity, and (iii) itinerary splitting.
%Figure~\ref{Figure:Itinerary_Delays} shows that doubling itinerary durations results in a proportional increase on the number of used vehicles in the nominal solution.
%Itinerary delays may significantly increase the required vehicles for delivering mail.
%Furthermore, the naive approach of dealing with uncertainty by assuming larger itinerary durations than their nominal values may attain poor performance. 
%Figure~\ref{Figure:Itinerary_Delays} also shows the amplified difference of the required vehicles between consecutive days. 
%These observations motivate the design of robust mail delivery scheduling methods.
%Figure~\ref{Figure:Gate_Capacity} illustrates the impact of the gate capacity constraint on the objective value.

For part (iii), we investigate the effect of modifying the BJSP parameter for each delivery office.
Figure~\ref{Figure:Gate_Capacity} depicts the obtained results.
Adding BJSP constraints, especially in the \colchester\ case, may significantly increase the number of used machines.
This outcome motivates further investigations on scheduling with BJSP constraints.

% In \basingstoke  and \colchester, the gate capacity may severely affect the final number of required vehicles.
% In \northampton, gate delays provoke less disturbances.
% These findings motivate the design of efficient policies to deal with the gate capacity constraint and stimulate further investigations on degenerate conditions.

% reveals an unfavorable consequence of uncertainty in Royal Mail fleet management.
% The obtained findings have been the motivating elements of the current work. 
% We observe a repetitive pattern with respect the number of vehicles on a weekly basis. 
% For instance, in \basingstoke, we observe a sharp nominal objective value increase ranging from 5 to 10 vehicles between Tuesdays and Wednesdays.
% These fluctuations motivate vehicle pooling between delivery offices.

% Figure~\ref{Figure:Iteration_Splitting} shows the benefit gained by allowing a splitting of a bounded number of jobs.

\begin{figure}[t]
    \begin{subfigure}[t]{0.5\textwidth}
        \begin{center}
        \includegraphics[width=\textwidth]{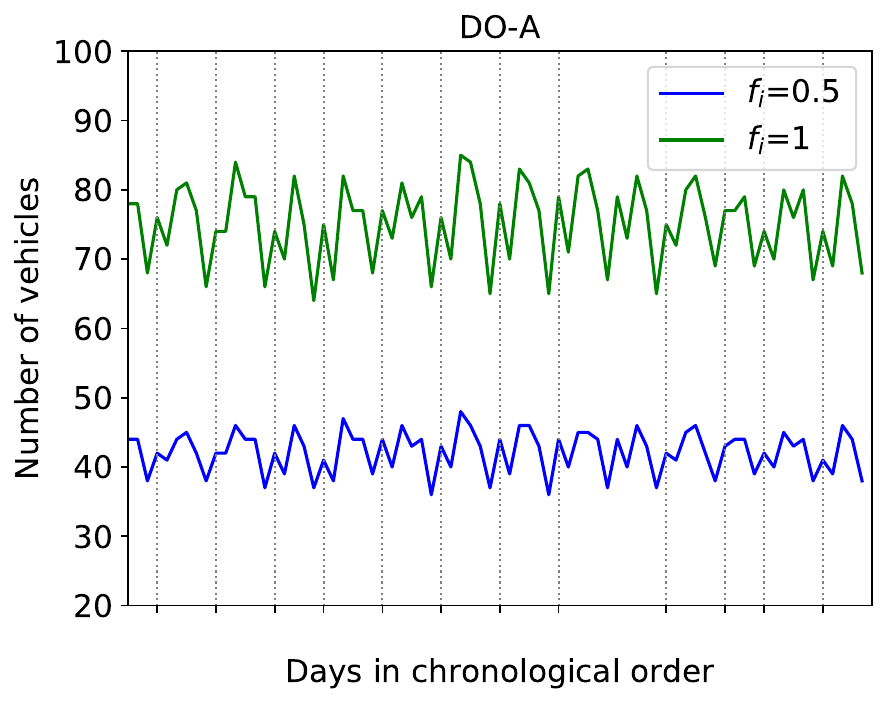}
        \end{center}
        \caption{\basingstoke}
        %\label{Figure:basingstoke_return_times_histogram}
    \end{subfigure}
    \begin{subfigure}[t]{0.5\textwidth}
        \begin{center}
        \includegraphics[width=\textwidth]{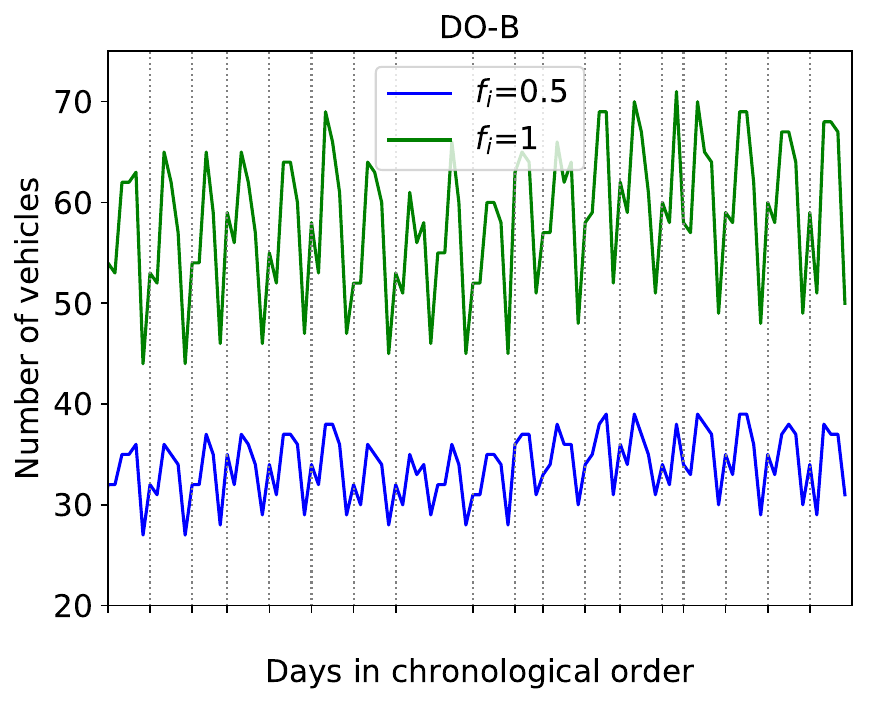}
        \end{center}
        \caption{\colchester}
        %\label{Figure:colchester_durations_histogram}
    \end{subfigure}
    \begin{subfigure}[t]{\textwidth}
        \begin{center}
        \includegraphics[width=0.5\textwidth]{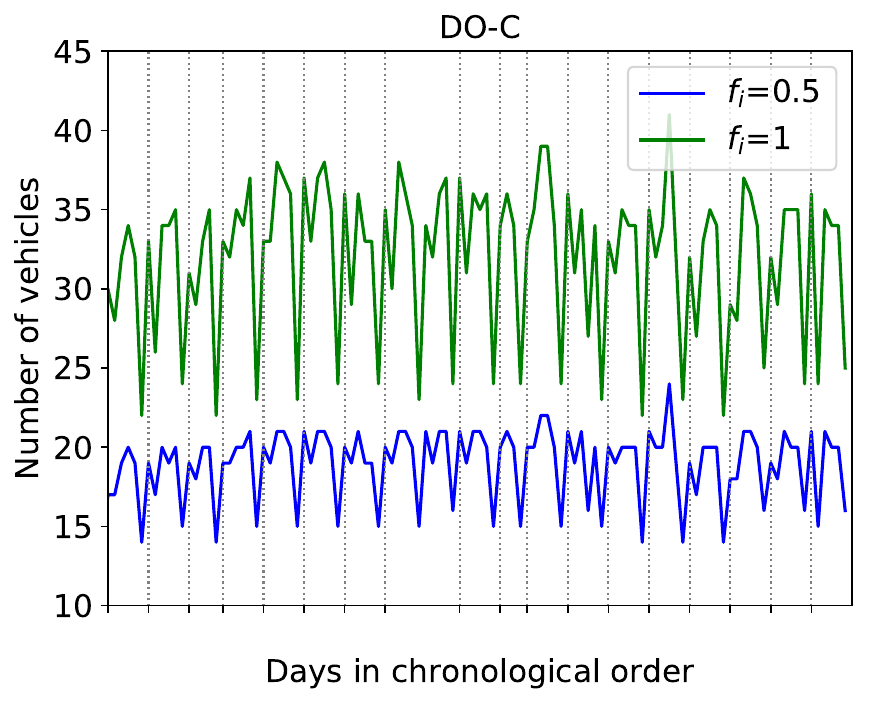}
        \end{center}
        \caption{\northampton}
        %\label{Figure:northampton_durations_histogram}
    \end{subfigure}
\caption{
Line chart comparing the number of used machines between the original instances and instances where the job processing times have been halved. 
%\todo[inline]{Can't read the figure fonts} 
%with different itinerary durations, i.e.\ perturbation factors. Each instance is solved with the nominal durations and the original durations halved. Dotted lines indicate Mondays.
}
\label{Figure:Itinerary_Delays}
\end{figure}

\begin{figure}[t]
    \begin{subfigure}[t]{0.5\textwidth}
        \begin{center}
        \includegraphics[width=\textwidth]{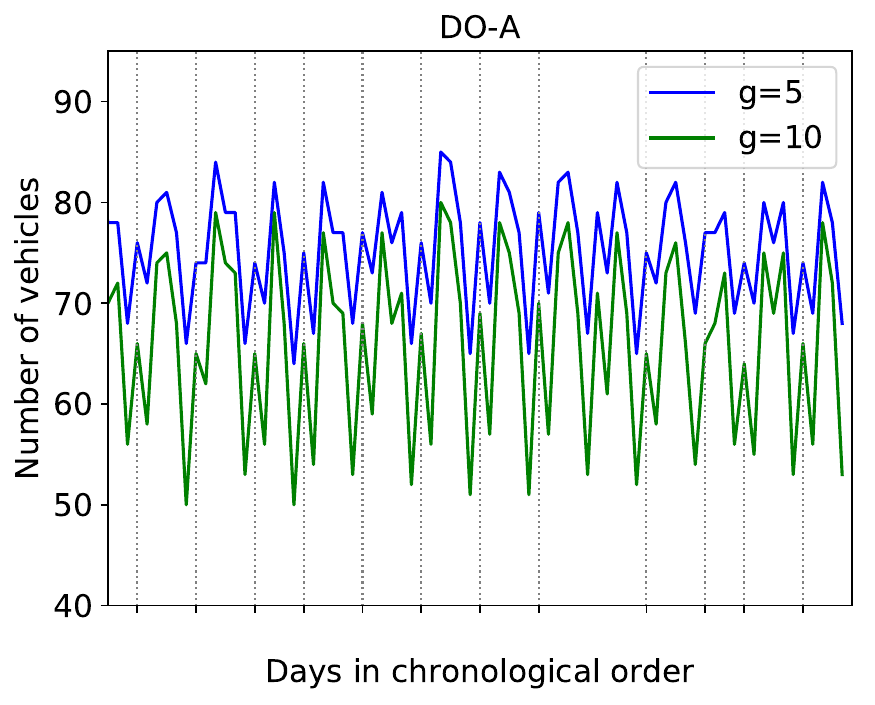}
        \end{center}
        \caption{\basingstoke}
        %\label{Figure:basingstoke_return_times_histogram}
    \end{subfigure}
    \begin{subfigure}[t]{0.5\textwidth}
        \begin{center}
        \includegraphics[width=\textwidth]{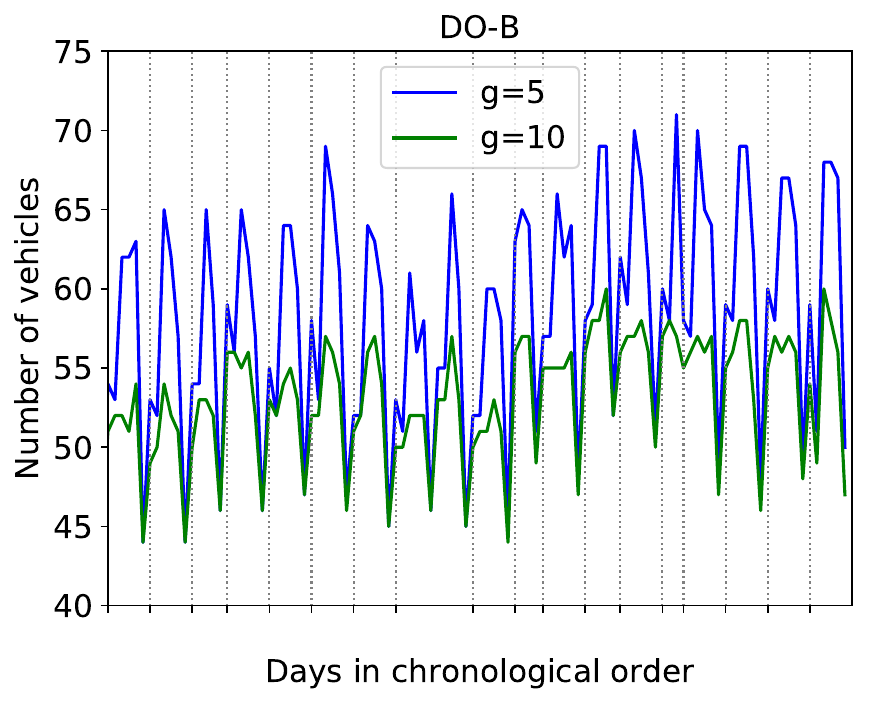}
        \end{center}
        \caption{\colchester}
        %\label{Figure:colchester_durations_histogram}
    \end{subfigure}
    \begin{subfigure}[t]{\textwidth}
        \begin{center}
        \includegraphics[width=0.5\textwidth]{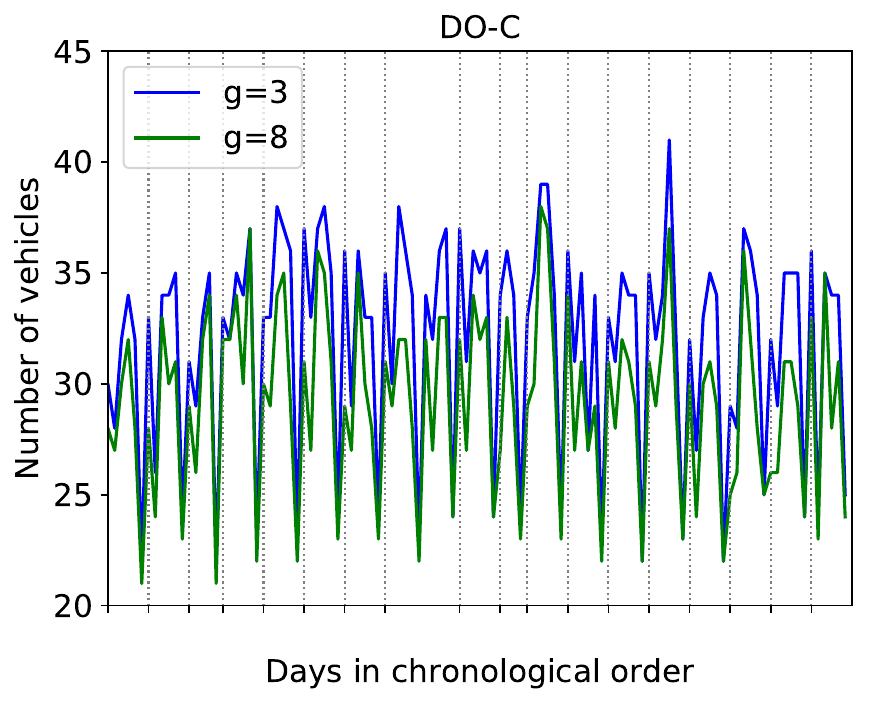}
        \end{center}
        \caption{\northampton}
        %\label{Figure:northampton_durations_histogram}
    \end{subfigure}
\caption{Line chart comparing the number of vehicles between instances with different BJSP parameters. 
%\todo[inline]{Can't read the figure fonts}
%for instances with different gate capacities. Each instance is solved with two different gate capacities. Dotted lines correspond to Mondays.
}
\label{Figure:Gate_Capacity}
\end{figure}

\subsection{Robustness Assessment}
\label{Section:Robustness_Assessment}

\change{Next, we evaluate the two-stage robust optimization method in Section~\ref{Section:LexOpt}.
Specifically, we show that low characteristic value results in more robust BJSP schedules.
%To this end, 
We adopt the experimental setup in \cite{Letsios2018}.
The Royal Mail instances are considered as the nominal ones before any disturbances occur.
The true instances after uncertainty realization are derived by choosing a new processsing time $\tilde{p}_j$ for job each $j\in\mathcal{J}$ uniformly at random from the interval $[p_j-50\%p_j,p_j+50\%p_j]$, where $p_j$ is the nominal processing time.
}

\change{For each nominal instance $I$, we generate a collection $\mathcal{C}(I)$ of feasible, diverse (i.e.\ with different characteristic values) first-stage schedules, using the CPLEX solution pool feature.
Let $F^*(I)=\min\{F(\mathcal{S}):\mathcal{S}\in\mathcal{C}(I)\}$ be the minimum characteristic value among all schedules in $\mathcal{C}(I)$.
Next, denote by $\tilde{I}$ and $\tilde{\mathcal{S}}$ the perturbed instance and recovered schedule from $\mathcal{S}$, after uncertainty realization. 
Moreover, let $V^*(\tilde{I})=\min\{V(\tilde{\mathcal{S}}):\mathcal{S}\in\mathcal{C}(I)\}$ be the minimum number of machines achievable for $\tilde{I}$ with perfect knowledge.
For each initial schedule $\mathcal{S}\in\mathcal{C}(I)$ and recovered schedule $\tilde{\mathcal{S}}$, we set a normalized characteristic value $F^N(\mathcal{S})=F(\mathcal{S})/F^*(I)$ and normalized number of used machines $V^N(\tilde{\mathcal{S}})=V(\tilde{\mathcal{S}})/V^*(\tilde{I})$.
Figure~\ref{Figure:Robustness} correlates $F^N(\cdot)$ to $V^N(\cdot)$, by plotting every computed pair $(F^N(\mathcal{S}),V^N(\tilde{\mathcal{S})})$, for every nominal instance and initial solution.
% is a scatter plot with a pair $(V^N(\mathcal{S}),F^N(\mathcal{S}))$ for each schedule and for all instances.
We observe that the smaller the initial characteristic value is, the better the final solution we get in terms of number of machines.}

\begin{figure}[t]
    \begin{subfigure}[t]{0.5\textwidth}
        \begin{center}
        \includegraphics[width=\textwidth]{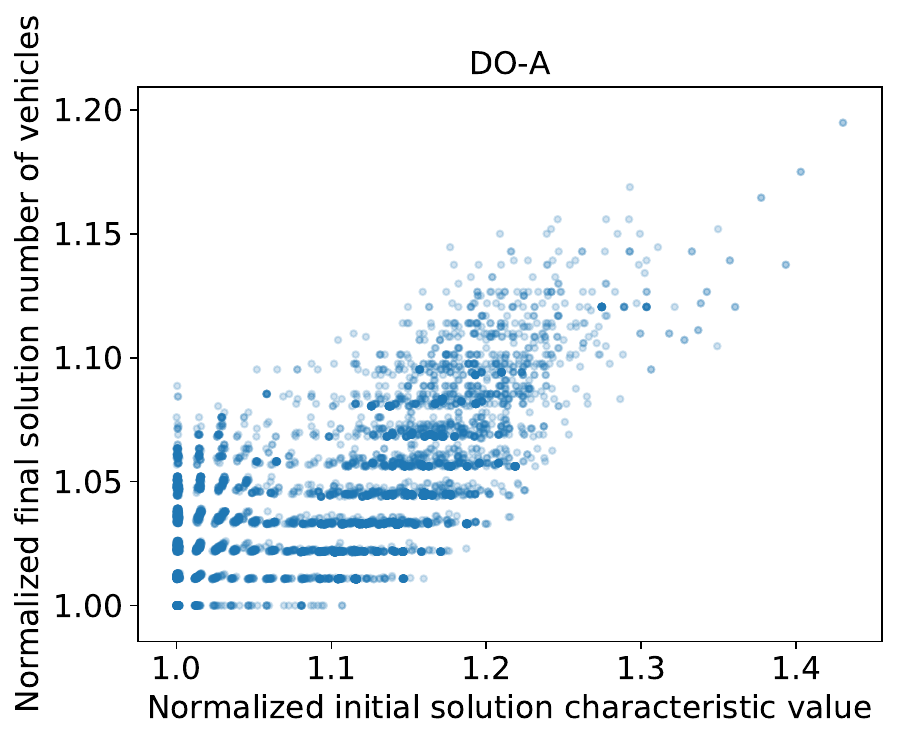}
        \end{center}
        \caption{\basingstoke}
        %\label{Figure:basingstoke_return_times_histogram}
    \end{subfigure}
    \begin{subfigure}[t]{0.5\textwidth}
        \begin{center}
        \includegraphics[width=\textwidth]{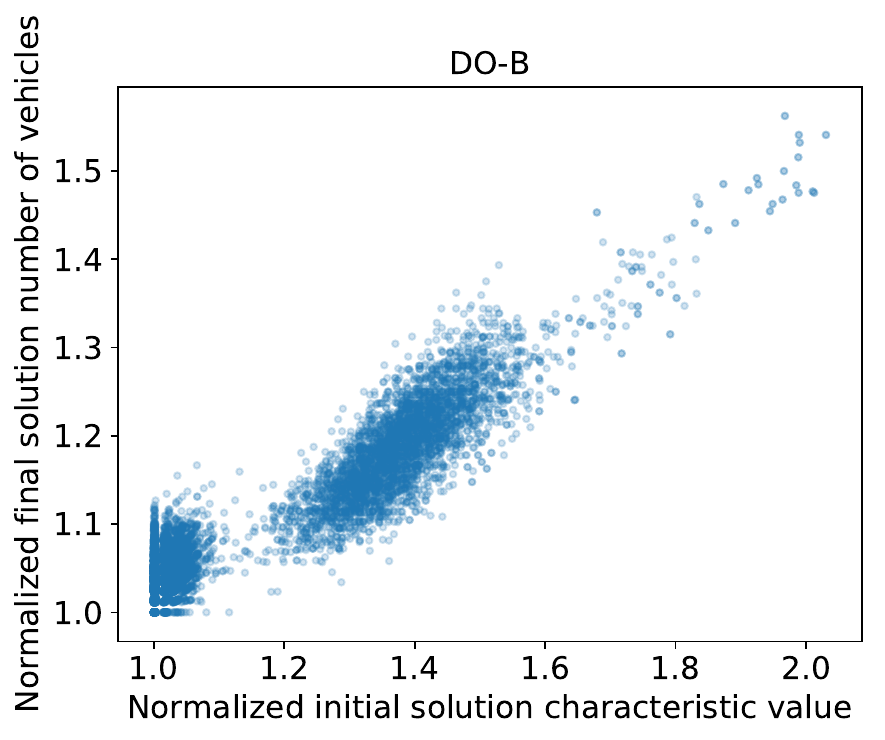}
        \end{center}
        \caption{\colchester}
        %\label{Figure:colchester_durations_histogram}
    \end{subfigure}
    \begin{subfigure}[t]{\textwidth}
        \begin{center}
        \includegraphics[width=0.5\textwidth]{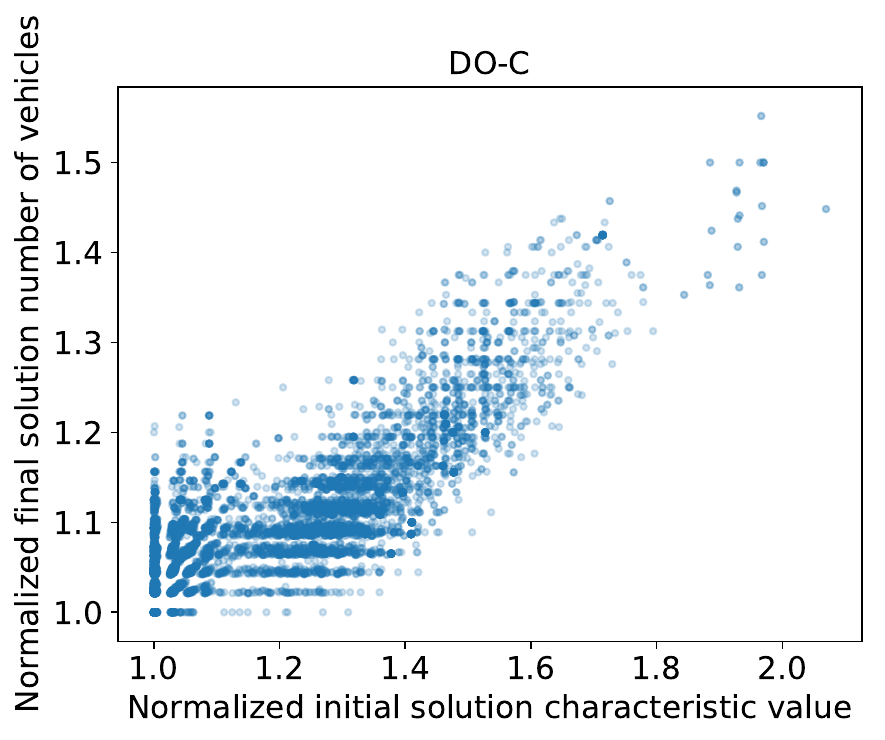}
        \end{center}
        \caption{\northampton}
        %\label{Figure:northampton_durations_histogram}
    \end{subfigure}
\caption{Scatter plots comparing the initial solution weighted value with the final solution number of vehicles. 
%\todo[inline]{Can't read the figure fonts}
}
\label{Figure:Robustness}
\end{figure}

% There are two approaches to detect how robust a solution structure is:
% \begin{itemize}
%     \item Generate multiple solutions and show that one solution is more robust than another.
%     \item Generate a solutions according to a strategy that is offline, and then according to a more robust, and show improvement if we take into account uncertainty.
% \end{itemize}

%%%%%%%%%%%%%%%%%%%%%%%%%%%%%%%%%%%%%%%%%%%%%%%%%%%%%%%%%%%%%%%%

\section{Conclusion}
\label{Section:Conclusion}
%\input{conclusion}
%\newpage

This manuscript initiates study of the \emph{bounded job start scheduling problem (BJSP)}, e.g.\ as arising in Royal Mail deliveries. This project is part of our larger aims towards approximation algorithms for process systems engineering \cite{LETSIOS2019106599}.
%{\color{red} 
The main contributions are (i) 
%a 1.985-approximation algorithm using 1.2-machine augmentation 
better than 2-approximation algorithms for various cases of the problem
and (ii) a two-stage robust optimization approach for BJSP under uncertainty based on machine augmentation and lexicographic optimization, whose performance is substantiated empirically.
%}
We conclude with a collection of future directions.
Because BJSP relaxes scheduling problems with non-overlapping constraints for which better than 2-approximation algorithms are impossible under widely adopted conjectures, the existence of an algorithm with an approximation ratio strictly better than 2 which does not use resource augmentation is an intriguing open question.
A positive answer combining LSM algorithm with a new algorithm specialized to instances with many very long jobs is possible.
Moreover, analyzing the price of robustness of the proposed two-stage robust optimization approach may provide new insights for effectively solving BJSP under uncertainty.
Our findings demonstrate a strong potential for more efficient Royal Mail resource allocation by using vehicle sharing between different delivery offices. 
The BJSP scheduling problem where multiple delivery offices are integrated in a unified setting and vehicle exchanges are performed on a daily basis consists a promising direction for fruitful investigations. %that may allow instantiate machine augmentation.
In this context, recent work on car pooling might be useful.
\change{Finally, bounded job start constraints are broadly relevant to vehicle routing problems \cite{Fisher1997, Gounaris2016}.
However, we are not aware of any work in this direction.}

\begin{acknowledgements}
This work was funded by Engineering \& Physical Sciences Research Council (EPSRC) grant EP/M028240/1 and an EPSRC Research Fellowship to RM (grant number EP/P016871/1).
\end{acknowledgements}

\bibliographystyle{spmpsci}      % mathematics and physical sciences
\bibliography{refs}   % name your BibTeX data base

\end{document}

%% file: opt_bjs.tex
\begin{tikzpicture}[scale=0.3]

\node (m1) at (-1.5,0.5) {\footnotesize $M_1$};
\node (m2) at (-1.5,1.5) {\footnotesize $M_2$};
\draw[dotted, thick] (-1.5,2.1) -- (-1.5,3);
\node (m4) at (-1.5,3.5) {\footnotesize $M_{m-1}$};
\node (m4) at (-1.5,4.5) {\footnotesize $M_m$};

%------------------------------------------------------

\filldraw[fill=black!20!white] (0,0) rectangle (8,1);
\node at (4,0.5) {\footnotesize $m$};

\filldraw[fill=black!20!white] (1,1) rectangle (8,2);
\node at (4.5,1.5) {\footnotesize $m-1$};

\draw[dotted, thick] (7,2.1) -- (7,3);

\filldraw[fill=black!20!white] (6,3) rectangle (8,4);
\node at (7,3.5) {\footnotesize 2};

\filldraw[fill=black!20!white] (7,4) rectangle (8,5);
\node at (7.5,4.5) {\footnotesize 1};

%------------------------------------------------------

\filldraw[fill=black!20!white] (8,0) rectangle (16,1);
\node at (12,0.5) {\footnotesize $m$};

\filldraw[fill=black!20!white] (9,1) rectangle (16,2);
\node at (12.5,1.5) {\footnotesize $m-1$};

\draw[dotted, thick] (15,2.1) -- (15,3);

\filldraw[fill=black!20!white] (14,3) rectangle (16,4);
\node at (15,3.5) {\footnotesize $2$};

\filldraw[fill=black!20!white] (15,4) rectangle (16,5);
\node at (15.5,4.5) {\footnotesize 1};

%------------------------------------------------------

\draw[dotted, thick] (18,0.5) -- (22,0.5);

%------------------------------------------------------

\filldraw[fill=black!20!white] (24,0) rectangle (32,1);
\node at (28,0.5) {\footnotesize $m$};

\filldraw[fill=black!20!white] (25,1) rectangle (32,2);
\node at (28.5,1.5) {\footnotesize $m-1$};

\draw[dotted, thick] (31,2.1) -- (31,3);

\filldraw[fill=black!20!white] (30,3) rectangle (32,4);
\node at (31,3.5) {\footnotesize $2$};

\filldraw[fill=black!20!white] (31,4) rectangle (32,5);
\node at (31.5,4.5) {\footnotesize 1};
%------------------------------------------------------

\end{tikzpicture}

%% file: opt_pcmax.tex
\begin{tikzpicture}[scale=0.3]

\node (m1) at (-1.5,0.5) {\footnotesize $M_1$};
\node (m2) at (-1.5,1.5) {\footnotesize $M_2$};
\draw[dotted, thick] (-1.5,2.1) -- (-1.5,3);
\node (m4) at (-1.5,3.5) {\footnotesize $M_{m-1}$};
\node (m4) at (-1.5,4.5) {\footnotesize $M_m$};

%------------------------------------------------------

\filldraw[fill=black!20!white] (0,0) rectangle (8,1);
\node at (4,0.5) {\footnotesize $m$};
\filldraw[fill=black!20!white] (8,0) rectangle (9,1);
\node at (8.5,0.5) {\footnotesize $1$};

\filldraw[fill=black!20!white] (0,1) rectangle (7,2);
\node at (3.5,1.5) {\footnotesize $m-1$};
\filldraw[fill=black!20!white] (7,1) rectangle (9,2);
\node at (8,1.5) {\footnotesize $2$};

\draw[dotted, thick] (4.5,2.1) -- (4.5,3);

\filldraw[fill=black!20!white] (0,3) rectangle (2,4);
\node at (1,3.5) {\footnotesize $2$};
\filldraw[fill=black!20!white] (2,3) rectangle (9,4);
\node at (5.5,3.5) {\footnotesize $m-1$};

\filldraw[fill=black!20!white] (0,4) rectangle (1,5);
\node at (0.5,4.5) {\footnotesize $1$};
\filldraw[fill=black!20!white] (1,4) rectangle (9,5);
\node at (5,4.5) {\footnotesize $m$};

%------------------------------------------------------

\draw[dotted, thick] (11,0.5) -- (15,0.5);

%------------------------------------------------------

\filldraw[fill=black!20!white] (17,0) rectangle (25,1);
\node at (21,0.5) {\footnotesize $m$};
\filldraw[fill=black!20!white] (25,0) rectangle (26,1);
\node at (25.5,0.5) {\footnotesize $1$};

\filldraw[fill=black!20!white] (17,1) rectangle (24,2);
\node at (20.5,1.5) {\footnotesize $m-1$};
\filldraw[fill=black!20!white] (24,1) rectangle (26,2);
\node at (25,1.5) {\footnotesize $2$};

\draw[dotted, thick] (21.5,2.1) -- (21.5,3);

\filldraw[fill=black!20!white] (17,3) rectangle (19,4);
\node at (18,3.5) {\footnotesize $2$};
\filldraw[fill=black!20!white] (19,3) rectangle (26,4);
\node at (22.5,3.5) {\footnotesize $m-1$};

\filldraw[fill=black!20!white] (17,4) rectangle (18,5);
\node at (17.5,4.5) {\footnotesize $1$};
\filldraw[fill=black!20!white] (18,4) rectangle (26,5);
\node at (22,4.5) {\footnotesize $m$};

\end{tikzpicture}

%% file: 2_tightness_lpt.tex
\begin{tikzpicture}[scale=0.3]

\node (m1) at (-0.8,0.5) {\footnotesize $M_1$};
\node (m2) at (-0.8,1.5) {\footnotesize $M_2$};
\node (m3) at (-0.8,2.5) {\footnotesize $M_3$};
\node (m4) at (-0.8,3.5) {\footnotesize $M_4$};

%------------------------------------------------------

\filldraw[fill=black!20!white] (0,0) rectangle (7,1);
\node at (3.5,0.5) {\footnotesize $p$};

\filldraw[fill=black!20!white] (1,1) rectangle (8,2);
\node at (4.5,1.5) {\footnotesize $p$};

\draw[dotted, thick] (4,2.5) -- (8,2.5);

\filldraw[fill=black!20!white] (3,3) rectangle (11,4);
\node at (7,3.5) {\footnotesize $p$};

%------------------------------------------------------

\draw[dotted, thick] (9,0.5) -- (12,0.5);

%------------------------------------------------------

\filldraw[fill=black!20!white] (14,0) rectangle (21,1);
\node at (17.5,0.5) {\footnotesize $p$};

\filldraw[fill=black!20!white] (15,1) rectangle (22,2);
\node at (18.5,1.5) {\footnotesize $p$};

\draw[dotted, thick] (18,2.5) -- (22,2.5);

\filldraw[fill=black!20!white] (17,3) rectangle (24,4);
\node at (20.5,3.5) {\footnotesize $p$};

%------------------------------------------------------

\filldraw[fill=black!20!white] (21,0) rectangle (22,1);
\node at (21.5,0.5) {\footnotesize 1};

\draw[dotted, thick] (23,0.5) -- (26,0.5);

\filldraw[fill=black!20!white] (27,0) rectangle (28,1);
\node at (27.5,0.5) {\footnotesize 1};

\end{tikzpicture}

%% file: 2_tightness_opt.tex
\begin{tikzpicture}[scale=0.3]

\node (m1) at (-0.8,0.5) {\footnotesize $M_1$};
\node (m2) at (-0.8,1.5) {\footnotesize $M_2$};
\node (m3) at (-0.8,2.5) {\footnotesize $M_3$};
\node (m4) at (-0.8,3.5) {\footnotesize $M_4$};
\node (m4) at (-0.8,4.5) {\footnotesize $M_5$};

%------------------------------------------------------

\filldraw[fill=black!20!white] (0,0) rectangle (9,1);
\node at (4.5,0.5) {\footnotesize $p$};

\filldraw[fill=black!20!white] (1,1) rectangle (10,2);
\node at (5.5,1.5) {\footnotesize $p$};

\draw[dotted, thick] (4,2.5) -- (8,2.5);

\filldraw[fill=black!20!white] (3,3) rectangle (12,4);
\node at (7.5,3.5) {\footnotesize $p$};

\filldraw[fill=black!20!white] (4,4) rectangle (5,5);
\node at (4.5,4.5) {\footnotesize 1};

\draw[dotted, thick] (5.5,4.5) -- (7.5,4.5);

\filldraw[fill=black!20!white] (8,4) rectangle (9,5);
\node at (8.5,4.5) {\footnotesize 1};

%------------------------------------------------------

\filldraw[fill=black!20!white] (9,0) rectangle (18,1);
\node at (13.5,0.5) {\footnotesize $p$};

\filldraw[fill=black!20!white] (10,1) rectangle (19,2);
\node at (14.5,1.5) {\footnotesize $p$};

\draw[dotted, thick] (13,2.5) -- (17,2.5);

\filldraw[fill=black!20!white] (12,3) rectangle (21,4);
\node at (16.5,3.5) {\footnotesize $p$};

\filldraw[fill=black!20!white] (13,4) rectangle (14,5);
\node at (13.5,4.5) {\footnotesize 1};

\draw[dotted, thick] (14.5,4.5) -- (16.5,4.5);

\filldraw[fill=black!20!white] (17,4) rectangle (18,5);
\node at (17.5,4.5) {\footnotesize 1};

%------------------------------------------------------

\draw[dotted, thick] (21,0.5) -- (25,0.5);

%------------------------------------------------------

\filldraw[fill=black!20!white] (27,0) rectangle (36,1);
\node at (31.5,0.5) {\footnotesize $p$};

\filldraw[fill=black!20!white] (28,1) rectangle (37,2);
\node at (32.5,1.5) {\footnotesize $p$};

\draw[dotted, thick] (31,2.5) -- (35,2.5);

\filldraw[fill=black!20!white] (30,3) rectangle (39,4);
\node at (34.5,3.5) {\footnotesize $p$};

\filldraw[fill=black!20!white] (31,4) rectangle (32,5);
\node at (31.5,4.5) {\footnotesize 1};

\draw[dotted, thick] (32.5,4.5) -- (34.5,4.5);

\filldraw[fill=black!20!white] (35,4) rectangle (36,5);
\node at (35.5,4.5) {\footnotesize 1};

%------------------------------------------------------

\end{tikzpicture}

%% file: compact_initial.tex
\begin{tikzpicture}[scale=0.3]

\node (m1) at (-1.5,0.5) {\footnotesize $M_1$};
\node (m2) at (-1.5,1.5) {\footnotesize $M_2$};
\node (m4) at (-1.5,2.5) {\footnotesize $M_3$};
\node (m4) at (-1.5,3.5) {\footnotesize $M_4$};

%------------------------------------------------------

\filldraw[fill=black!20!white] (0,0) rectangle (10,1);
\node at (5,0.5) {\footnotesize $10$};

\filldraw[fill=black!20!white] (2,1) rectangle (10,2);
\node at (6,1.5) {\footnotesize $8$};

\filldraw[fill=black!20!white] (4,2) rectangle (10,3);
\node at (7,2.5) {\footnotesize $6$};

\filldraw[fill=black!20!white] (6,3) rectangle (10,4);
\node at (8,3.5) {\footnotesize $4$};

\filldraw[fill=black!20!white] (10,0) rectangle (14,1);
\node at (12,0.5) {\footnotesize $4$};

\filldraw[fill=black!20!white] (11,1) rectangle (14,2);
\node at (12.5,1.5) {\footnotesize $3$};

\filldraw[fill=black!20!white] (12,2) rectangle (14,3);
\node at (13,2.5) {\footnotesize $2$};

\filldraw[fill=black!20!white] (13,3) rectangle (14,4);
\node at (13.5,3.5) {\footnotesize $1$};

%------------------------------------------------------

\end{tikzpicture}

%% file: compact_final.tex
\begin{tikzpicture}[scale=0.3]

\node (m1) at (-1.5,0.5) {\footnotesize $M_1$};
\node (m2) at (-1.5,1.5) {\footnotesize $M_2$};
\node (m4) at (-1.5,2.5) {\footnotesize $M_3$};
\node (m4) at (-1.5,3.5) {\footnotesize $M_4$};

%------------------------------------------------------

\filldraw[fill=black!20!white] (0,0) rectangle (10,1);
\node at (5,0.5) {\footnotesize 10};

\filldraw[fill=black!20!white] (10,0) rectangle (11,1);
\node at (10.5,0.5) {\footnotesize 1};

\filldraw[fill=black!20!white] (1,1) rectangle (9,2);
\node at (5,1.5) {\footnotesize 8};

\filldraw[fill=black!20!white] (9,1) rectangle (11,2);
\node at (10,1.5) {\footnotesize 2};

\filldraw[fill=black!20!white] (2,2) rectangle (8,3);
\node at (5,2.5) {\footnotesize 6};

\filldraw[fill=black!20!white] (8,2) rectangle (11,3);
\node at (9.5,2.5) {\footnotesize 3};

\filldraw[fill=black!20!white] (3,3) rectangle (7,4);
\node at (5,3.5) {\footnotesize 4};

\filldraw[fill=black!20!white] (7,3) rectangle (11,4);
\node at (9,3.5) {\footnotesize 4};

\end{tikzpicture}

%% file: 43_tightness_lpt.tex
\begin{tikzpicture}[scale=0.3]

\node (m1) at (-1.5,0.5) {\footnotesize $M_1$};
\node (m2) at (-1.5,1.5) {\footnotesize $M_2$};
\draw[dotted, thick] (-1.5,2.1) -- (-1.5,3);
\node (m4) at (-1.5,3.5) {\footnotesize $M_{m-1}$};
\node (m4) at (-1.5,4.5) {\footnotesize $M_m$};

%------------------------------------------------------

\filldraw[fill=black!20!white] (0,0) rectangle (10,1);
\node at (5,0.5) {\footnotesize $2m-1$};

\filldraw[fill=black!20!white] (1,1) rectangle (10,2);
\node at (5.5,1.5) {\footnotesize $2m-2$};

\draw[dotted, thick] (7,2.1) -- (7,3);

\filldraw[fill=black!20!white] (4,3) rectangle (10,4);
\node at (7,3.5) {\footnotesize $m+1$};

\filldraw[fill=black!20!white] (5,4) rectangle (10,5);
\node at (7.5,4.5) {\footnotesize $m$};

\filldraw[fill=black!20!white] (10,0) rectangle (15,1);
\node at (12.5,0.5) {\footnotesize $m$};

%------------------------------------------------------

\end{tikzpicture}

%% file: 43_tightness_opt.tex
\begin{tikzpicture}[scale=0.3]

\node (m1) at (-1.5,0.5) {\footnotesize $M_1$};
\node (m2) at (-1.5,1.5) {\footnotesize $M_2$};
\draw[dotted, thick] (-1.5,2.1) -- (-1.5,3);
\node (m4) at (-1.5,3.5) {\footnotesize $M_{m-1}$};
\node (m4) at (-1.5,4.5) {\footnotesize $M_m$};

%------------------------------------------------------

\filldraw[fill=black!20!white] (0,0) rectangle (5,1);
\node at (2.5,0.5) {\footnotesize $m$};

\filldraw[fill=black!20!white] (5,0) rectangle (10,1);
\node at (7.5,0.5) {\footnotesize $m$};

\filldraw[fill=black!20!white] (1,1) rectangle (10,2);
\node at (5.5,1.5) {\footnotesize $2m-1$};

\draw[dotted, thick] (7,2.1) -- (7,3);

\filldraw[fill=black!20!white] (3,3) rectangle (10,4);
\node at (6.5,3.5) {\footnotesize $m+2$};

\filldraw[fill=black!20!white] (4,4) rectangle (10,5);
\node at (7,4.5) {\footnotesize $m+1$};

\end{tikzpicture}

%% file: lsm.tex
\begin{tikzpicture}[scale=0.5]

\node (m1) at (-1,0.5) {\footnotesize $M_1$};
\node (m2) at (-1,4.1) {\footnotesize $M_k$};
\draw[dotted, thick] (-1,1.5) -- (-1,3.25);
\node (m4) at (-1,4.75) {\footnotesize $M_{k-1}$};
\node (m4) at (-1,5.75) {\footnotesize $M_m$};

%------------------------------------------------------

% \filldraw[fill=black!20!white] (0,0) rectangle (10,1);
% \node at (5,0.5) {\footnotesize $2m-1$};

\draw[thick] (0,0) -- (0,6);
\draw[thick] (10,0) -- (10,6);
\draw[thick] (0,0) -- (10,0);
\draw[thick] (0,6) -- (10,6);
\draw[thick] (0,4.5) -- (10,4.5);

%------------------------------------------------------

\draw [decorate,decoration={brace,mirror,amplitude=10pt},xshift=0pt,yshift=0pt]
(10,0) -- (10,4.5) node [black,midway,xshift=0.7cm] {\footnotesize $\mathcal{M}^L$};

\draw [decorate,decoration={brace,mirror,amplitude=5pt},xshift=0pt,yshift=0pt]
(10,4.5) -- (10,6) node [black,midway,xshift=0.5cm] {\footnotesize $\mathcal{M}^S$};

\end{tikzpicture}

%% file: two_stage_model.tex
\begin{tikzpicture}[scale=0.8]

\node (a) at (-1,0) {};
\node (b) at (11.5,0) {\scriptsize time};
\draw[->] (a) edge (b) ;

\draw[dotted] (0,-0.4) -- (0,0.45) -- (3,0.45) -- (3,-0.4) -- (0,-0.4);

\draw[-] (1.5,0.05) edge (1.5,-0.05);
%\node (a) at (1.5,0.5) {\scriptsize Planning};
\node (a) at (1.5,0.25) {\scriptsize Planning Phase};
\node (a) at (1.5,-0.25) {\scriptsize Stage 1};

\draw[-] (5,0.05) edge (5,-0.05);
%\node (a) at (5,0.5) {\scriptsize Uncertainty};
\node (a) at (5,0.25) {\scriptsize Uncertainty Realization};
\node (a) at (5,-0.25) {\scriptsize Disturbances};

\draw[dotted] (7,-0.4) -- (7,0.45) -- (10,0.45) -- (10,-0.4) -- (7,-0.4);

\draw[-] (8.5,0.05) edge (8.5,-0.05);
%\node (a) at (8.5,0.5) {\scriptsize Recovery};
\node (a) at (8.5,0.25) {\scriptsize Recovery Phase};
\node (a) at (8.5,-0.25) {\scriptsize Stage 2};

\end{tikzpicture}

%% file: rescheduling.tex
\begin{tikzpicture}[scale=0.8]

\node at (-0.5,0.7) {\scriptsize Input $I$};

\node at (-2.6,0.0) {\scriptsize Efficient schedule $\mathcal{S}$ for input $I$};
%\node at (-1.3,-0.15) {\scriptsize for input $I_{\text{init}}$};

\node at (-1.35,-0.7) {\scriptsize Perturbed input $\tilde{I}$};

\draw [black, fill=black!5, rounded corners] (1.3,0.6) rectangle (3.7,-0.6);
\node at (2.5,0.15) {\scriptsize Recovery};
\node at (2.5,-0.15) {\scriptsize algorithm};

\node at (6.5,0.15) {\scriptsize Efficient schedule $\tilde{\mathcal{S}}$};
\node at (6.5,-0.15) {\scriptsize for input $\tilde{I}$};

\node (p1) at (0.2,0.7) {};
\node (p2) at (1.2,0.3) {};
\draw[->] (p1) edge (p2);

\node (p3) at (0.2,-0.7) {};
\node (p4) at (1.2,-0.3) {};
\draw[->] (p3) edge (p4);

\node (p5) at (0.2,0) {};
\node (p6) at (1.2,0) {};
\draw[->] (p5) edge (p6);

\node (p7) at (3.85,0) {};
\node (p8) at (4.85,0) {};
\draw[->] (p7) edge (p8);

\end{tikzpicture}

%% file: main.bbl
\begin{thebibliography}{10}
\providecommand{\url}[1]{{#1}}
\providecommand{\urlprefix}{URL }
\expandafter\ifx\csname urlstyle\endcsname\relax
  \providecommand{\doi}[1]{DOI~\discretionary{}{}{}#1}\else
  \providecommand{\doi}{DOI~\discretionary{}{}{}\begingroup
  \urlstyle{rm}\Url}\fi

\bibitem{BenTal2009}
Ben-Tal, A., El~Ghaoui, L., Nemirovski, A.: Robust optimization, vol.~28.
\newblock Princeton University Press (2009)

\bibitem{BenTal2004}
Ben-Tal, A., Goryashko, A., Guslitzer, E., Nemirovski, A.: Adjustable robust
  solutions of uncertain linear programs.
\newblock Mathematical Programming \textbf{99}(2), 351--376 (2004)

\bibitem{BenTal2000}
Ben-Tal, A., Nemirovski, A.: Robust solutions of linear programming problems
  contaminated with uncertain data.
\newblock Mathematical Programming \textbf{88}(3), 411--424 (2000)

\bibitem{Bertsimas2011}
Bertsimas, D., Brown, D.B., Caramanis, C.: Theory and applications of robust
  optimization.
\newblock SIAM review \textbf{53}(3), 464--501 (2011)

\bibitem{Bertsimas2010}
Bertsimas, D., Caramanis, C.: Finite adaptability in multistage linear
  optimization.
\newblock IEEE Transactions on Automatic Control \textbf{55}(12), 2751--2766
  (2010)

\bibitem{Bertsimas2004}
Bertsimas, D., Sim, M.: The price of robustness.
\newblock Operations Research \textbf{52}(1), 35--53 (2004)

\bibitem{Billaut2009}
Billaut, J.C., Sourd, F.: Single machine scheduling with forbidden start times.
\newblock 4OR \textbf{7}(1), 37--50 (2009)

\bibitem{Chassein2019}
Chassein, A.B., Dokka, T., Goerigk, M.: Algorithms and uncertainty sets for
  data-driven robust shortest path problems.
\newblock European Journal of Operational Research \textbf{274}(2), 671--686
  (2019)

\bibitem{davis2011survey}
Davis, R.I., Burns, A.: A survey of hard real-time scheduling for
  multiprocessor systems.
\newblock ACM computing surveys (CSUR) \textbf{43}(4), 35 (2011)

\bibitem{Fisher1997}
Fisher, M.L., J{\"o}rnsten, K.O., Madsen, O.B.: Vehicle routing with time
  windows: Two optimization algorithms.
\newblock Operations research \textbf{45}(3), 488--492 (1997)

\bibitem{suraj-g}
G, Suraj: A practical analysis of optimisation and recovery under uncertainty.
\newblock Master's thesis, Imperial College London (2019).
\newblock
  \url{https://www.imperial.ac.uk/media/imperial-college/faculty-of-engineering/computing/public/1819-ug-projects/GS-A-Practical-Analysis-of-Optimisation-and-Recovery-Under-Uncertainty.pdf}

\bibitem{Gabay2016}
Gabay, M., Rapine, C., Brauner, N.: High-multiplicity scheduling on one machine
  with forbidden start and completion times.
\newblock Journal of Scheduling \textbf{19}(5), 609--616 (2016)

\bibitem{Garey1979}
Garey, M.R., Johnson, D.S.: Computers and intractability, vol. 174.
\newblock Freeman (1979)

\bibitem{Goerigk2016}
Goerigk, M., Sch{\"{o}}bel, A.: Algorithm engineering in robust optimization.
\newblock In: L.~Kliemann, P.~Sanders (eds.) Algorithm Engineering - Selected
  Results and Surveys, \emph{Lecture Notes in Computer Science}, vol. 9220, pp.
  245--279. Springer (2016)

\bibitem{Gounaris2016}
Gounaris, C.E., Repoussis, P.P., Tarantilis, C.D., Wiesemann, W., Floudas,
  C.A.: An adaptive memory programming framework for the robust capacitated
  vehicle routing problem.
\newblock Transportation Science \textbf{50}(4), 1239--1260 (2016)

\bibitem{Graham1969}
Graham, R.L.: Bounds on multiprocessing timing anomalies.
\newblock SIAM Journal on Applied Mathematics \textbf{17}(2), 416--429 (1969)

\bibitem{Hanasusanto2015}
Hanasusanto, G.A., Kuhn, D., Wiesemann, W.: K-adaptability in two-stage robust
  binary programming.
\newblock Operations Research \textbf{63}(4), 877--891 (2015)

\bibitem{kalyanasundaram2000speed}
Kalyanasundaram, B., Pruhs, K.: Speed is as powerful as clairvoyance.
\newblock Journal of the ACM (JACM) \textbf{47}(4), 617--643 (2000)

\bibitem{Kolen2007}
Kolen, A.W.J., Lenstra, J.K., Papadimitriou, C.H., Spieksma, F.C.R.: Interval
  scheduling: {A} survey.
\newblock Naval Research Logistics \textbf{54}(5), 530--543 (2007)

\bibitem{Kouvelis2013}
Kouvelis, P., Yu, G.: Robust discrete optimization and its applications,
  vol.~14.
\newblock Springer Science \& Business Media (2013)

\bibitem{LETSIOS2019106599}
Letsios, D., Baltean-Lugojan, R., F., Mistry, M., Wiebe, J., Misener, R.:
  Approximation algorithms for process systems engineering.
\newblock Computers \& Chemical Engineering p. 106599 (2019).
\newblock \doi{https://doi.org/10.1016/j.compchemeng.2019.106599}

\bibitem{source_code}
Letsios, D., Mistry, M., Misener, R.: Source code (2020).
\newblock \url{https://github.com/dimletsios/}

\bibitem{Letsios2018}
Letsios, D., Mistry, M., Misener, R.: Exact lexicographic scheduling and
  approximate rescheduling.
\newblock European Journal of Operational Research  ((accepted for
  publication))

\bibitem{Liebchen2009}
Liebchen, C., L{\"u}bbecke, M., M{\"o}hring, R., Stiller, S.: The concept of
  recoverable robustness, linear programming recovery, and railway
  applications.
\newblock In: Robust and online large-scale optimization, pp. 1--27. Springer
  (2009)

\bibitem{Mnich2018}
Mnich, M., van Bevern, R.: Parameterized complexity of machine scheduling: 15
  open problems.
\newblock Computers \& Operations Research  (2018)

\bibitem{phillips2002optimal}
Phillips, C.A., Stein, C., Torng, E., Wein, J.: Optimal time-critical
  scheduling via resource augmentation.
\newblock Algorithmica \textbf{32}(2), 163--200 (2002)

\bibitem{Rapine2013}
Rapine, C., Brauner, N.: A polynomial time algorithm for makespan minimization
  on one machine with forbidden start and completion times.
\newblock Discrete Optimization \textbf{10}(4), 241--250 (2013)

\bibitem{Schaffter1997}
Sch{\"a}ffter, M.W.: Scheduling with forbidden sets.
\newblock Discrete Applied Mathematics \textbf{72}(1-2), 155--166 (1997)

\bibitem{Sherali1982}
Sherali, H.D.: Equivalent weights for lexicographic multi-objective programs:
  Characterizations and computations.
\newblock European Journal of Operational Research \textbf{11}(4), 367--379
  (1982)

\bibitem{Skutella2016}
Skutella, M., Verschae, J.: Robust polynomial-time approximation schemes for
  parallel machine scheduling with job arrivals and departures.
\newblock Mathematics of Operations Research \textbf{41}(3), 991--1021 (2016)

\bibitem{Svensson2011}
Svensson, O.: Hardness of precedence constrained scheduling on identical
  machines.
\newblock {SIAM} Journal on Computing \textbf{40}(5), 1258--1274 (2011)

\bibitem{NIPS2006_3053}
Xu, H., Mannor, S.: The robustness-performance tradeoff in {Markov} decision
  processes.
\newblock In: B.~Sch\"{o}lkopf, J.C. Platt, T.~Hoffman (eds.) Advances in
  Neural Information Processing Systems 19, pp. 1537--1544. MIT Press (2007)

\end{thebibliography}
